\documentclass{article}
\usepackage[utf8]{inputenc}
\usepackage{float}
\usepackage{jheppub}
\usepackage{amsmath}
\usepackage[english]{babel}
\usepackage{amssymb}
\usepackage{mathtools}
\usepackage{color}
\definecolor{ao}{rgb}{0.0, 0.5, 0.0}

\usepackage{graphicx}
\usepackage{longtable}
\usepackage{tikz}
\usepackage{comment}
\usepackage{capt-of}
\usepackage{amsthm}
\usepackage{hyperref}
\hypersetup{
    colorlinks=true,
    linkcolor=blue
    }

\usepackage{bm}

\usepackage{caption}
\usepackage{subcaption}

\newtheorem{remark}{Remark}

\newtheorem{proposition}{Proposition}
\newtheorem{corollary}{Corollary}

\newtheorem{definition}{Definition}
\newtheorem{lemma}{Lemma}
\newtheorem{assumption}{Assumption}

\newcommand{\NotAlpha}{\varphi}

\def\0{\mbox{\tiny $0$}}
\def\1{\mbox{\tiny $1$}}
\def\2{\mbox{\tiny $2$}}
\def\3{\mbox{\tiny $3$}}
\def\4{\mbox{\tiny $4$}}
\def\5{\mbox{\tiny $5$}}
\def\6{\mbox{\tiny $6$}}
\def\7{\mbox{\tiny $7$}}
\def\8{\mbox{\tiny $8$}}
\def\9{\mbox{\tiny $9$}}

\def\circaa{\approx}
\def\k{k_{_{B}}}

\def\r{\rangle}

\def\muu{\resub{\tilde{\mu}}}

\usepackage{tikz}
\usetikzlibrary{backgrounds}
\usepackage{listofitems} % for \readlist to create arrays
\usetikzlibrary{arrows.meta} % for arrow size
\usepackage[outline]{contour} % glow around text
\contourlength{1.4pt}

\tikzset{>=latex} % for LaTeX arrow head
\usepackage{xcolor}
\colorlet{myred}{red!80!black}
\colorlet{myblue}{blue!80!black}
\colorlet{mygreen}{green!60!black}
\colorlet{myyellow}{yellow!60!black}
\colorlet{myorange}{orange!85!red!90!black}
\colorlet{mydarkred}{red!30!black}
\colorlet{mydarkblue}{blue!40!black}
\colorlet{mydarkgreen}{green!30!black}
\tikzstyle{node}=[thick,circle,draw=myblue,minimum size=22,inner sep=0.5,outer sep=0.6]
\tikzstyle{node1}=[thick,rectangle,draw=myblue,minimum width = 1cm,  minimum height = 2cm,inner sep=-0.3,outer sep=0.3]
\tikzstyle{node in}=[node,green!20!black,draw=mygreen!30!black,fill=mygreen!5]
\tikzstyle{node hidden}=[node,blue!20!black,draw=myblue!30!black,fill=myblue!20]
\tikzstyle{node convol}=[node,orange!20!black,draw=myorange!30!black,fill=myorange!20]
\tikzstyle{node out}=[node,red!20!black,draw=myred!30!black,fill=myred!20]
\tikzstyle{connect}=[thick,mydarkblue] %,line cap=round
\tikzstyle{connect arrow}=[-{Latex[length=4,width=3.5]},thick,mydarkblue,shorten <=0.5,shorten >=1]
\tikzset{ % node styles, numbered for easy mapping with \nstyle
  node 1/.style={node in},
  node 2/.style={node hidden},
  node 3/.style={node out},
}
\def\l{\langle}
\def\m{\bar{m}}

\def\bbt{\tilde{\beta}}
\def\R{\mathcal{R}}

\def\q{\bar{q}}
\def\n{\bar{n}}

\def\b{\beta^{'}}

\def\etaM{\hat{\eta}_{M}}

\newcommand{\resub}[1]{\textcolor{black}{#1}}%era red
\newcommand{\SOMMA}[2]{\displaystyle\sum\limits_{#1}^{#2}}

\long\def \beq#1\eeq {\begin{equation} #1 \end{equation}}
\long\def \beaq#1\eeaq {\begin{equation}\begin{aligned} #1 \end{aligned}\end{equation}}
\long\def \bes#1\ees {\begin{equation}\begin{split} #1 \end{split} \end{equation}}
\long\def \bea#1\eea {\begin{eqnarray} #1 \end{eqnarray}}
\long\def \bse[#1]#2\ese {\begin{subequations}\label{#1}\begin{align} #2 \end{align}\end{subequations}}

\newcommand{\si}{\sigma_i}

\newcommand{\mb}{\bar{m}}

\title{Dense Hebbian neural networks: \\ a replica symmetric picture of supervised learning}

\author[a]{Elena Agliari\footnote{Given her role as Editor of this journal, EA had no involvement in the peer-review of articles for which she was an author and had no access to information regarding their peer-review. Full responsibility for the peer-review process for this article was delegated to another Editor.},}
\author[b,f,g]{Linda Albanese,}
\author[b,f]{Francesco Alemanno,}
\author[c]{Andrea Alessandrelli,}
\author[b,f]{Adriano Barra,}
\author[d,e]{Fosca Giannotti,}
\author[c,f]{Daniele Lotito,}
\author[c]{Dino Pedreschi.}

\affiliation[a]{Dipartimento di Matematica, Sapienza Universit\`a di Roma, Piazzale Aldo Moro, 5, 00185, Roma, Italy}

\affiliation[b]{Dipartimento di Matematica e Fisica,  Universit\`a  del Salento, Via per Arnesano, 73100, Lecce, Italy}

\affiliation[c]{Dipartimento di Informatica, Università di Pisa, Lungarno Antonio Pacinotti, 43, 56126, Pisa, Italy}

\affiliation[d]{Scuola Normale Superiore, Piazza dei Cavalieri 7, 56126, Pisa, Italy}

\affiliation[e]{Istituto di Scienza e Tecnologie dell' Informazione, Via Giuseppe Moruzzi, 1, 56124 Pisa, Italy}

\affiliation[f]{Istituto Nazionale di Fisica Nucleare, Campus Ecotekne, Via Monteroni, 73100, Lecce, Italy}

\affiliation[g]{Scuola Superiore ISUFI, Campus Ecotekne, Via Monteroni, 73100, Lecce, Italy}

\abstract{We consider dense, associative neural-networks trained by a teacher (i.e., with supervision) and we
investigate their computational capabilities analytically, via statistical-mechanics tools, and numerically,
via Monte Carlo simulations. In particular, we obtain a phase diagram which summarizes their
performance as a function of the control parameters (e.g., quality and quantity of the training dataset,  network storage, noise), that is valid in the limit of large network-size and structureless datasets.
We also numerically test the learning, storing and retrieval capabilities of these networks on structured datasets such as MNist and Fashion MNist.
As technical remarks, on the analytic side, we extend Guerra's interpolation to tackle the non-Gaussian distributions involved in the post-synaptic potentials while, on the computational side, we insert Plefka's approximation in the Monte Carlo scheme, to speed up the evaluation of the synaptic tensors, overall obtaining a novel and broad
approach to investigate supervised learning in neural networks, beyond the shallow limit.
}

\begin{document}

\maketitle

\section{Introduction}
\label{sec:intro}
%\begin{figure}[h]
%    \centering
    %\includegraphics[scale=0.95]{Nuovo_dust_3-10.png}
%    \caption{Examples of pattern reconstruction 
    %at a very low signal-to-noise ratio
%    : we select an example of a pattern, a $0$ from the MNist dataset  (first column), and we corrupt it by flipping randomly \resub{half} of its pixels (second column): this image constitutes the input $\boldsymbol \sigma^{(0)}$ that the networks experience. In particular, three networks are used, $P=2$ (first row) that is the standard Hopfield model, $P=4$ (second row) and $P=6$ (third row). In all the three \resub{rows}, in the third column we report reconstructions $\langle \boldsymbol \sigma^{(1)}\rangle$ (achieved by a single Monte Carlo (MC) step: with a degree of interaction $P=6$ the network manages to perfectly retrieve the pattern, while \resub{less dense} networks fail.}
    %\label{fig:mnistultranoise}
%\end{figure}
%
%
The research covered in this paper, and in its twin \cite{unsup} addressing the unsupervised counterpart, aims to provide an exhaustive picture of (supervised) Hebbian learning by dense networks (namely networks where interactions involve assemblies of $P>2$ units rather than standard couples, i.e. $P=2$), inspecting their emerging computational capabilities by means of statistical mechanics of disordered systems \cite{MPV,Amit,BarraLenka,Haiping}.  Indeed, in the last decades, statistical mechanical has played a pivotal role for describing and quantifying information processing by neural networks for shallow (see e.g. \cite{AGS,EmergencySN,Coolen,Engel,Kanso, mei2022generalization,montanari2022interpolation}), deep (see e.g., \cite{AAAF-JPA2021,alberici2020annealing,Kadmon,metanfetamina,Haiping-Deep}) and dense (see e.g. \cite{Densesterne,DenseRSB,Fachechi1,HopKro1,Krotov2018,AFMJMP2022}) architectures; in particular, investigations on dense spin-glasses \cite{Uffa,Antony1,burioni,Belius,Gerardo,Tommaso,Caracciolo,Subag1} constitute a fertile ground where  Hebbian theories on dense neural networks can germinate. The ultimate reward we obtain by this approach lies in the knowledge of {\em phase diagrams}, namely plots in the space of the network control parameters (e.g., network size and connectivity, storage load, noise, dataset quality and quantity)  where different emerging computational capabilities (e.g., learning, storage,  recognition, associativity, denoising) are effectively related to particular regions of these diagrams, much as like the phase diagram of the water summarizes in a plot with solely three control parameters (i.e., pressure, volume and temperature) the different {\em regimes} (vapor, liquid, solid) in which a network of water molecules can be found. 
For the various modern neural architectures, the knowledge of such diagrams -- where different working regions are split by ``computational phase transitions'' -- can be helpful in the field of Sustainable AI (SAI) as this allows preparing the network in an optimal setting for a given task, with possible  energy saving (e.g., by choosing the minimal architecture, or by avoiding training when a successful learning is theoretically forbidden) \cite{Computa1,Computa2,Computa3}. 
% Further, the mathematical control that results by statistical mechanical treatments of neural networks suggests that such a discipline may constitute a main strand also toward eXplainable AI (XAI).
\newline
Since the first wave of the statistical mechanical formalization in the late eighties and early nineties of the past century, inspecting how these networks can be trained and can retrieve the learnt information has been a central question, addressed from various perspectives (see e.g. \cite{storage1,storage2,storage3,Engel}), yet in the dense network scenario most of the results are limited to retrieval issues \cite{Baldi,Bovier,Elisabeth1}. In that simpler context, the network experiences just once a set of patterns and stores them in its Hebbian synaptic tensor for successive pattern-recognition usage but it does not undergo  a real ``learning process'', where, instead, the network typically has access only to a (noisy) sample of the reality by which it forms its own representation.
%does not have access to the patterns to store rather it has to form its own representation of them just by experiencing their noisy examples provided in form of datasets to analyze: 
Here we deepen this phenomenon moving from Hebbian storing to supervised Hebbian learning (and refer to \cite{unsup} for its unsupervised counterpart).   
\newline
We stress that, while from a biological modeling perspective these dense networks lack a clear inspiration (as neurons interact mainly in couples, although higher-order generalizations can still be seen as effective models \cite{HopfieldKrotovEffective}), in Machine Learning there are no restrictions preventing their usage, and here we prove that density can lead to significant rewards. In fact, dense networks can be used in two different operational modes, both forbidden to shallow machines, (i) a ``ultra-storage'' regime (that extends the high-load regime of the standard Hopfield model to the dense case), where these networks handle a by far larger number of patterns w.r.t. the standard $P = 2$ limit and (ii) a ``ultra-tolerance'' regime (that extends the low-load regime of the standard Hopfield model to the dense case), where these networks perform pattern recognition at prohibitive signal-to-noise ratio w.r.t. the $P=2$ limit. 
As we will quantify along the paper, these capabilities have a cost: the main flaw in the usage of dense networks lays in the large volume of training examples required to achieve a satisfactory learning. However, unlike the unsupervised case \cite{unsup}, in the supervised setting this cost is independent of the order of interactions $P$, thus, if the dataset is sufficiently large to train a dense network with a relatively low interaction order, we can ``freely'' increase the interaction order and take advantage of the related benefits.

The theory we work out here deals with networks learning from random datasets, where the analytical treatment is feasible, and it is successfully tested against Monte Carlo (MC) simulations and further corroborated on structured datasets (MNist and Fashion-MNist), where numerical results return overall a very good qualitative agreement with theoretical predictions. More technically, our analytical investigations allow us to solve for the statistical mechanics of these networks by adapting Guerra's interpolation technique (see e.g. \cite{GuerraNN,glassy,Fachechi1}).  This generalization is non-trivial: as interactions are dense the fluctuations in the post-synaptic potentials are no longer Gaussian and the universality property of the quenched noise in spin glasses \cite{CarmonaWu} can not be applied directly, however, we could prove its validity also for the current case and this required a few passages, implying, for instance, the evaluation of the lower-order momenta of the distributions of the post-synaptic potentials and the application of central limit theorem (CLT). At the numerical level, as the update of the dense synaptic tensor results is a bottleneck for any dynamical update rule, we implement Plefka's effective scheme \cite{Plefka1} on the restricted Boltzmann machine (RBM) equivalent to the supervised dense Hopfield model to speed up this evaluation;  this route provides overall a new approach for tackling the statistical mechanics of neural networks with complex architectures in broad generality. 

The remaining of this paper is structured as follows: in Sec. \ref{sec:super} we introduce the general setting of dense supervised Hebbian networks, whose analytical findings are summarized in Sec. \ref{sec:RS_Guerra_sup}, including the expression of the quenched statistical pressure and the self-consistency equations for the system observables in the high-load (Subsec. \ref{ultrastorage}) and in the low-load (Subsec. \ref{ultratolerance}) settings; the numerical results are reported in Sec. \ref{sec:numerical} along with some experiments on structured datasets (Subsec. \ref{sec:mnist}) and a comparison between unsupervised and supervised regimes (Subsec. \ref{TacciDeThe}); finally, in Sec. \ref{sec:conclusion} we comment results and discuss possible outlooks. Technical details and lengthy proofs are collected in the Appendices.

\section{Dense Hebbian Neural Network in the supervised setting}
\label{sec:super}

In this section we describe dense associative neural networks, namely Hebbian networks with a {\em density} (i.e. the order of interaction $P$ among neurons) $P >2$,  trained in a supervised regime. We start addressing  them via statistical mechanics: the ultimate purpose of this approach is to work out phase diagrams splitting the space of the control parameters (e.g. storage, density, noise and dataset size) into regions where the network shows different emerging computational capabilities. To reach this goal we need to introduce control parameters and order parameters\footnote{We recall that the ``order parameters'' are macroscopic observables which allow describing the emergent proprieties of the network.} suitable for the present case and find an explicit expression for the (quenched) statistical pressure \cite{MPV} in terms of these parameters. Then, by extremizing the	 statistical pressure w.r.t. the order parameters we get a set of self-consistency equations for the evolution of these order parameters in the space of the control parameters, whose investigation ultimately provides the phase diagrams of these networks.
As a technical remark, we stress that we work under the standard replica-symmetry (RS) assumption, namely, we assume that all the order parameters exhibit vanishing fluctuations  around their means in the thermodynamic limit. 

The dense associative neural networks studied in this manuscript are made up of $N$ binary neurons $\boldsymbol \sigma = (\sigma_1, \sigma_2, ..., \sigma_N) \in \{-1, +1\}^N$, interacting in groups of $P$  units  and the classical Hebb's rule is revised in such a way that  it can be meant as a result of a supervised learning. Specifically, it is built over a set of patterns that constitute a perturbed version of some unknown archetypes that the network has to infer, store and possibly retrieve. These archetypes are $K$ binary vectors of length $N$, denoted with $\{ \boldsymbol{\xi}^{\mu}\}^{\mu=1, \hdots , K}$ and their entries are Rademacher random variables  drawn as 
\begin{equation} \label{eq:quenched_xi}
\mathbb P (\xi_i^{\mu}) = \frac{1}{2}\left[\delta_{\xi_i^{\mu},-1} + \delta_{\xi_i^{\mu},+1}\right]
\end{equation}
for $i=1, \hdots, N$ and $\mu=1,\hdots, K$. We introduce $M$ examples for each archetype and we denote them as $\{ \boldsymbol{\eta}^{\mu,a}\}_{a=1,...,M}^{\mu=1,...,K}$ with $\boldsymbol{\eta}^{\mu,a} \in \{ -1, +1\}^N$; these are corrupted versions of the related archetypes, and theirs entries are drawn as 
%\begin{equation}
%\eta_{i}^{\mu, a}= \xi_i^\mu \cdot %\chi_i^{\mu,a},
%\end{equation}
%where $\chi_{i}^{\mu,a}$ is a Bernoullian random variable, namely 
%\begin{equation}\label{eq:quenched_chi}
%\mathbb{P}(\chi_{ i}^{\mu,a}) = \frac{1+r}{2}\delta_{\chi_i^{\mu,a},+1} +  \frac{1-r}{2}\delta_{\chi_i^{\mu,a},-1}, 
%\end{equation}
\begin{equation}\label{eq:quenched_chi}
\mathbb{P}(\eta_{ i}^{\mu,a}|\xi_{ i}^{\mu}) = \frac{1+r}{2}\delta_{\eta_{ i}^{\mu,a},+\xi_{ i}^{\mu}} +  \frac{1-r}{2}\delta_{\eta_{ i}^{\mu,a},-\xi_{ i}^{\mu}}, 
\end{equation} where $r \in (0,1)$ for $i=1, \hdots, N$ and $\mu=1,\hdots, K$. Notice that $r$ plays as a dataset-quality  control  parameter: as $r \to 1$ the examples collapse on the related archetype, while as $r \to 0$ the examples turn out to be uncorrelated with the related archetype.
\newline 
The network is supplied with the dataset $\{ \boldsymbol{\eta}^{\mu,a}\}_{a=1,...,M}^{\mu=1,...,K}$ and it is asked to reconstruct the unknown archetypes $\{ \boldsymbol{\xi}^{\mu}\}^{\mu=1,...,K}$. The available information is allocated in the interaction strength among neurons, that is in the Hebbian synaptic tensors $J_{i_1,...,i_P}^{(sup)}$ (\emph{vide infra} and Eq. (\ref{SynapticTensor})).  We now proceed by defining the following   
\begin{definition}
\label{def:H_sup}
The cost function (or \textit{Hamiltonian}) of the dense Hebbian neural network in the supervised regime is 
\begin{align}
    \mathcal H^{(P)}_{N,K,M,r}(\boldsymbol{\sigma} \vert \bm \eta)=& -\dfrac{N}{ 2\mathcal{R}^{P/2}}\SOMMA{\mu=1}{K}\left(\dfrac{1}{N^P\,M^P}\SOMMA{(i_1,\hdots,i_P)=1}{N,\hdots,N}\SOMMA{a_1,\hdots,a_P=1}{M,\hdots,M}\eta_i^{\mu,a} \cdots \eta_{i_P}^{\mu,a_P} \sigma_{i_1}\cdots\sigma_{i_P}\right) 
    \label{def:H_Psup}
\end{align}
% \begin{align}
%     \mathcal H^{(P)}_{N,K,M,r}(\boldsymbol{\sigma} \vert \bm \xi, \bm \chi)=& -\dfrac{N}{P! \mathcal{R}^{P/2}}\SOMMA{\mu=1}{K}\left(\dfrac{1}{N\,M}\SOMMA{i=1}{N}\SOMMA{a=1}{M}\xi^{\mu}_{i}\chi^{\mu, a}_{i}\sigma_{i}\right)^P 
%     \label{def:H_Psup}
% \end{align}
where $P$ is the interaction order (assumed as even) and we posed $\mathcal{R}:=r^2 + \frac{1-r^2}{M}$.\footnote{The value $\R$ corresponds to the variance of the random variable $\dfrac{1}{M} \SOMMA{a=1}{M} \eta_i^{\mu,a}$. This normalization will be convenient in the calculations presented in Sec.~\ref{sec:RS_Guerra_sup}.}
Moreover, we define  $\sum\limits_{(i_1,\hdots,i_P)}^{N,\hdots,N} =\sum\limits_{\underset{i_1\neq\hdots\neq i_P}{i_1,\hdots,i_P}}^{N,\hdots,N}$ (namely the summation in which only terms with all different "$i$" indices are taken into account)

The related Boltzmann-Gibbs partition function is defined as 
\begin{align}
    \mathcal{Z}^{(P)}_{N,K,M ,r,\beta} (\bm \eta) = \sum_{\bm \sigma} \exp \left( -\beta \mathcal H^{(P)}_{N,K,M, r}(\bm \sigma \vert \bm \eta)\right)=: \sum_{\bm \sigma} \mathcal{B}^{(P)}_{N,K, M,r,\beta}(\bm \sigma  \vert \bm \eta)
    \label{partitionfunction}
    \end{align}
    where $\mathcal{B}^{(P)}_{N,K, M,\beta}(\bm \sigma  \vert \bm \eta)$ is referred to as Boltzmann factor and $\beta \in \mathbb{R}^+$ tunes the broadness of the underlying measure\footnote{If $\beta \to 0$ the Boltzmann-Gibbs measure over $\{ - 1, +1\}^N$ is flat,  while for $\beta \to \infty$ it gets delta-peaked at configurations corresponding to the ground states.}.
    
    The quenched statistical pressure\footnote{We stress that the statistical pressure $\mathcal A$ share the same information content of the free energy  $\mathcal F$ as $\mathcal A = -\beta \mathcal F$: while the usage of $\mathcal F$ is more abundant in the Literature, preferring the former over the latter obviously does not alter the results stemming from their analyses.} of this model reads as 
    \begin{align}
         \mathcal A^{(P)}_{N,K,M,r,\beta} = \frac{1}{N}\mathbb{E}\log  \mathcal{Z}^{(P)}_{N,K,M r,\beta} (\boldsymbol \eta)
        \label{PressureDef_unsup}
    \end{align}
    where $\mathbb{E}=\mathbb{E}_{\xi}\mathbb{E}_{(\eta|\xi)}$ denotes the average over the realization of examples, namely over the distributions \eqref{eq:quenched_xi} and \eqref{eq:quenched_chi}.

    By combining the quenched average $\mathbb{E}[\cdot]$ and the Boltzmann-Gibbs average
    \begin{equation} \label{omegaNKM}
    \omega[(\cdot)]:= \frac{1}{\mathcal{Z}^{(P)}_{N,K,M ,r,\beta} (\bm \eta)} \sum_{\boldsymbol \sigma}^{2^N} ~ (\cdot) ~ \mathcal{B}^{(P)}_{N,K, M,r,\beta}(\bm \sigma  \vert \bm \eta) \, 
    \end{equation}
    -- possibly replicated over two or more replicas\footnote{A replica is an independent copy of the system characterized by the same realization of disorder, namely by the same realization of the archetypes and examples and, thus, of the synaptic tensors.} that is, $\Omega(\cdot):=\omega (\cdot) \times \omega (\cdot) ...\times \omega (\cdot)$ -- we get the expectation
       \begin{equation}
    \langle \cdot \rangle := \mathbb{E} \Omega(\cdot).
     \end{equation}
\end{definition}

We stress that from the Hamiltonian in Eq. \eqref{def:H_Psup}, we get
% \footnotesize
\begin{align}
    \mathcal H^{(P)}_{N,K,M,r}(\boldsymbol{\sigma} \vert \bm \eta)=& - \dfrac{1}{2 N^{P-1}} \SOMMA{(i_1,\cdots,i_P)}{N,\cdots,N} J_{i_1\cdots i_P}^{(sup)} \sigma_{i_1}\cdots\sigma_{i_P} \notag \\
    \nonumber
    =&-\dfrac{1}{2 \mathcal{R}^{P/2}M^P N^{P-1}}\SOMMA{\mu=1}{K}\SOMMA{(i_1,\cdots,i_P)}{N,\cdots,N}\SOMMA{a_1,\cdots,a_P}{M,\cdots,M} \eta_{i_1}^{\mu,a_1}\cdots \eta_{i_P}^{\mu,a_P}\sigma_{i_1}\cdots\sigma_{i_P}  . 
\end{align}
\normalsize
% \footnotesize
% \begin{align}
%     \mathcal H^{(P)}_{N,K,M,r}(\boldsymbol{\sigma} \vert \bm \xi, \bm \chi)= &-\dfrac{1}{P! \mathcal{R}^{P/2}M^P N^{P-1}}\SOMMA{\mu=1}{K}\SOMMA{i_1,\cdots,i_P}{N,\cdots,N}\SOMMA{a_1,\cdots,a_P}{M,\cdots,M}\xi^{\mu}_{i_1}\chi^{\mu, a_1}_{i_1}\cdots\xi^{\mu}_{i_P}\chi^{\mu, a_P}_{i_P}\sigma_{i_1}\cdots\sigma_{i_P}  + \mathcal{O}\left(K N^{1-P/2}\right); 
% \end{align}
% \normalsize
%where we neglect the subleading terms in $N$ as, in the thermodynamic limit \resub{where our theory is developed, these contributions are vanishing}. 
Thus, the entries of the interaction matrix in the supervised regime for $P\geq 2$, are
\begin{align}\label{SynapticTensor}
    J_{i_1\cdots i_P}^{(sup)}= \dfrac{1}{\mathcal{R}^{P/2}M^P }\SOMMA{\mu=1}{K}\SOMMA{a_1,\cdots,a_P}{M,\cdots,M} \eta_{i_1}^{\mu,a_1}\cdots \eta_{i_P}^{\mu,a_P}\,,
\end{align}
and they differ from those introduced in \cite{unsup} for the unsupervised regime which are, instead, 
\begin{align} \label{SynapticTensor_unsup}
    J_{i_1\cdots i_P}^{(unsup)}= \dfrac{1}{\mathcal{R}^{P/2}M }\SOMMA{\mu=1}{K}\SOMMA{a=1}{M} \eta_{i_1}^{\mu,a}\cdots \eta_{i_P}^{\mu,a}\,.
\end{align}
% \begin{align}\label{SynapticTensor}
%     J_{i_1\cdots i_P}^{(sup)}= \dfrac{1}{\mathcal{R}^{P/2}M^P }\SOMMA{\mu=1}{K}\SOMMA{a_1,\cdots,a_P}{M,\cdots,M}\xi^{\mu}_{i_1}\chi^{\mu, a_1}_{i_1}\cdots\xi^{\mu}_{i_P}\chi^{\mu, a_P}_{i_P}\,. 
% \end{align}
By matching \eqref{SynapticTensor} and \eqref{SynapticTensor_unsup}, one can see that the role of the teacher in Hebbian learning -- resembling classical supervision in Machine Learning -- consists in grouping together all the examples pertaining to the same archetype before presenting them to the network. We anticipate that in Sec.~\ref{TacciDeThe} we compare outcomes of these networks trained under the two protocols \cite{unsup}.
 
Differences are pronounced especially when $r$ (and $M$) is relatively small, while, as expected, when $r=1$ supervised and unsupervised settings both converge to the standard Hopfield model, namely to the dense Hebbian storing network, endowed with the same interaction order.

\begin{remark}
\label{dualEff}
%In Sec. \ref{sec:numerical} having an integral representation of the partition function in Eq. \eqref{partitionfunction} will be more convenient. 
We can handle the partition function $\mathcal{Z}^{(P)}_{N,K,M ,r,\beta} (\bm \eta)$ in order to get an integral representation that shall be more convenient when approaching the model numerically (see Sec. \ref{sec:numerical}). Starting from Eq. \eqref{partitionfunction}, we apply the Hubbard-Stratonovich transformation to get, 
%\blue{neglecting the diagonal term}
\begin{equation}
\begin{array}{lll}
     \mathcal{Z}^{(P)}_{N,K, \rho ,\beta} (\bm \eta)
    &=  \SOMMA{\bm \sigma}{} \int \mathcal{D}\bm\tilde{\mu}(\bm z) \exp \Bigg\{\sqrt{\dfrac{\beta'}{N^{P-1} \R^{P/2}}}\SOMMA{\mu=1}{K}&\SOMMA{i_1,\cdots i_{P/2}}{N,\cdots,N}\,\Bigg[\left( \dfrac{1}{M}\SOMMA{a_{1}=1}{M}\eta_{i_1}^{\mu,a_1} \right)\cdot \hdots
    \\\\
    &&\left.\left.\cdot\left(\dfrac{1}{M}\SOMMA{a_{P/2}=1}{M}\eta_{i_{P/2}}^{\mu,a_{P/2}} \right)\right]\sigma_{i_1}\cdots\sigma_{i_{P/2}}z_{\mu} \right\} 
\end{array}
    \label{eq:partition}
\end{equation}
where we posed $\rho:=\dfrac{1-r^2 }{r^2M}$, $\beta':=2\beta/P!$, $ \mathcal D\bm \tilde{\mu}(\bm z)=\prod\limits_{\mu} d\tilde{\mu}(z_{\mu}) $ and  $d\tilde{\mu}(z_{\mu})=\dfrac{dz_\mu}{\sqrt{2\pi}}\exp\left( -\dfrac{z_\mu^2}{2}\right)$ represents the Gaussian measure. Moreover, we have exploited 
\begin{equation}
    \begin{array}{lll}
         &\dfrac{P!}{2 N^{P-1}}\SOMMA{(i_1,\cdots i_{P})}{N,\cdots,N}\,\left( \Phi_{i_1}^{\mu} \hdots \Phi_{i_{P}}^{\mu} \right)=\dfrac{1}{2 N^{P-1}}\SOMMA{i_1,\cdots i_{P}}{N,\cdots,N}\,\left( \Phi_{i_1}^{\mu} \hdots \Phi_{i_{P}}^{\mu} \right)- \mathcal{O}(N^{P/2-1})\,,
    \end{array}
\end{equation}
and we neglect the subleading network-size terms.

We can think  of the above transformation as a mapping between  the original dense Hebbian network and a RBM where $K$ hidden neurons $z$ (equipped with a Gaussian prior) interact with the $N$ visible neurons $\boldsymbol \sigma$ grouped in sets, each made of $P/2$ neurons $\sigma_{i_1}\cdots\sigma_{i_{P/2}}$ with weight considered as $\left(M^{-1}\sum_{a_1=1}^M \eta_{i_1}^{\mu,a_1}\right) \cdots \left(M^{-1}\sum_{a_{P/2}=1}^M \eta_{i_{P/2}}^{\mu,a_{P/2}}\right)$. 
%\footnote{This equivalence holds also for the pairwise Hopfield model \cite{Contucci} and higher order-Boltzmann machines \cite{Senio1,BarraPRLdetective}.}. 
In this integral representation of the partition function there is a one-to-one correspondence  between  the number of archetypes and the amount of hidden neurons, suggesting to allocate one hidden neuron to the recognition of a specific archetype (this setting is historically known as {\em grandmother cell}  \cite{NonnaCell,prlmiriam,AM-2020}). A schematic representation of this mapping is presented for a simple case in Fig. \ref{fig:network}.
\end{remark}

%%%%QUI RETE
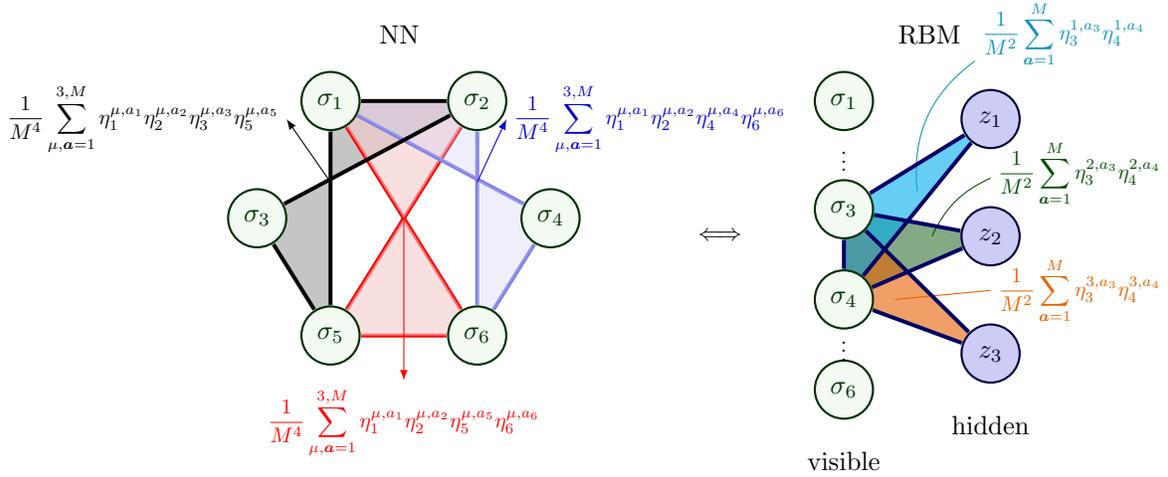
\begin{figure}[t]
\centering
\begin{tikzpicture}[x=1.95cm,y=1.2cm]

\node[node in,outer sep=0.6] (S-1) at (-3.5,0.2) {$\sigma_{1}$};
\node[node in,outer sep=0.6] (S-2) at (-2.5,0.2) {$\sigma_{2}$};
\node[node in,outer sep=0.6] (S-3) at (-4.0,-1.1) {$\sigma_{3}$};
\node[node in,outer sep=0.6] (S-4) at (-2.0,-1.1) {$\sigma_{4}$};
\node[node in,outer sep=0.6] (S-5) at (-3.5,-2.4) {$\sigma_{5}$};
\node[node in,outer sep=0.6] (S-6) at (-2.5,-2.4) {$\sigma_{6}$};

\draw[line width=1.5pt] (S-1) -- (S-5);
\draw[line width=1.5pt] (S-1) -- (S-2);
\draw[line width=1.5pt] (S-2) -- (S-3);
\draw[line width=1.5pt] (S-3) -- (S-5);

\begin{scope}[on background layer]
\path [fill=lightgray,opacity=0.9,very thin] (S-1.center) to  (S-5.center) 
    to  (S-3.center) to (S-2.center);
    \begin{scope}[on background layer]
    \draw[line width=1.5pt, red] (S-1) -- (S-6);
    \draw[line width=1.5pt, red] (S-2) -- (S-5);
    \draw[line width=1.5pt, red] (S-5) -- (S-6);
        \path [fill=myred!20,opacity=0.6,very thin] (S-1.center) to  (S-2.center) 
    to  (S-5.center) to (S-6.center);
    \begin{scope}[on background layer]
    \draw[line width=1.5pt, myblue!50] (S-6) -- (S-2);
    \draw[line width=1.5pt, myblue!50] (S-4) -- (S-1);
    \draw[line width=1.5pt, myblue!50] (S-4) -- (S-6);
        \path [fill=myblue!20,opacity=0.3,very thin] (S-1.center) to  (S-2.center)  to  (S-6.center) to (S-4.center);
    \end{scope}
    \end{scope}
\end{scope}

 \draw[<-]        (-3.8,-0.0)node[left, scale = 0.8] {$\dfrac{1}{M^4}\SOMMA{\mu,\bm a=1}{3,M}\eta_{1}^{\mu,a_1}\eta_{2}^{\mu,a_2}\eta_{3}^{\mu,a_3}\eta_{5}^{\mu,a_5}$}   -- (-3.5,-0.7) ;

 \draw[->, myblue]        (-2.5,-0.7)   -- (-2.3,0.0) node[right, scale = 0.8]{$\dfrac{1}{M^4}\SOMMA{\mu,\bm a=1}{3,M}\eta_{1}^{\mu,a_1}\eta_{2}^{\mu,a_2}\eta_{4}^{\mu,a_4}\eta_{6}^{\mu,a_6}$};

 \draw[->, red]        (-3.0,-1.1)   -- (-3.0,-2.9) node[below, scale = 0.8] {$\dfrac{1}{M^4}\SOMMA{\mu,\bm a=1}{3,M}\eta_{1}^{\mu,a_1}\eta_{2}^{\mu,a_2}\eta_{5}^{\mu,a_5}\eta_{6}^{\mu,a_6}$};

  % INPUT LAYER
\node[node in,outer sep=0.6] (NI-1) at (0,0.2) {$\sigma_{1}$};
\node[node in,outer sep=0.6] (NI-3) at (0,-1) {$\sigma_{3}$};
\node[node in,outer sep=0.6] (NI-4) at (0,-2) {$\sigma_{4}$};
\node[node in,outer sep=0.6] (NI-5) at (0,-3) {$\sigma_{6}$};
    
  % OUTPUT LAYER
\node[node hidden] (NO-1) at (1,-0.8+0.8) {$z_{1}$};
\node[node hidden] (NO-2) at (1,-0.8-0.5) {$z_{2}$};
\node[node hidden] (NO-3) at (1,-0.8-1.8) {$z_{3}$};
  
  % DOTS
  \draw[connect, line width=1.5pt] (NI-3) -- (NO-1);
  \draw[connect, line width=1.5pt] (NI-4) -- (NO-1);
  \draw[connect, line width=1.5pt] (NI-3) -- (NI-4);
  \draw[connect, line width=1.5pt] (NI-3) -- (NO-3);
  \draw[connect, line width=1.5pt] (NI-4) -- (NO-3);
  \draw[connect, line width=1.5pt] (NI-3) -- (NO-2);
  \draw[connect, line width=1.5pt] (NI-4) -- (NO-2);
  \path (NI-4) --++ (NI-5) node[midway,scale=0.8] {$\vdots$};
  \path (NI-1) --++ (NI-3) node[midway,scale=0.8] {$\vdots$};

\node[above=25, right=17,align=center] at (NI-1) {RBM};
\node[below=20, align=center] at (NI-5) {visible};
\node[below=20,align=center] at (NO-3) {hidden};
\node[above=25, right=15,align=center] at (S-1) {NN};

\begin{scope}[on background layer]
    \path [fill=mydarkgreen!80,opacity=0.6,very thin] (NI-3.center) to  (NI-4.center) 
        to  (NO-2.center);
    \draw[-,cyan!70!black] (1.5,0.5) node[above, scale = 0.8]{$\dfrac{1}{M^2}\SOMMA{\bm a=1}{M}\eta_{3}^{1,a_3}\eta_{4}^{1,a_4}$}   edge[bend right=30,right]   (0.5,-0.8);
    %\draw[->, red!80!black]        (0.5,-0.8)    -- (1.2,0.9) node[right, scale = 0.8] {\large$\SOMMA{a=1}{M}\eta_3^{2,a}\eta_4^{2,a}$};

    \path [fill=myorange,opacity=0.6,very thin] (NI-3.center) to  (NI-4.center) 
        to  (NO-3.center);
    \draw[-, myorange]        (0.35,-2.0)   -- (1.0,-1.9) node[right, scale = 0.8] {$\dfrac{1}{M^2}\SOMMA{\bm a=1}{M}\eta_{3}^{3,a_3}\eta_{4}^{3,a_4}$};

    \path [fill=cyan,opacity=0.6,very thin] (NI-3.center) to  (NI-4.center) 
        to  (NO-1.center);
    \draw[-, mydarkgreen]          (1.0,-0.6)  node[right, scale = 0.8] {$\dfrac{1}{M^2}\SOMMA{\bm a=1}{M}\eta_{3}^{2,a_3}\eta_{4}^{2,a_4}$}  edge[bend right=20,right]  (0.6,-1.3);
\end{scope}
  \node[right,scale=0.9] at (-1.05,-1.3)
    {$\Longleftrightarrow$};
\end{tikzpicture}
\caption{Representation of the dense supervised neural network (NN, left) and its dual RBM (right), setting $N=6$, $K=3$ and $P=4$. The areas filled in different colours correspond to different interactions.
The RBM is built with an hidden layer made of $K$ Gaussian neurons $\{z_\mu\}_{\mu=1,2,3}$ and a visible one made of $N$ binary variables $\{\sigma_i\}_{i=1,\hdots,6}$. In particular, any $z_\mu$ can interact simultaneously with a set of $P/2$ (namely, $2$) visible spins $\{\sigma_i, \sigma_j\}$ and the corresponding coupling is $\left(\frac{1}{M}\sum\limits_{a_1=1}^{M}\eta_i^{\mu,a_1}\right)\left(\frac{1}{M}\sum\limits_{a_2=1}^{M}\eta_j^{\mu,a_2}\right)$. In the NN, the neurons interact $4-$wise and the
synaptic weight for any set of variables $\{\sigma_i, \sigma_j,\sigma_k,\sigma_l\}$ is $\sum\limits_{\mu=1}^{K}\left(\frac{1}{M}\sum\limits_{a_1=1}^{M}\eta_i^{\mu,a_1}\right)\left(\frac{1}{M}\sum\limits_{a_2=1}^{M}\eta_j^{\mu,a_2}\right)\left(\frac{1}{M}\sum\limits_{a_3=1}^{M}\eta_k^{\mu,a_3}\right)\left(\frac{1}{M}\sum\limits_{a_4=1}^{M}\eta_l^{\mu,a_4}\right)$.}
\label{fig:network}
\end{figure}

In the thermodynamic limit $N \to \infty$ we also introduce 
\begin{definition}
    The network load is defined as
\begin{equation}
\lim_{N\to +\infty} \frac{K}{N^{b}} =: \alpha_{b} < \infty
\label{eq:carico_TDL}
\end{equation}
with $b \leq P-1$.\footnote{The case $b>P-1$ is known to lead to a black-out scenario \cite{Baldi, Bovier} not useful for computational purposes and shall be neglected here.  For more detailed considerations on the upper bound of the parameter $b$ we refer to \cite{Albanese2021}.}
Further, the statistical pressure in the thermodynamic limit is denoted as
\begin{equation}
\mathcal A^{(P)}_{\alpha_{b} , M,r,\beta} = \lim_{N \to \infty} \mathcal A^{(P)}_{N,K,M,r,\beta}.
\end{equation}
\end{definition}

According to its value, the parameter $b$ gives rise to two different operational regimes, as mentioned in Sec. \ref{sec:intro}. The case $b = P - 1$ is called \emph{high-load regime} and, as we will deepen in Sec. \ref{ultrastorage}, in this setting the synaptic tensor allows allocating a great amount of information that is sufficient to infer, store  and retrieve up to $K \circaa N^{P-1}$ archetypes, where with the symbol " $\circaa$ ", we mean \textit{at the order of magnitude of}. Conversely, when $b < P - 1$, we have the \emph{low-load regime}, deepened in Sec. \ref{ultratolerance}, where the amount of resources is oversized compared to the amount of patterns to be handled so that we can employ the resulting redundancy to cope with the presence of additional noise affecting examples on tensor entries. 

It is also convenient to introduce the $P$-independent load $\gamma$ defined by
\begin{equation}
\label{eq:load_gamma}
    \alpha_{P-1}:=\gamma\dfrac{2}{P!} \, .
\end{equation}
We observe that, as long as $P$ is fixed, the assumption $\alpha_{P-1} < \infty$  also means that $\gamma < \infty$.

Finally, we introduce a set of observables, that play as \textit{order parameters} of the theory,   
\begin{definition}
The order parameters of the dense supervised Hebbian network introduced in Def. \ref{def:H_sup} are
    \begin{equation}
    \begin{array}{lll}
    \label{eq:orderparameters}
         n_\mu&=\dfrac{r}{\mathcal{R}}\dfrac{1}{N M}\SOMMA{i,a=1}{N,M}\eta_i^{\mu,a}\sigma_i,\\ 
         m_\mu&=\dfrac{1}{N}\SOMMA{i=1}{N}\xi_i^{\mu}\sigma_i, \\
         q_{cd}&=\dfrac{1}{N}\SOMMA{i=1}{N}\sigma_{i}^{(c)}\sigma_{i}^{(d)},
    \end{array}
\end{equation}
where $\mu=1, \hdots, K$ and $c$ and $d$ label two different replicas. 
\end{definition}
We highlight that the Mattis magnetization $m_{\mu}$ 
quantifies the alignment of the network configuration $\boldsymbol \sigma$ with the archetype $\boldsymbol \xi^{\mu}$, $n_{\mu}$ compares the alignment of the network configuration with the average of all the examples labeled by $\mu$ (i.e. pertaining to the $\mu$-th archetype)  and $q_{cd}$ is the standard two-replica overlap between the replicas $\boldsymbol \sigma^{(c)}$ and $\boldsymbol \sigma^{(d)}$.

\section{Analytical findings}
\label{sec:RS_Guerra_sup}

In this section we solve the statistical mechanics of the dense supervised Hebbian network, namely we obtain the explicit expression of its quenched statistical pressure in terms of control and order parameters in the thermodynamic limit and under the RS assumption. Next, we extremize the statistical pressure to obtain a set of self-consistency equations for the order parameters, whose solution allows us to draw the phase diagram of these networks. 
To do so, we exploit Guerra's interpolation technique \cite{guerra_broken}, yet some adjustments are in order. In fact, in these dense networks the internal fields (also referred to as post-synaptic potentials) acting on neurons are non-Gaussian variables and this prevents the (direct) application of Guerra's interpolation technique. However, we can show that these random variables occur in calculations only through suitable averages which can be replaced by a single normal variable via the CLT and theorems on the universality of the quenched noise  \cite{CarmonaWu,Genovese}.  Finally at the end of Sec. \ref{subsec:S2N_sup}, we corroborate this procedure by means of MC simulations (Fig. \ref{fig:s2n_unsup}).

%\resub{we stress that each example making up the available dataset is a noisy version of its corresponding pattern, namely $\eta_i^{\mu,a}=\xi_i^\mu \chi_i^{\mu,a}$ for $\mu=1, \hdots, K$, $a=1, \hdots, M$, $i=1, \hdots, N$, therefore we can exploit this relation and work with the Hamiltonian dependent on $\bm \xi$ and $\bm \eta$. Moreover}, 
Further, we recall that, as standard (see e.g. \cite{Coolen}) and with no loss of generality, in the following we will focus on the ability of the network to learn and retrieve the first archetype $\boldsymbol \xi^1$. Thus, in the next equation, the contribution corresponding to $\mu=1$ shall be split from all the others and interpreted as the {\em signal} contribution, while the remaining ones (i.e. those with $\mu \neq 1$) make up the {\em slow-noise} contribution impairing both learning and retrieval of $\boldsymbol \xi^1$, moreover we apply the functional generator\footnote{We add the term $J \sum_i \xi_i^1 \si$ in the exponent of the Boltzmann factor to ``generate'' the expectation of the Mattis magnetization $m_1$. To do so, we evaluate the derivative w.r.t. $J$ of the quenched statistical pressure and then set $J=0$.}:
\begin{align}
\label{eq:partizione_originale}
    \mathcal{Z}^{(P)}_{N,K, \rho ,\beta} (\bm \eta)
    &= \sum_{\bm \sigma}  \exp \left[ J \SOMMA{i=1}{N} \xi_i^1 \si +\beta'\dfrac{N}{2}(1+\rho)^{P/2}n_1^{P}(\boldsymbol{\sigma}) \right. 
    \\
    &\left.+\dfrac{\b P!}{2N^{P-1}\mathcal{R}^{P/2} }\SOMMA{\mu>1}{K}\SOMMA{(i_1,\cdots i_{P})}{N,\cdots,N}\,\left( \dfrac{1}{M^{P}}\SOMMA{a_1,\dots,a_{P}}{M,\cdots,M} \eta_{i_1}^{\mu,a_1} \hdots \eta_{i_{P}}^{\mu,a_{P}}\right)\sigma_{i_1}\cdots\sigma_{i_{P}} \right]. \notag 
\end{align}

Applying the CLT on the variable $\dfrac{1}{M}\sum_a^M  \eta_i^{\mu,a}$, we can replace each $i$-th contribution in the noise term in the \emph{round} brackets in \eqref{eq:partition} with $\sqrt{\mathcal{R}}\phi_{i}^\mu$ with $\phi_i^\mu\sim \mathcal{N}(0,1)$, thus we get
\begin{align}
\label{eq:partitionfunction2}
    \mathcal{Z}^{(P)}_{N,K, \rho ,\beta} (\bm \eta)
    &=  \sum_{\bm \sigma}  \exp \left[ J \SOMMA{i=1}{N} \xi_i^1 \si +\beta'\dfrac{N}{2}(1+\rho)^{P/2}n_1^{P}(\boldsymbol{\sigma}) \right. 
    \\
    &\left.+\dfrac{\b P!}{2N^{P-1} }\SOMMA{\mu>1}{K}\SOMMA{(i_1,\cdots i_{P})}{N,\cdots,N}\,\left(\phi_{i_1}^{\mu} \hdots \phi_{i_{P}}^{\mu}\right)\sigma_{i_1}\cdots\sigma_{i_{P}} \right]. \notag 
\end{align}

\begin{remark}
\label{rem:entrpoy_rho}
As discussed in \cite{prlmiriam} for the case of pair-wise interactions, it can be proved that the conditional entropy $H(\boldsymbol{\xi}^{\mu}|\{\boldsymbol{\eta}^{\mu,a}\}^{a=1,\hdots,M})$ for each $\mu =1, \hdots , K$, that quantifies the amount of information needed to describe the original message $\boldsymbol{\xi}^\mu$ given the $M$ examples $\{\boldsymbol{\eta}^{\mu,a}\}^{a=1,\hdots,M}$, is a monotonically increasing function of $\rho$ and, in the following, with a slight abuse of language, we will refer to $\rho$ as the \textit{dataset entropy}.
\end{remark}

% Applying the CLT on the variable $\dfrac{1}{M}\sum_a^M \resub{\eta_i^{\mu,a}}$, we can replace each $i$-th contribution in the noise term in the \emph{round} brackets in \eqref{eq:partition} with $\sqrt{\mathcal{R}}\phi_{i}^\mu$ with $\phi_i^\mu\sim \mathcal{N}(0,1)$, thus we get
% \begin{align}
%    \mathcal{Z}^{(P)}_{N,K, \rho ,\beta} (\bm \eta)
%    &= \lim_{J \rightarrow 0} \mathcal{Z}^{(P)}_{N,K,\resub{\rho},\beta}( \bm \eta ;J) \notag \\
%    &= \lim_{J \rightarrow 0} \sum_{\bm \sigma} \int \resub{\mathcal{D}\bm\muu(\bm z)} \exp \left[ J \SOMMA{i=1}{N} \xi_i^1 \si +\beta'\dfrac{N}{2}(1+\rho)^{P/2}n_1^{P}(\boldsymbol{\sigma}) \right. \label{eq:partition_CLT}\\
%    &\left.+\sqrt{\green{\dfrac{P!(P/2)!}{2}}\dfrac{\beta'}{N^{P-1} }}\SOMMA{\mu>1}{K}\SOMMA{(i_1,\cdots i_{P/2})}{N,\cdots,N}\,\left(\phi_{i_1}^\mu \hdots \phi_{i_{P/2}}^\mu \right)\sigma_{i_1}\cdots\sigma_{i_{P/2}}z_{\mu} \right]. \notag
% \end{align}

If we focus on the noise term of Eq. \eqref{eq:partitionfunction2} and, in particular, on the random variable in the round brackets in \eqref{eq:partitionfunction2}, we can notice that it is given by a product of Gaussians. Using the theorem of the universality of the quenched noise stated in \cite{CarmonaWu,Genovese}, we can solve equivalently the model with a noise distribution which is a single normal random variable $\lambda_{i_1, \hdots, i_{P}}^\mu \sim \mathcal{N}(0,1)$. 
Therefore, we have reached an effective partition function for supervised dense Hebbian neural networks that reads as 
\begin{align}
    \mathcal{Z}^{(P)}_{N,K,\rho,\beta} (\bm \eta) &= \lim_{J \rightarrow 0} \mathcal{Z}^{(P)}_{N,K, \rho,\beta}( \bm \eta ;J) \notag \\
    &= \lim_{J \rightarrow 0} \sum_{\bm \sigma} \exp \left[ J \SOMMA{i=1}{N} \xi_i^1 \si +\beta'\dfrac{N}{2}(1+\rho)^{P/2} n_1^{P}(\boldsymbol{\sigma})\notag \right.\\
    &\left.+\dfrac{P!}{2}\dfrac{\beta'}{N^{P-1} }\SOMMA{\mu>1}{K}\SOMMA{(i_1,\cdots , i_{P})}{N,\cdots,N}\lambda^{\mu}_{i_1,\hdots,i_P}\sigma_{i_1}\cdots\sigma_{i_{P}} \right]
    \label{eq:partition_supervised}
\end{align}
where $\lambda^{\mu}_{i_1,\hdots,i_P}$ are the aforementioned Gaussian i.i.d. variables. 
% \resub{We anticipate that $\V= \sqrt{P(2P-3)!!/2}$ (see Eq. (\ref{eq:VV})).}
% \begin{equation} \label{eq:V_anticipo}
%    \V= \sqrt{\dfrac{P(2P-3)!!}{2}} \,.
% \end{equation} 
%\end{remark}
%

\begin{remark}
 In the case of a quality $r_{\mu}$ that is archetype-dependent and of imbalanced datasets, namely considering a different amount of items for each class $$\bm \eta= \{ \{ \eta_i^{1,a_1}\}^{a_1=1, \hdots, M_1}, \hdots ,\{\eta_i^{K,a_K}\}^{a_K=1, \hdots, M_K}\}_{i=1, \hdots, N}$$ 
$$\mathbb{P}(\eta_{ i}^{\mu,a_\mu}|\xi_i^\mu) = \frac{1+r_\mu}{2}\delta_{\eta_i^{\mu,a_\mu},+\xi_i^{\mu}} +  \frac{1-r_\mu}{2}\delta_{\eta_i^{\mu,a_\mu},-\xi_i^{\mu}},$$ 
one can introduce different ``entropies'' for each class, defined as $\rho_\mu:=\frac{1-r_\mu^2}{r_\mu^2 M_\mu}$, and, by generalizing the Hamiltonian introduced in Def. \ref{def:H_sup}, it is possible to get 
\begin{align}
    \mathcal H^{(P)}_{N,K, \bm M,\bm r}(\boldsymbol{\sigma} \vert \bm \eta)&= -N\SOMMA{\mu=1}{K}\left(\dfrac{1}{N^P\mathcal{R}_\mu^P  M_\mu^P}\SOMMA{(i_1,\cdots,i_P)}{N,\cdots,N}\SOMMA{a_1,\cdots,a_P}{M_\mu,\cdots,M_\mu}\eta^{\mu, a_1}_{i_1}\cdots\eta^{\mu, a_P}_{i_P}\sigma_{i_1}\cdots\sigma_{i_P}\right),
\end{align}
where $\mathcal{R}_\mu= r_\mu^2 + \frac{1-r_\mu^2}{ M_\mu}$.
As long as the sizes $M_{\mu}$'s are all finite, mirroring Eq. \eqref{eq:partitionfunction2}, we apply the CLT on $(M_\mu)^{-1}\sum\limits_{a=1}^M \eta_i^{\mu,a}$ and we replace it as $\phi^\mu_i\sqrt{\mathcal{R}_\mu}$, where $\phi^\mu_i \sim \mathcal{N}(0,1)$. Therefore the partition function can be written as 
\begin{align}
    \mathcal{Z}^{(P)}_{N,K, \bm \rho ,\beta} (\bm \eta)
    & = \lim_{J \rightarrow 0} \sum_{\bm \sigma}  \exp \left[ J \SOMMA{i=1}{N} \xi_i^1 \si +\beta'\dfrac{N}{2}(1+\rho_1)^{P/2}n_1^{P}(\boldsymbol{\sigma})\notag \right. \\
    &\left.+\dfrac{P!}{2}\dfrac{\beta'}{N^{P-1} }\SOMMA{\mu>1}{K}\SOMMA{(i_1,\cdots i_{P})}{N,\cdots,N}\left(\phi_{i_1}^\mu \hdots \phi_{i_{P}}^\mu \right)\sigma_{i_1}\cdots\sigma_{i_{P}} \right].
    \label{eq:partition_CLT_imb}
\end{align}
This expression is equivalent to the one in Eq. \eqref{eq:partitionfunction2}, a shift $\rho$ into $\rho_1$ apart. 
\end{remark}

Now we want to find the self-consistency relations for the order parameters introduced in Eqs. \eqref{eq:orderparameters} by Guerra's interpolation. The underlying idea behind the original technique is to introduce a generalized statistical pressure that interpolates between the original model (which is the target of our investigation but we are not able to address directly) and a simple one (which is usually a one-body model that we can solve exactly). We thus find the solution of the latter and then propagate the obtained solution back to the original model by the fundamental theorem of calculus. In this last passage we assume RS: this property makes the integral appearing in the propagating term analytical. Let us start with some definitions.

\begin{definition}
Given the interpolating parameter $t \in [0,1]$, the real-valued constants $A,\ \psi$ (whose specific values are to be set {\em a posteriori}) and the i.i.d. standard Gaussian variables $Y_i \sim \mathcal{N}(0,1)$ for $i=1, \hdots , N$,  the  interpolating partition function is given as 

\begin{equation}
\begin{array}{lll}
     \mathcal{Z}^{(P)}_{N,K,\rho,\beta}( \bm \eta ;J, t)=&\SOMMA{\lbrace\boldsymbol{\sigma}\rbrace}{} \displaystyle\int\, \mathcal{D}\bm\muu(\bm z) \exp{\Bigg[J \SOMMA{i=1}{N} \xi_i^1 \si+t\beta'\dfrac{N}{2}(1+\rho)^{P/2}n_1^{P}(\boldsymbol{\sigma})}+(1-t)\dfrac{N}{2}\psi \, n_1(\boldsymbol{\sigma})
     \\\\
      &+\sqrt{t} \, \dfrac{P!}{2}\dfrac{\beta'}{N^{P-1}}\SOMMA{\mu>1}{K}\SOMMA{(i_1, \hdots , i_{P})}{N,\dots,N}\lambda_{i_1, \hdots , i_{P}}^{\mu}\sigma_{i_1}\cdots\sigma_{i_{P}}+\sqrt{1-t}A\SOMMA{i=1}{N}Y_i\sigma_i \Bigg].
\end{array}
\label{def:partfunct_GuerraRS}
\end{equation}
\normalsize
% where {$ \mathcal D\bm \muu(\bm z)=\prod\limits_{\mu} d\muu(z_{\mu}) $} and  $d\muu(z_{\mu})=\dfrac{dz_\mu}{\sqrt{2\pi}}\exp\left( -\dfrac{z_\mu^2}{2}\right)$ represents the Gaussian measure.
%
%
    When $t=1$ we recover the original model \eqref{partitionfunction} and for $t=0$ we recover a model simple enough to be solved directly. 
\end{definition}

\begin{definition} The interpolating  statistical pressure related to the partition function (\ref{def:partfunct_GuerraRS}) is introduced as
\begin{eqnarray}
\mathcal{A}^{(P)}_{N,K,\rho,\beta}(t) &\coloneqq& \frac{1}{N} \mathbb{E} \left[  \ln\mathcal{Z}^{(P)}_{N, K, \rho, \beta}(\bm \eta, t)  \right],
\label{hop_GuerraAction}
\end{eqnarray}
where the expectation $\mathbb E$ is now also meant over $\lambda^\mu_{i_1, \hdots , i_{P}}$ and $\lbrace Y_i\rbrace_{i=1,\hdots,N}$.
\newline
In the thermodynamic limit,
\begin{equation}
\mathcal{A}_{\alpha_{b}, \rho, \beta}^{(P)}(t) \coloneqq \lim_{N \to \infty} \mathcal{A}^{(P)}_{N, K,\rho, \beta }(t).
\label{hop_GuerraAction_TDL}
\end{equation}
By setting $t=1$ the interpolating pressure recovers the original one (\ref{PressureDef_unsup}), that is $\mathcal A_{\alpha_{b},\rho, \beta }^{(P)} = \mathcal{A}^{(P)}_{\alpha_{b},\rho, \beta}(t=1)$. 
\newline
	The interpolation also implies a generalized measure, whose related Boltzmann factor reads as
	\begin{equation}
	\begin{array}{lll}
	\mathcal B_{N,K, \rho, \beta}^{(P)} 
 (\boldsymbol{\sigma} |\bm \eta; J, t )&\coloneqq&  
	\exp{\Bigg[J \sum_i \xi_i^1 \si+t\beta'\dfrac{N}{2}(1+\rho)^{P/2}n_1^{P}(\boldsymbol{\sigma})}+(1-t)\dfrac{N}{2}\psi n_1(\boldsymbol{\sigma})
 \\\\
      &&+\sqrt{t} \,  \dfrac{P!}{2}\dfrac{\beta'}{N^{P-1}}\SOMMA{\mu>1}{K}\SOMMA{(i_1, \hdots , i_{P})}{N,\dots,N}\lambda_{i_1, \hdots , i_{P}}^{\mu}\sigma_{i_1}\cdots\sigma_{i_{P}} +\sqrt{1-t}\,A\SOMMA{i=1}{N}Y_i\sigma_i \Bigg],
	\end{array}
	\end{equation}
by which the partition function can be written as $$\mathcal Z^{(P)}_{N, K, \rho, \beta}( \bm \eta; J, t ) =  \sum_{\boldsymbol \sigma} \int  \mathcal{D}\bm\muu(\bm z)   \mathcal B_{N,K, \rho, \beta}^{(P)} 
 (\boldsymbol{\sigma} |\bm \eta; J, t ).$$ \\
A generalized average follows from this generalized measure as
\beq
	\omega_{t}  [(\cdot)]  \coloneqq  \frac{1}{\mathcal Z^{(P)}_{N, K, \rho, \beta }( \bm \eta;J, t )} \sum_{\boldsymbol \sigma}\displaystyle\int \mathcal{D}\tilde{\bm\mu}(\bm z)~ (\cdot) ~   \mathcal B_{N,K, \rho, \beta }^{(P)} 
 (\boldsymbol{\sigma} |\bm \eta;J,  t )
	\eeq
	and
\beq
\langle \cdot   \rangle_{t}  \coloneqq \mathbb E [ \omega_{t} ( \cdot) ].
\eeq
%where $ \mathbb E$ denotes the average over $\lambda_{i_1, \hdots , i_{P/2}}$ and $\lbrace Y_i\rbrace_{i=1,\cdots,N}$, $\lbrace J_{\mu,a}\rbrace_{\mu=1,\hdots,K, a=1, \hdots, M}$.
Of course, when $t=1$ the standard Boltzmann-Gibbs measure and related averages are recovered.
\end{definition}

As anticipated, the following analytical results are obtained under the RS assumption, namely assuming 
the self-averaging property for any order parameter $X$, i.e. the ﬂuctuations around their expectation values vanish in the thermodynamic limit. In distributional sense, 
this corresponds to assuming
\begin{equation} \label{eq:RS}
\lim_{N \to \infty}  \mathbb P_{RS}(X) = \delta(X - \bar X)
\end{equation}
where $\bar X = \langle  X \rangle_t$ is the expectation value w.r.t. the interpolating measure. Hereafter, in order to lighten the notation, we  drop the subscript $t$ in the average symbol. 

\par\medskip
En route toward the self-consistency equations for the order parameters, we need to take the thermodynamic limit of the statistical pressure and, to this task, we must split the discussion into two cases: the high-load regime $b=P-1$, discussed in the following subsection and the low-load regime $b<P-1$, tackled in the subsequent one.

\subsection{High-load regime}\label{ultrastorage}
In this subsection we study the network at work with the maximum storage of archetypes \{$\bm\xi^{\mu}\}^{\mu=1,\hdots,K}$ it can afford. %In this \resub{high-load} regime the signal-to-noise to be preserved is the standard one, i.e.  the signal can not be lower than  the noise.
\begin{proposition}
\label{prop:highnoise}
\textit{(High-load regime)} In the thermodynamic limit ($N\to\infty$) and under the RS assumption, the quenched statistical pressure for the dense, supervised Hebbian network for even $P \geq 4$ in high-storage regime $b=P-1$ reads as 
%\footnotesize
\begin{equation}
\begin{array}{lll}
         \mathcal{A}_{\gamma,\rho  , \beta}^{(P)} (J)
     &=&\mathbb{E}\left\lbrace{\ln\Bigg[2\cosh\left(J\xi^1 +\beta'\dfrac{P}{2}(1+\rho)^{P/2-1}\n^{P-1}\tilde{\eta}+Y\beta'\sqrt{   \gamma
         \dfrac{P}{2}\q^{P-1}}\right)}\Bigg]\right\rbrace+ 
      \\\\
      &&-\dfrac{\beta'}{2}(1+\rho)^{P/2}(P-1)\n^{P}+    \gamma
         \dfrac{\b\,^2}{4}\left(1-P\q^{P-1}+(P-1)\q^P\right),
\end{array}
\end{equation}
\normalsize
with $\mathbb{E}=\mathbb{E}_{\xi}\mathbb{E}_{(\eta|\xi)}\mathbb{E}_Y$, $\etaM=\dfrac{1}{rM}\SOMMA{a=1}{M}\eta^{1,a}$, and $\bar n$ and $\q$ fulfill the following self-consistency equations
\begin{equation}
    \begin{array}{lll}
        \n=\dfrac{1}{(1+\rho)}\mathbb{E}\left\{\tanh{\left[\beta'\dfrac{P}{2}(1+\rho)^{P/2-1}\n^{P-1}\tilde{\eta}+Y\beta'\sqrt{   \gamma\dfrac{P}{2}\q^{P-1}}\right]}\tilde{\eta}\right\},
        
        \\\\
        \q=\mathbb{E}\left\{\tanh{}^2{\left[\beta'\dfrac{P}{2}(1+\rho)^{P/2-1}\n^{P-1}\tilde{\eta}+Y\beta'\sqrt{    \gamma
         \dfrac{P}{2}\q^{P-1}}\right]}\right\}.
        %\p=\b\V\q^{P/2}
        \label{eq:n_selfRS_sup} 
    \end{array}
\end{equation}
\normalsize

Furthermore, considering the auxiliary field $J$ linked to $\bar{m}$ as $\bar{m}= \left. \nabla_J \mathcal{A}^{(P)}_{\alpha_{P-1},\rho,\beta   }( J)\right\vert_{J=0}$, we have
\begin{equation}
    \begin{array}{lll}
    \m=\mathbb{E}\left\{\tanh{\left[\beta'\dfrac{P}{2}(1+\rho)^{P/2-1}\n^{P-1}\tilde{\eta}+Y\beta'\sqrt{   \gamma\dfrac{P}{2}\q^{P-1}}\right]}\xi\right\}.
    \end{array}
    \label{eq:m_selfRS_sup}
\end{equation}
\normalsize
\end{proposition}

For the proof we refer to Appendix \ref{app:proof}.

\subsection{Low-load regime} \label{ultratolerance}
In this subsection we inspect a peculiarity of dense networks \cite{BarraPRLdetective}: when set in the low-load regime ($b<P-1$), these networks can work even in the presence of severe noise affecting the tensor entries and mimicking flaws  occurring during the learning stage \cite{AgliariDeMarzo}.
\begin{proposition}
\textit{(Low-load regime)} In the thermodynamic limit ($N\to\infty$), under the RS assumption, the quenched statistical pressure for the dense, supervised Hebbian network for even $P \geq 4$ and a low-storage regime $b<P-1$ reads as
\begin{equation}
\begin{array}{lll}
   \mathcal{A}_{0,\rho, \beta }^{(P)}( J) &=&\mathbb{E}\left\{\ln{}\left\lbrace{2\cosh\left[J+\beta'\dfrac{P}{2}(1+\rho)^{P/2-1}\n^{P-1}\tilde{\eta}\right]}\right\rbrace  +\dfrac{\beta'}{2}(1+\rho)^{P/2}(1-P)\n^{P}\right\}
\end{array}
\end{equation}
\normalsize
with $\mathbb{E}=\mathbb{E}_{\xi}\mathbb{E}_{(\eta|\xi)}$ and $\bar n$ fulfill the following self-consistency equation
\begin{equation}
\label{eq:nlow}
    \begin{array}{lll}
        \n=\dfrac{1}{r(1+\rho)}\mathbb{E}\left\{\tanh{\left[\beta'\dfrac{P}{2}(1+\rho)^{P/2-1}\n^{P-1}\tilde{\eta}\right]}\tilde{\eta} \right\}.
    \end{array}
\end{equation}
Furthermore, considering the auxiliary field $J$ linked to $\bar{m}$  as $\bar{m}= \nabla_J \mathcal{A}^{(P)}_{\rho,\beta , r, \alpha_{b} }( J)\vert_{J=0}$, we have
\begin{equation}
    \begin{array}{lll}
        \m=\mathbb{E}\left\{\tanh{\left[\beta'\dfrac{P}{2}(1+\rho)^{P/2-1}\n^{P-1}\tilde{\eta}\right]}\xi\right\}.
    \end{array}
\end{equation}
\end{proposition}
The proof is similar to the one presented in Appendix \ref{app:proof} for the high-load regime, therefore we omit it. 

As mentioned in the preamble of this subsection, now we briefly dwell on the ability of this system to work even when an extensive amount of noise %(scaling with the size of the network) 
is added in the synaptic tensor.
To do so, we need to present some details regarding the signal-to-noise analysis, which is a technique used to ascertain the stability of a given configuration.
In the absence of external noise, the neuronal dynamics is described by the following updating rule
\begin{equation}    \sigma_i^{(n+1)}=\sigma_i^{(n)}\mathrm{sign}\left[\sigma_i^{(n)}h_i\left(\bm\sigma^{(n)}|\bm\eta\right)\right],
    \label{eq:dynam_rul}
\end{equation}
where
\begin{align}
    h_i(\boldsymbol{\sigma}|\bm \eta)=& \dfrac{1}{N^{P-1}M^P \mathcal{R}^{P/2}}\SOMMA{\mu=1}{K}\SOMMA{(i_2,\cdots,i_P)\neq i}{N,\cdots,N}\SOMMA{a_1,\cdots,a_P}{M,\cdots,M} \eta^{\mu, a_1}_{i}\eta^{\mu, a_2}_{i_2}\cdots\eta^{\mu,a_P}_{i_P}\sigma_{i_2}...\sigma_{i_P} \, ,\label{eq:campo_h}
\end{align}
is the internal local field acting on the $i$-th neuron\footnote{It is easy to see, using Eq.\eqref{eq:campo_h} in Eq. \eqref{def:H_Psup}, that
 \begin{equation}
    -\mathcal{H}^{(P)}_{N,K,M, r}(\boldsymbol{\sigma} \vert \bm \eta)=\SOMMA{i=1}{N}h_i(\bm \sigma|\bm \eta)\sigma_i \, .
\end{equation}}.
Then, the configuration $\boldsymbol \sigma = \boldsymbol \xi^{\mu}$ is stable under the dynamics \eqref{eq:dynam_rul} as long as $h_i(\bm \xi^\mu|\bm \eta)\xi_i^\mu>0$ for $i=1,\hdots, N$.
Now, the l.h.s. in \eqref{eq:campo_h} can be split into a \textit{signal} term $S$ and a \textit{noise} term $R$. More precisely, without loss of generality, if we focus on the the stability of the configuration $\bm\xi^1$, $S$ includes the field contribution which tends to align the neural configuration with the chosen pattern $\bm\xi^1$, while $R$ includes the remaining contributions, which tend to destroy the correlation of the neural configuration and the first pattern. For the supervised setting we have
\begin{equation}
\label{eq:S_N}
    \begin{array}{lll}
         S&=&  \dfrac{1}{N^{P-1}M^P \mathcal{R}^{P/2}}\SOMMA{(i_2,\cdots,i_P)\neq i}{N,\cdots,N}\SOMMA{a_1,\cdots,a_P}{M,\cdots,M} \eta^{1, a_1}_{i}\eta^{1, a_2}_{i_2}\cdots\eta^{1,a_P}_{i_P}\xi_{i}^1\xi_{i_2}^1\cdots\xi_{i_P}^1\,,
         \\\\
         R&=& \dfrac{1}{N^{P-1}M^P \mathcal{R}^{P/2}}\SOMMA{\mu>1}{K}\SOMMA{(i_2,\cdots,i_P)\neq i}{N,\cdots,N}\SOMMA{a_1,\cdots,a_P}{M,\cdots,M} \eta^{\mu, a_1}_{i}\eta^{\mu, a_2}_{i_2}\cdots\eta^{\mu,a_P}_{i_P}\xi_{i}^1\xi_{i_2}^1\cdots\xi_{i_P}^1\,.
    \end{array}
\end{equation}
We can estimate the largest amount of archetypes that the system can retrieve, asking for the largest value of $K$ which still ensures that $S$ is {more than or at the order of magnitude of} $R$ (from now on expressed as $S\gtrsim R$).
%, where with the symbol " $\gtrsim$ ", we mean \textit{more than or at the order of magnitude of} . 
It can be proved that, as a function of the network size $N$ and the amount of patterns $K$, the behaviours of the signal and the noise terms presented in \eqref{eq:S_N} are $S\circaa N^0$ and  $R\circaa K/N^{P-1}$ respectively. Therefore, we expect that when $K \circaa N^{P-1}$ any additional source of noise may impair retrieval. On the other hand, although our model assumes that the coupling tensor \eqref{SynapticTensor} is devoid of flaws, in general, the communication among neurons can be disturbed and this may imply a supplementary noise affecting the synaptic weight \cite{BarraPRLdetective,AgliariDeMarzo}. 
We model this noise by introducing an additional, random contribution as
\begin{equation}
    \sum_{\mu=1}^{K}\sum_{a_1, \hdots, a_P}^{M,\hdots,M} \eta_{i_1}^{\mu,a_1} \cdots  \eta_{i_P}^{\mu,a_P} \xrightarrow[\;\;\;\;\;]{} \sum_{\mu=1}^{K}\sum_{a_1, \hdots, a_P}^{M,\hdots,M} \eta_{i_1}^{\mu,a_1} \eta_{i_1}^{\mu,a_1} \cdots  \eta_{i_P}^{\mu,a_P} + w\: \tilde{J}_{ i_1, \hdots, i_P}^{\mu, a_1, \hdots, a_P}
    \label{eq:J_con_noise}
\end{equation}
where $\tilde{J}^{\mu, a_1, \hdots, a_P}_{ i_1, \hdots, i_P}\sim\mathcal{N}(0,1)$, $K\circaa\gamma N^b$ and $w\circaa\tau N^\delta$, with $\gamma,\tau \in \mathbb{R}$ and $b,\delta \in \mathbb{R}^{+}$. This ``synaptic noise'' has an effect on the noise contribution, in such a way that  $R\circaa K w^2/N^{P-1}$, thus, to guarantee retrieval, namely $S\gtrsim R$, we need  
\begin{equation}
    \begin{cases}
        b\leq P-1 
        \\
        \delta \leq \dfrac{P-1-b}{2}
    \end{cases}\,.
    \label{eq:delta}
\end{equation}
Therefore, provided that $b$ is strictly smaller than $P-1$, the system can tolerate a synaptic noise that grows with the system size and its rate of growth can be made larger according to the load. Furthermore, if we now set $\delta =(P-1-b)/2$, by calculations analogous to those presented in the first part of this section, using the noisy synaptic tensor defined in Eq.\eqref{eq:J_con_noise},  we recover the self-consistency equations for the order parameters that display exactly the same expressions of the ones for the high-storage regime, namely Eqs. \eqref{eq:n_selfRS_sup}, as long as we replace $\b$ with $\b \tau$. 
For an exhaustive description of the computation we refer to \cite{unsup}, where we treated in details the unsupervised regime.

\subsection{High-load with low entropy datasets}
\label{sec:datalimit_sup}

As explained in Remark~\ref{rem:entrpoy_rho} the parameter $\rho = (1-r^2)/(Mr^2)$ quantifies the amount of information needed to describe the original message
$\boldsymbol\xi^\mu$ given the set of related examples $\{\boldsymbol\eta^{\mu,a}\}_{a=1,...,M}$. In this section we focus on the case $\rho \ll 1$ that corresponds to a low-entropy dataset or, otherwise stated, to a high-informative dataset.

The advantage of the present analysis is that, under the condition  $\rho \ll 1$, we obtain a relation between $\bar n$ (a natural order parameter of the model) and $\bar m$ (a practical order parameter of the model\footnote{It is worth recalling that the model is supplied only with examples -- upon which the $n^{\mu,a}$'s are defined -- while it is not aware of archetypes  -- upon which the $m^{\mu}$'s are defined. The former constitute natural order parameters and, in fact, the Hamiltonian $\mathcal H^{(P)}_{N,K,M,r}$ in \eqref{def:H_Psup} can be written in terms of the example overlaps; the latter are practical order parameters through which we can assess the capabilities of the network.}), thus, the self-consistency equation for $\bar n$ can be recast into a self-consistency equation for $\bar m$ and 
%there is no need to introduce the functional generator $J$, moreover, 
the numerical solution of the self-consistency equations can be achieved
% its numerical solution versus the control parameters allows us to get the phase diagram for the system 
more straightforwardly.  
%Second, the self consistency equation for $\bar m$ can be directly compared with an \resub{explicit} expression obtained from the stability analysis approach (developed for $\rho \ll 1$ {and} $T \to 0$) in such a way that we can finally estimate $\mathcal V$, see Sec.~\ref{subsec:S2N_sup}.   
\newline
As explained in Appendix \ref{app:proof}, we start from the self-consistency equations found in the high-load regime \eqref{eq:n_selfRS_sup}-\eqref{eq:m_selfRS_sup} and we exploit the CLT in order to replace $\tilde{\eta}$ with $1+\lambda\sqrt{\rho}$, where $\lambda\sim\mathcal{N}(0,1)$. In this way we reach the simpler expressions
\begin{eqnarray}
    (1+\rho)\n&=&\m +\beta ' \dfrac{P}{2}\rho\left(1+\rho\right)^{P/2-1}(1-\q)\n^{^{P-1}},\label{eq:n_Mgrande}
    \\
    \q&=&\mathbb{E}_{_Z}\left[\tanh{}^{\2}{ g(\beta, Z, \bar{n})}\right],\label{eq:n_of_M_unsup}
    \\
    \m&=& \mathbb{E}_{_Z}\left[\tanh{ g(\beta, Z, \bar{n})}\right],\label{eq:m_Mgrande}
\end{eqnarray}
where
\begin{align}
    &g(\beta, Z, \bar{n})=\beta '\dfrac{P}{2}\n^{^{P-1}}(1+\rho)^{P/2-1}+\beta 'Z\sqrt{\rho\dfrac{P^2}{4}\n^{^{2P-2}}(1+\rho)^{P-2}+\gamma \dfrac{P}{2}  \,\q^{^{P-1}}\;}\;
    \label{eq:g_of_super_n}
\end{align}
and $Z \sim \mathcal{N}(0,1)$ is a standard Gaussian variable. 
\\
Now, as long as $\rho \ll 1$, we can truncate the right-hand-side of \eqref{eq:n_Mgrande} into $\n (1+\rho)= \m$  and, by replacing $\n$ with $\m/(1+\rho)$  in the argument of the hyperbolic tangent, we get
\begin{equation}
    g(\beta ,Z,\m)=\Tilde{\beta}\dfrac{P}{2}\m^{P-1}+Z \Tilde{\beta} \sqrt{\rho\left(\dfrac{P}{2}\m^{P-1}\right)^2 +\gamma\dfrac{P}{2}(1+\rho)^{P}\q^{P-1}}
    \label{eq:g_supervised_large_M}
\end{equation}
where we posed
\begin{equation}
\tilde{\beta} = {\beta'}{(1+\rho)^{-\frac{P}{2}}} =  
%\b\left(\dfrac{r^2}{\R}\right)^{\frac{P}{2}} = 
\frac{2 \beta}{P!} {(1+\rho)^{-\frac{P}{2}}}.
\label{eq:tildebeta}
\end{equation}

As a consequence, we are left with only two out of the three self-consistency equations, namely only the ones for $\q$ and $\m$, while keeping the error (resulting in this truncation) numerically small for $\tilde{\beta}^{-1}$ small (i.e., in the region of interest, where the network shows emergent computational capabilities, namely in the retrieval region depicted in Fig.~\ref{fig:my_label2}); this yields a significant speed-up in the numerical evaluation of the order parameters. The validity of this approximation is corroborated in Fig.~\ref{fig:check_n} where we plot $\n$ versus $r$ for different values of $\rho$, $\gamma$ and $\tilde{\beta}^{-1}$ and compare the outcomes obtained with and without the truncation.

With this truncation we further handle Eqs.~\eqref{eq:n_Mgrande}-\eqref{eq:m_Mgrande} and we compute their zero-temperature limit. 
%Beyond an interest in the ground-state \emph{per se}, here this analysis is pivotal because, combined with the results obtained by the stability analysis  discussed in the next Sec.~\ref{subsec:S2N_sup}, it allows us to estimate the factor $\V$ that appears in the self-consistency equations and whose knowledge is mandatory to solve those equations numerically.
As detailed in the Appendix \ref{app:proof}, by taking the limit $\beta \to \infty$ in Eqs.~\eqref{eq:n_Mgrande}-\eqref{eq:m_Mgrande} we get
\begin{equation}
    \begin{array}{lll}
         \m = \mathrm{erf}\left[\dfrac{P}{2}\dfrac{\m^{P-1}}{G}\right] \ ,
          \\\\
         G:=\sqrt{2\left[\rho\left(\dfrac{P}{2}\m^{P-1}\right)^2 +\gamma \dfrac{P}{2}(1+\rho)^{P}\right]} \ ,
         \\\\
         \q=1.
    \end{array}
    \label{eq:noiseless}
\end{equation}
where we use the definition of the P-independent load presented in Eq. \eqref{eq:load_gamma}.

\begin{figure}[t]
    \centering
    \includegraphics[width =15cm]{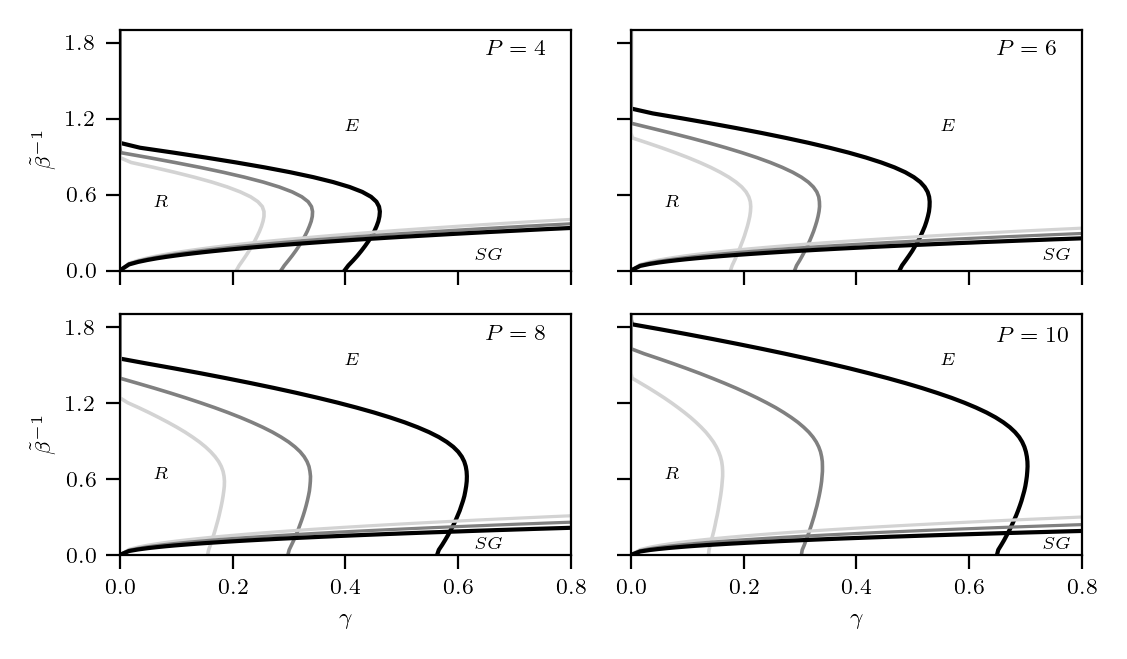}
    \caption{Phase diagrams of the dense Hebbian network, trained under supervised learning and comparison with standard dense Hopfield networks (black line).  Different panels correspond to different values of $P$ and in each panel we plot the $P$-independent load $\gamma$ vs the rescaled temperature $\tilde{\beta}^{-1}$ (as introduced respectively in Eqs. \eqref{eq:load_gamma} and \eqref{eq:tildebeta}), while we fixed $r=0.2$ and two values of $\rho$ (which are $\rho=0.1$ light gray and $\rho=0.05$ gray). 
    % \resub{In x-axis there is the P-independent load network $\gamma$ as defined in Eq. \eqref{eq:load_gamma} and y-axis $\tilde{\beta}$ as in Eq. \eqref{eq:tildebeta}}. 
     We notice that the retrieval zone gets wider as $P$ increases and the re-entrance in the transition line delimiting this region is expected to be a pathology of the RS solution, that should be removed by a RSB-scheme \cite{Albanese2021}. Moreover, in each plot we have a region of the phase diagram which is an intersecting zone between retrieval and spin-glass regions, that shrinks as $P$ increases. In this region the self-consistency equations \eqref{eq:n_of_M_unsup}--\eqref{eq:n_Mgrande} have both a \textit{retrieval solution} ($\m\sim 1$ and $\q \sim 1$) and a \textit{spin-glass solution} ($\m\sim 0$ and $\q\sim 1$), but only the former is stable}. 
    \label{fig:my_label2}
\end{figure}

Finally, in Fig. \ref{fig:my_label2} we plot the phase diagrams obtained by solving the self-consistency equations \eqref{eq:n_Mgrande}-\eqref{eq:m_Mgrande} for different values of $\rho$ and for different values of $P$: the three regions represented by the letters (R, SG, E) are, respectively, the  {\em retrieval region} (where both the Mattis magnetization and the overlap are close to one),  the {\em spin glass region} (where solely the overlap is close to one, but there is no longer Mattis magnetization) and the {\em ergodic region} (where both the order parameters are null). It is also instructive to compare the transition lines that split the phase diagram of these dense networks trained with supervised learning with those obtained for the standard Hebbian dense networks: the retrieval region of the latter works as an upper bound for the retrieval region of the former. In fact, in the next section we prove that, in the large dataset-size limit, supervised Hebbian learning reaches the theoretical bound predicted by Hebbian storing (i.e. the width of the various regions of their phase diagrams do coincide).

% NELLA CAPTION DI FIG. 3 SI PARLA DI 'INSTABILITY REGION' VA CHIARITO COSA SI INTENDE \\ COSI'? NI': PERCHE' SI CHIAMA REGIONE DI INSTABILITA' E QUALE STATO PREVALE?

% \green{HO RIMODIFICATO}

\begin{figure}[t]
    \centering
    \includegraphics [width=15.5cm]{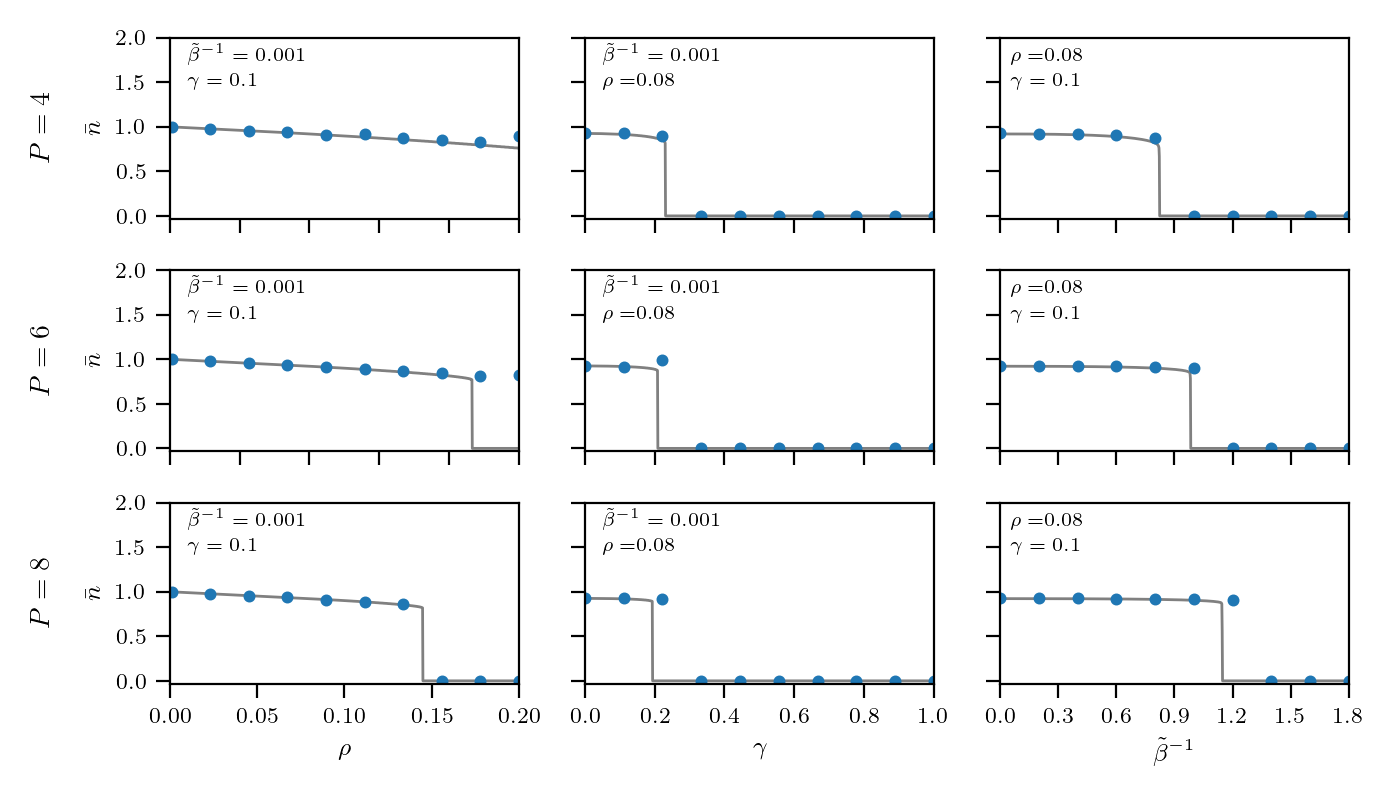}
        \caption{In order to check the goodness of the truncation $\n(1+\rho) = \m$ we compute $\n$ solving numerically \eqref{eq:n_Mgrande}-\eqref{eq:m_Mgrande}, with and without approximation, for different values of $P$, $\tilde\beta^{-1}$, $\gamma$ and $\rho$.  In particular, in this mosaic of panels, each row corresponds to a different value of $P$ (from top to bottom $P=4$, $P=6$, and $P=8$) and each columns corresponds to a different quantity which $\n$ is plotted versus (from left to right $\rho$, $\gamma$, and $\tilde\beta^{-1}$), while the other considered parameters are kept fixed as reported. In each panel we compare the results obtained using the exact expression of $\n$ \eqref{eq:n_Mgrande} (blue dots) and those obtained using the approximated one (namely $\n=\m(1+\rho)^{-1}$, solid grey line). Although the approximation does not seem entirely accurate, it is reasonable for our purposes, and in particular for $\rho \to 0$ and $\beta \rightarrow \infty$.}  
    \label{fig:check_n}
\end{figure}

\section{Numerical findings}
\label{sec:numerical}

\subsection{Stability analysis and Monte Carlo simulations}
\label{subsec:S2N_sup}
% RISCRIVERE TUTTA LA PARTE RESUB.
% A PARTE I REFUSI E LE IMPRECISIONI LINGUISTICHE, NON E' CHIARO PER CHI LEGGE COSA SI INTENDE PER "the noise distribution of the network is not scale-free" ABBIAMO MAI DEFINITO UNA NOISE DISTRIBUTION?\\
% \green{MOdificato leggermente, ma ancora di sistemare}
% \\
%%%%%%%%%%%%%%%%%%%%%%%%%%%%%%%%%%%%%%%%%%%%%%%%%
%%%%%%%%%%%%%%%%%%%%%%%%%%%%%%%%%%%%%%%%%%%%%%%%%
%%%%%%%%%%%%%%%%%%%%%%%%%%%%%%%%%%%%%%%%%%%%%%%%%
%%%%%%%%%%%%%%%%%%%%%%%%%%%%%%%%%%%%%%%%%%%%%%%%%
% \resub{As mentioned at the beginning of Sec. \ref{sec:RS_Guerra_sup}, to avoid the problem of non-gaussianity of the post-synaptic potentials distribution, we exploit probability theory to prove that this one is not heavy-tailed. Thus, using the central limit theorem, we approximate it as a Gaussian distribution with the suitable mean and variance. In order to find the suitable value of the variance $\V$, in this Section, at first, independently from the statistical mechanical route\footnote{This new approach needs no Gaussian approximation therefore it is $\V$ independent.}, we achieve in the noiseless limit ($\beta \to +\infty$)  an explicit expression  of the self-consistency equation for $\m$. Then, comparing it with the corresponding noiseless-limit self equation for $\m$ (i.e. Eq. \eqref{eq:noiseless}), we are able to estimate $\V$.} 
%%%%%%%%%%%%%%%%%%%%%%%%%%%%%%%%%%%%%%%%%%%%%%%%%
%%%%%%%%%%%%%%%%%%%%%%%%%%%%%%%%%%%%%%%%%%%%%%%%%
%%%%%%%%%%%%%%%%%%%%%%%%%%%%%%%%%%%%%%%%%%%%%%%%%
%%%%%%%%%%%%%%%%%%%%%%%%%%%%%%%%%%%%%%%%%%%%%%%%%
As mentioned at the beginning of Sec. \ref{sec:RS_Guerra_sup}, to avoid the problem of non-gaussianity of the post-synaptic potentials distribution, we exploit the CLT and the universality of the quenched noise \cite{Genovese, CarmonaWu} and, in the thermodynamic limit, we approximate the non-gaussian post-synaptic potentials as (namely the terms in the round brackets in Eq. \eqref{eq:partitionfunction2}) with a Gaussian distribution with the suitable mean and variance. 
Now, in this Section, independently from the statistical mechanical route, using a stability analysis approach\footnote{This new approach needs no Gaussian approximation.}, we achieve in the noiseless limit ($\beta \to +\infty$)  an explicit expression  of the self-consistency equation for $\m$. The comparison of the results obtained from this approach with noiseless limit self equations  obtained in Eqs. \eqref{eq:noiseless}  by Guerra's Interpolation technique provides a corroboration of the correctness of our results. 
%We  suppose that the network is in a retrieval configuration, say $\bm \sigma = \bm \xi^1$ without loss of generality, we evaluate the local field $h_i(\bm \xi^1)$ acting on the generic neuron $\sigma_i$, and inspect if and when $h_i(\bm \xi^1)\si >0$ is satisfied for any $i=1, \hdots, N$: this is known as {\em stability analysis}.

%\resub{To do so, we suppose that the network is in a retrieval configuration, say $\bm \sigma = \bm \xi^1$, without loss of generality, we evaluate the local field $h_i(\bm \xi^1)$ acting on the generic neuron $\sigma_i$, and, using the spin configuration  updating rule (\eqref{eq:dynam_rul}), we inspect how the value of the Mattis-magnetization changes after one step of MC update: this is known as {\em stability analysis}.}

\par\medskip
Resuming Eq. \eqref{eq:dynam_rul} obtained within the signal-to-noise analysis, we now consider the configuration $\boldsymbol \sigma ^{(1)} = \boldsymbol \xi^1$ and we update the system to get $\boldsymbol \sigma ^{(2)}$, whose related Mattis magnetization reads as
\begin{equation}
    m_1^{(2)}=\dfrac{1}{N}\SOMMA{i=1}{N}\xi_i^1\sigma_i^{(2)}=\dfrac{1}{N}\SOMMA{i=1}{N}\mathrm{sign}\left[\xi_i^{1}h_i^{(1)}\Big(\boldsymbol{\xi}^{1}\Big\vert\bm\eta\Big)\right].
    \label{eq:m_1_step}
\end{equation}
Recall the definition of the internal local field in Eq. \eqref{eq:campo_h}:  
\begin{align*}
    h_i(\boldsymbol{\sigma}|\bm \eta)=& \dfrac{1}{2N^{P-1}M^P \mathcal{R}^{P/2}}\SOMMA{\mu=1}{K}\SOMMA{\underset{}{(i_2,\cdots,i_P)}}{N,\cdots,N}\SOMMA{a_1,\cdots,a_P}{M,\cdots,M}\eta^{\mu, a_1}_{i}\eta^{\mu, a_2}_{i_2}\cdots\eta^{\mu,a_P}_{i_P}\sigma_{i_2}...\sigma_{i_P} \, 
\end{align*}
and using the CLT on the sums over $\mu, i_2,\hdots,i_P$, and $a_1,\hdots,a_P$, we can replace each term $\xi_i^1h_i(\bm \xi^1|\bm\eta)$ appearing in the argument of the sign function in Eq.\eqref{eq:m_1_step},  as $\mu_1+z_i\sqrt{\mu_2}$ with $z_i \sim \mathcal{N}(0,1)$ and
\begin{eqnarray}
\mu_1&\coloneqq \mathbb{E}_{\xi}\mathbb{E}_{(\eta|\xi)}\left[\xi_i^{1}h_i^{(1)}\Big(\boldsymbol{\xi}^{1}\Big\vert\bm\eta\Big)\right],\label{eq:mu1}
\\
\mu_2&\coloneqq \mathbb{E}_{\xi}\mathbb{E}_{(\eta|\xi)}\left\{\left[ h_i^{(1)}\Big(\boldsymbol{\xi}^{1}\Big\vert\bm\eta\Big)\right]^2\right\}.\label{eq:mu2}
\end{eqnarray}
Thus, Eq. \eqref{eq:m_1_step} becomes
\begin{equation} \label{eq:bibo}
    m_1^{(2)}=\dfrac{1}{N}\SOMMA{i=1}{N}\mathrm{sign}\left(\mu_1 +z_i\sqrt{\mu_2-\mu_1^2}\right)\,.
\end{equation}
For large values of $N$, the arithmetic mean appearing in the r.h.s. of Eq.~\eqref{eq:bibo} can be replaced with the expected value, that is, being $g(z)=\mathrm{sign}(\mu_1 + z \sqrt{\mu_2})$,
\begin{equation}
    \dfrac{1}{N}\SOMMA{i=1}{N} g(z_i) \;\;\mathrm{with}\;\; z_i\sim \mathcal{N}(0,1) \xrightarrow[N\to\infty]{} \mathbb{E}[g(z)]=\displaystyle\int \dfrac{dz}{\sqrt{2\pi}}e^{-\frac{z^2}{2}} g(z)\,.
    \label{eq:largemean}
\end{equation}
Therefore, as $N\gg1$, we can rewrite Eq. \eqref{eq:m_1_step} as
\begin{equation}
    m_1^{(2)}=\displaystyle\int\dfrac{dz\,e^{-\frac{z^2}{2}}}{\sqrt{2\pi}} \mathrm{sign}\left(\mu_1+z\sqrt{\mu_2-\mu_1^2}\right)= \mathrm{erf}\left(\dfrac{\mu_1}{\sqrt{2(\mu_2-\mu_1^2)}}\right) \,.\label{eq:m1}
\end{equation}

We carry out the computations of $\mu_1$ and $\mu_2$ in Appendix \ref{app:momenta}, while hereafter we simply report their values
\begin{eqnarray}\label{CoeffiMu1}
\mu_1&=& \dfrac{1}{(1+\rho)^{P/2}} \, ,
\\ \label{CoeffiMu2}
\mu_2&=& \dfrac{1}{(1+\rho)^{P}}\left[\alpha_{P-1}(P-1)!(1+\rho)^P+1+\rho\right] \, .
\end{eqnarray}
By replacing $\mu_1$ and $\mu_2$ in Eq. \eqref{eq:m1} with the r.h.s. of Eqs.  \eqref{CoeffiMu1} \eqref{CoeffiMu2}, we get
\begin{equation}
    m_1^{(2)}= \mathrm{erf}\left\{\left[2\left(\alpha_{P-1}(P-1)!(1+\rho)^P+\rho\right)\right]^{-1/2}\right\} \, .
\end{equation}
Remarkably, this is an explicit expression depending directly on the control parameters of the system.

Thus, the expression for $m_1^{(2)}$, using the explicit expression of $\alpha_{P-1}= 2\gamma/P! $ (see Eq. \eqref{eq:load_gamma}), becomes
\begin{equation}
\label{eq:expr_m1}
    m_1^{(2)}= \mathrm{erf}\left\{\left[{2\left(\gamma\dfrac{2}{P}(1+\rho)^P+\rho\right)}\right]^{-1/2}\right\}.
\end{equation}

%We can finally evaluate the parameter $\V$, \resub{by comparing \eqref{eq:expr_m1} and}
% via stability analysis we obtained
% \begin{equation}
%     m_1^{(2)}\sim \mathrm{erf}\left(\dfrac{1}{\sqrt{2\left[\gamma\dfrac{2}{P}(1+\rho)^P+\rho\right]}}\right),
% \end{equation}
% while, from 
% the self-consistency equations in \eqref{eq:noiseless} for $\beta \to \infty$, that is
% \begin{equation}
%     \begin{array}{lll}
%          \m = \mathrm{erf}\left\{\left[{{2\left(\gamma\dfrac{4}{P^2(2P-3)!!}\V^2(1+\rho)^P+\rho\,\;\right)}}\right]^{-1/2}\right\}.
%     \end{array}
% \end{equation}
% To have a perfect match between the two equations, we have to fix the variance of the noise term, namely $\V$, as 
% \begin{equation} \label{eq:VV}
%     \V= \sqrt{\dfrac{P(2P-3)!!}{2}} \,.
% \end{equation}

% Therefore, \resub{we can replace the value of $\V$ found in \eqref{eq:VV} in Eq. \eqref{eq:g_supervised_large_M} which} finally becomes
% \begin{align}
%     &g({\beta}, \gamma, Z)=\Tilde{\beta}\dfrac{P}{2}\m^{^{P-1}}+\Tilde{\beta} Z\sqrt{\rho\left(\dfrac{P}{2}\m^{^{P-1}}\right)^2+\gamma\dfrac{P}{2}(1+\rho)^P\,\q^{^{P-1}}\;}  \ .
%  \end{align}

%In order to corroborate the (stability analysis driven) estimate of $\V$, 
In order to check the goodness of the RS solution, we run numerical computations using MC simulations with Plefka's dynamics \cite{Plefka1,SollichBraviOpper}: this technique allows us to speed up the updating procedure of the synaptic tensor which is enormous in dense Hebbian networks. We refer to Appendix \ref{app:plefka} for discussions and proofs about this method, whereas, here we compare the estimates for the magnetization $\bar m$ and the susceptibility $\partial_r \bar m$ obtained by Eq.~\eqref{eq:expr_m1} and by simulations finding an excellent alignment, see Fig. \ref{fig:s2n_unsup}. 

Finally, in order to get some insights into the dynamical process, in Fig.~\ref{fig:sup_attractor}, assuming that attractors are centered at $\hat{\eta}_M$, we plot 
the Hamming distance between $\hat{\eta}_M$ and the initial state $\boldsymbol \sigma^{(0))}$, versus the Hamming distance between $\hat{\eta}_M$ and the final state $\boldsymbol \sigma^{(\infty)}$, the latter being obtained using MC simulations with Plefka's dynamics. We highlight that, as $\rho$ increases, the width of the attractors is only slightly enlarged.

\begin{figure}[tb]
    \centering
    \includegraphics[scale=0.41]{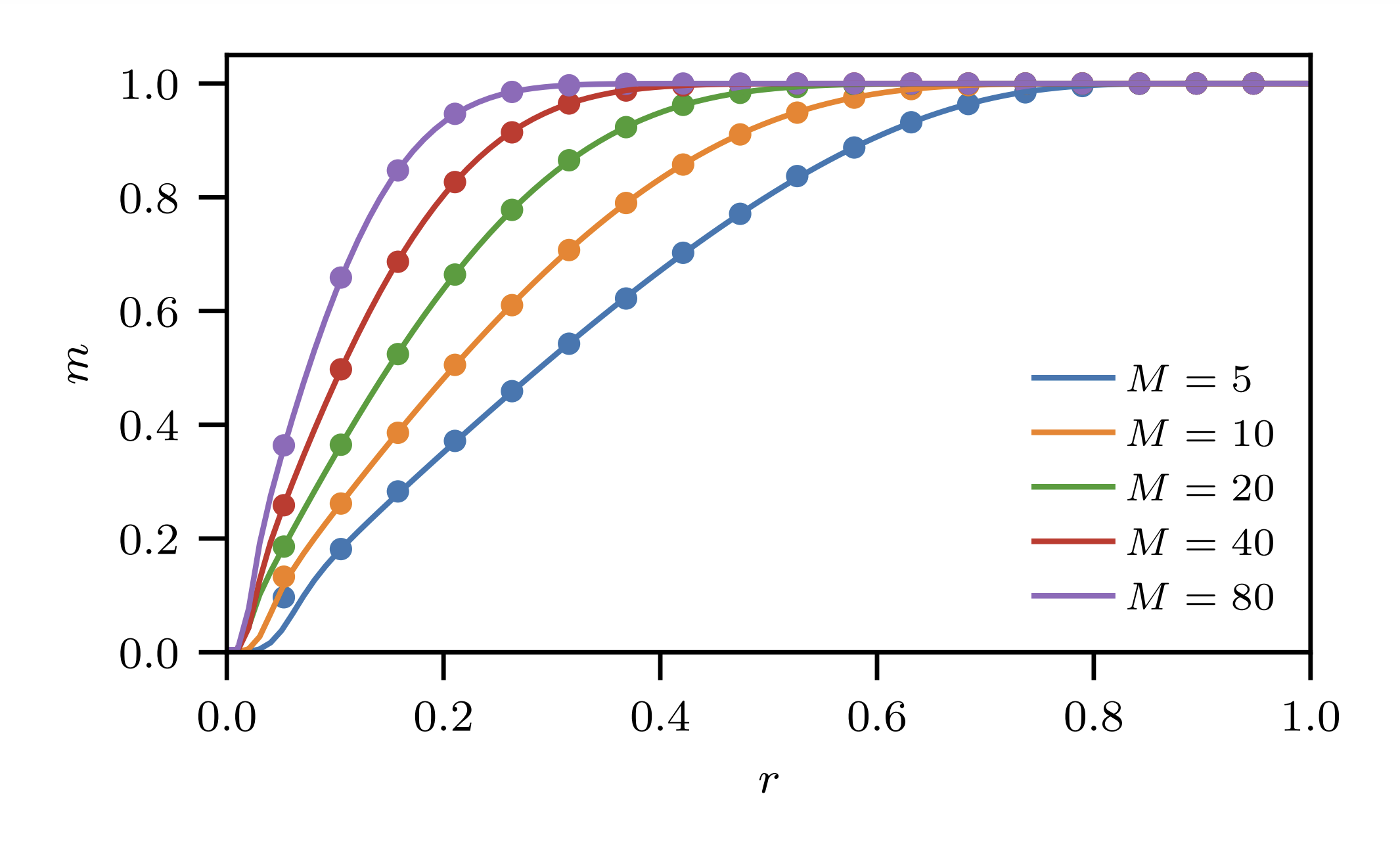}
    \includegraphics[scale=0.41]{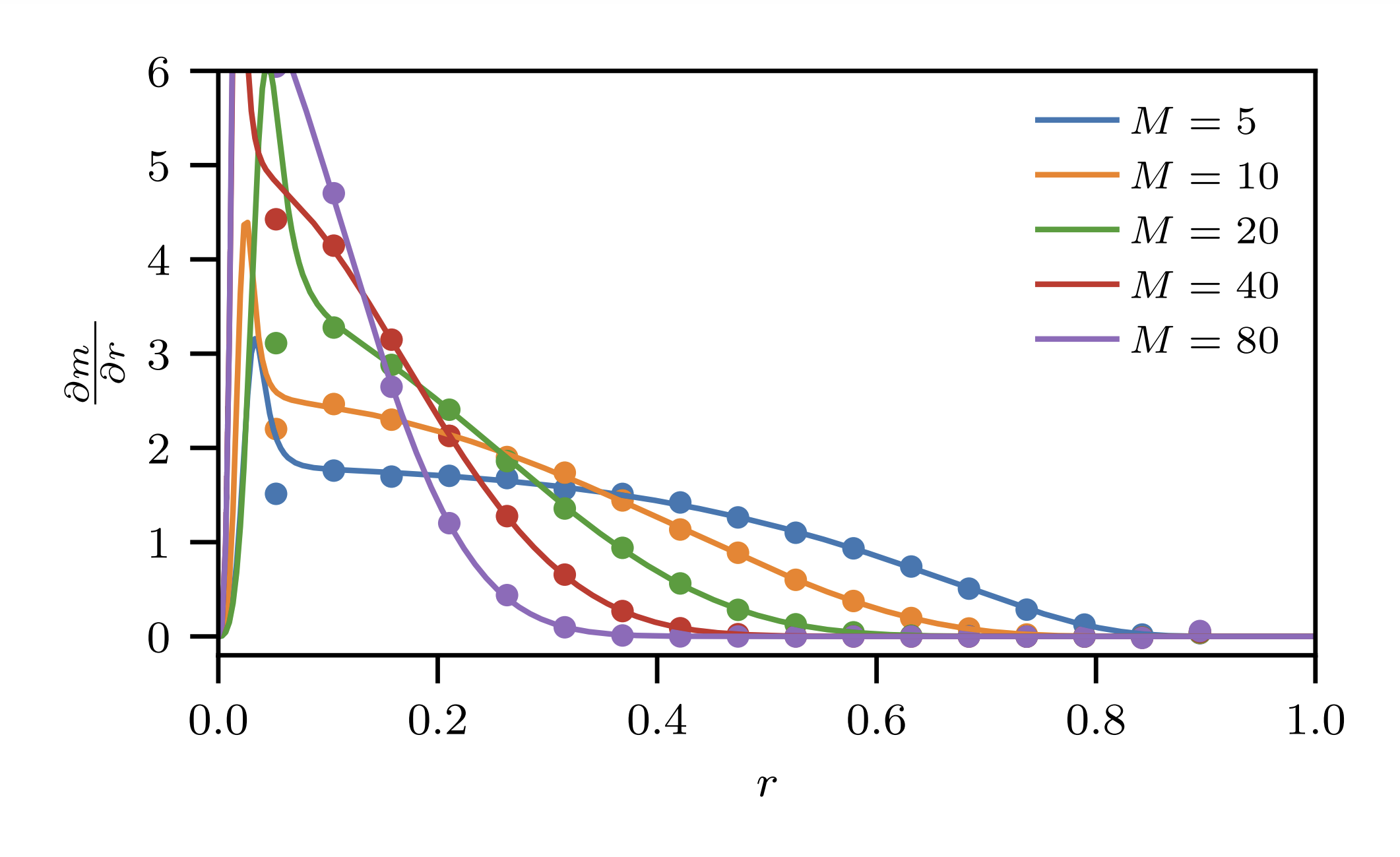}
    %$P=4, \ N=500, \ K=400,$
    \caption{Comparisons between MC simulations equipped with Plefka's dynamic (lines) and stability analysis (dots). 
    The number of examples $M$ varies as specified by the legend, while the number of neurons and classes and are kept fixed at $N=500, \ K=400$. As a consequence, the load is fixed below its critical value, $\gamma<\gamma_C$. The simulations concern the supervised regime with degree of interaction $P=4$ 
    %and corroborate our estimate for $\V$ (see Eq. \eqref{eq:VV})
    . In particular, we report the  archetype magnetization $m$ and its susceptibility $\partial_{r}m$ at various training set sizes $M$ by making the noise $r$ in the training set vary from $0$,  where all the examples are pure random noise, to $1$, where there is no difference among examples and archetype.
    We note that in the small noise limit $r\to1$ the network always perfectly retrieves the archetype as expected, whereas, for $r\to 0$, no retrieval is possible.
    We notice that, when $m > 0.1$, the stability analysis perfectly captures the behaviour of the susceptibility predicted by Plefka's dynamics. 
    }
    \label{fig:s2n_unsup}
\end{figure}
  
\begin{figure}[tb] 
         \centering
         \includegraphics[width=11cm]{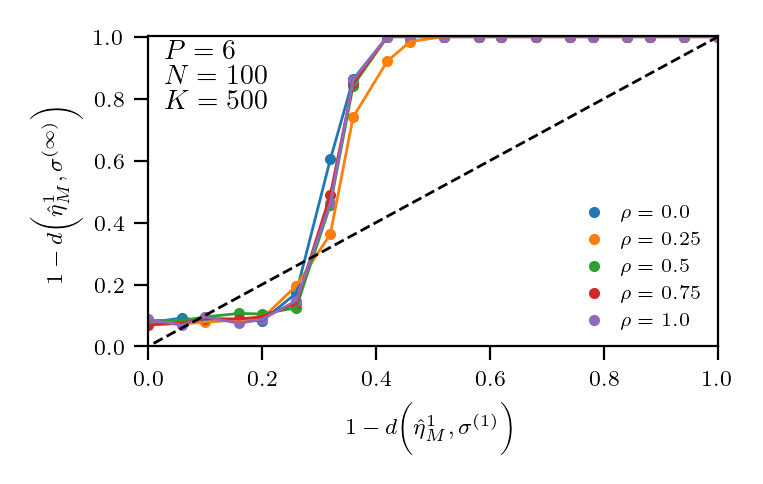}
    \caption{Investigation on the basin of the attractors in a dense network built of by $N=100$ neurons, interacting with a density $P=6$ and handling $K=500$ patterns. 
    For different values of the dataset entropy $\rho$, as reported in the legend, we run MC simulations mixed with Plefka's dynamics starting from the configuration $\boldsymbol \sigma^{(1)}$ and collecting the final configuration $\boldsymbol \sigma^{(\infty)}$ reached by the neural dynamics. Then, we plot the related Hamming distances with respect to the average $\hat{\boldsymbol \eta}_M^1$ among the items of the class labeled with $1$. Configurations relatively close to $\hat{\boldsymbol \eta}_M^1$ tend to get closer (data points fall in the upper region identified by the bisector represented by the dashed line) and when the initial distance is smaller than approximately $0.5$, $\hat{\boldsymbol \eta}_M^1$ plays as an attractor.}
    \label{fig:sup_attractor}	
\end{figure}

\subsection{Critical load and bounds for the dataset size}\label{sec:o_times}
The explicit expression obtained for $m_1^{(2)}$ also allows us to assess a necessary amount of examples which ensure a good performance of the network. In fact,  
if the network has to retrieve one of the archetypes successfully, we need to require that this one-step MC magnetization $m_1^{(2)}$ is larger than  $\mathrm{erf}(\Theta)$ where $\Theta\in\mathbb{R}^+$ is a tolerance level, thus we obtain the criterion 
\begin{equation}
     \dfrac{1}{\sqrt{2\left[\rho +\gamma\dfrac{2}{P}(1+\rho)^{P}\right]}}>\Theta,
 \end{equation}
that can be written as 
\begin{equation}
     1>2\Theta^2\left[\rho +\gamma\dfrac{2}{P}(1+\rho)^{P}\right].
     \label{eq:retriv_condiction_superv}
 \end{equation}

Setting the confidence level $\Theta=1/\sqrt{2}$, which corresponds to the fairly standard condition $$\mathbb{E}_{\xi}\mathbb{E}_{(\eta|\xi)}[\xi_i^1h_i^{(1)}(\boldsymbol \xi^1)]>\sqrt{\mathrm{Var}[\xi_i^1h_i^{(1)}(\boldsymbol \xi^1)]}$$ (see Eqs.  \eqref{eq:mu1}, \eqref{eq:mu2}, \eqref{eq:m1}), the previous relation determines a lower bound for $M$, denoted as $M^{(sup)}_\otimes(r,P,\gamma)$, that guarantees the signal (l.h.s) to be prevailing over the noise (r.h.s.).
\newline
\newline
We now discuss a few special cases under the assumption $r \ll 1$ to inspect if the developed theory smoothly recovers known limits.
\newline 
In the low-load regime $\gamma=0$, the expression in \eqref{eq:retriv_condiction_superv} becomes 
\begin{equation}
    M>\dfrac{1}{2}\dfrac{(1-r^{^2})}{r^{\2}} \Longrightarrow M^{(sup)}_\otimes(r,P,0)\circaa \dfrac{1}{2r^{\2}}
\end{equation}
where the last relation holds for $r \ll 1$: note that, as expected \cite{KanterPowLow1,KanterPowLow2}, the power-law scaling  $M(r)\propto r^{-2}$ in the dataset threshold for learning is recovered.  

If $\gamma \neq 0$ and $P=2$, the classic Hopfield picture  \cite{AgliariDeMarzo} is also recovered
\begin{equation}
    M>\dfrac{\sqrt{\gamma}}{r^{\2}\sqrt{2}}\;\Longrightarrow M^{(sup)}_\otimes(r,2,\gamma)= \sqrt{\dfrac{\gamma}{2}}\dfrac{1}{r^{\2}}.
\end{equation}
Finally, if $\gamma \neq 0$ and $P>2$, we have
\begin{equation}
\begin{array}{lll}
\label{eq:loadM}
   M>\left(\dfrac{\gamma}{P}\right)^{1/P}\dfrac{1}{r^{\2}}\;\Longrightarrow M^{(sup)}_\otimes(r,P,\gamma)= \left(\sqrt[P]{\dfrac{\gamma}{P}}\right)\dfrac{1}{r^{\2}}.
    \end{array}
\end{equation}
 
These results  are corroborated by numerical simulations: in Fig. \ref{fig:sup_normal} we plot the number of examples $M$ w.r.t. $\gamma_C$  at different values of $P$ where we define the critical load $\gamma_C$ as the load beyond which a black-out scenario emerges, as predicted by Hebbian storing, namely $\lim\limits_{\gamma\to\gamma_C^{-}}\m\neq 0$ and $\lim\limits_{\gamma\to \gamma_C^{+}} \m= 0$. When $M$ is chosen following the prescription in Eq. \eqref{eq:loadM}, in supervised Hebbian learning a retrieval zone starts to appear. For $M \gg M_\otimes^{(sup)}$ supervised Hebbian learning and Hebbian storing -- set at a given interaction order $P$ -- yield analogous performances, namely, the retrieval regions resulting from the two schemes are the same. 

\begin{figure}[t]
    \centering
     \begin{subfigure}[b]{0.48\textwidth}
         \centering
         \includegraphics[width=1.02\textwidth]{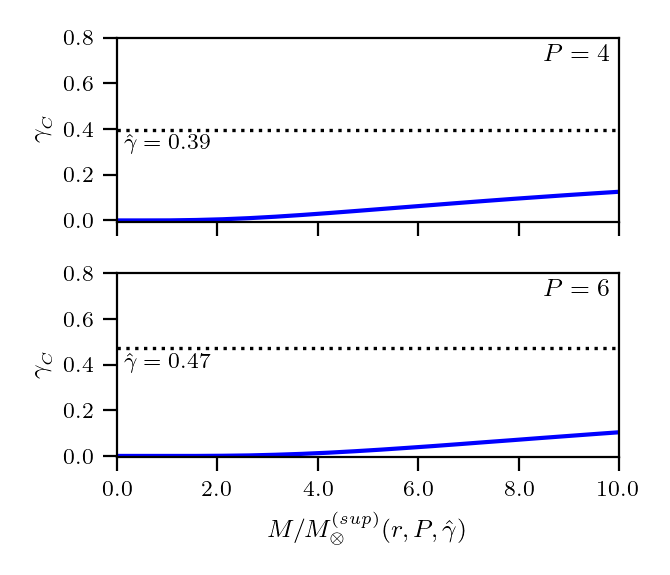}
     \end{subfigure}
     %\hspace{1cm}
     \begin{subfigure}[b]{0.48\textwidth}
         \centering
         \includegraphics[width=1.02\textwidth]{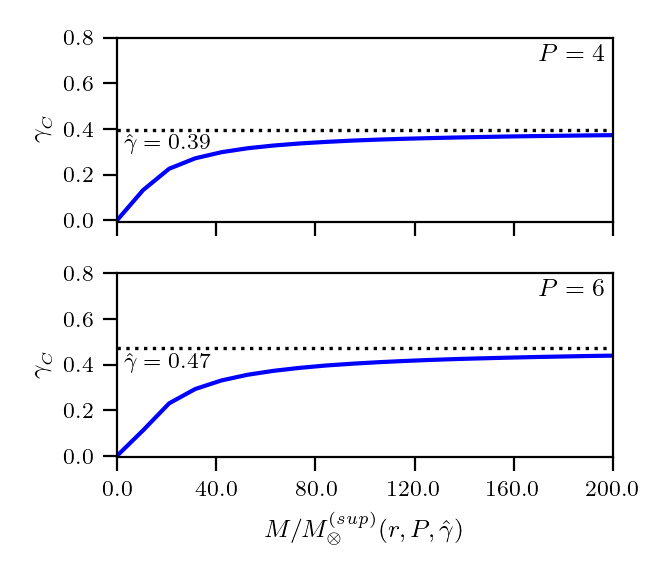}
     \end{subfigure}
     \caption{To inspect how the critical storage evolves as the dataset size varies, we numerically solve the self-consistency equations in the $\beta \to \infty$ limit, namely Eqs. \eqref{eq:noiseless}, for $r=0.2$. In the left column we plot the critical load $\gamma_C$ (blue line) w.r.t. the ratio of the number of examples $M$ and the related lower bound $M_{\otimes}^{sup}$  (see Eq. \eqref{eq:loadM}) for different degrees of interaction $P=4$ (upper panels) and $P=6$ (lower panels): as $\gamma_C$ detaches from zero a non-null retrieval region appears in the phase diagram. In the right column we inspect the large dataset-size limit, namely we broaden the $x-$axis to inspect how $\gamma_C$ increases with $M$. In particular, denoting with $\hat{\gamma}$ the critical load for the standard dense Hopfield network with the same interaction order \cite{EmergencySN}, as long as $M \gg M^{(sup)}_{\otimes}(r, P, \gamma=\hat\gamma)$, $\gamma_C$ saturates to $\hat\gamma$} (represented by the horizontal dashed line). 
    \label{fig:sup_normal}
\end{figure}

\subsection{Application to structured dataset}
\label{sec:mnist}

In the previous sections we worked out the theory under the assumption of a random, structureless dataset and we showed general agreement with numerical simulations run on finite-size architectures.
Now, we extend the numerics to cope with structured datasets: examples are no longer  random, but the pixels in the images are correlated in order to create a real pattern \cite{leonelli2021effective}; in particular, we explore MNist and Fashion MNist datasets. 
\begin{figure}[t]
     \centering
     \begin{subfigure}[b]{0.4\textwidth}
         \centering
         \includegraphics[width=\textwidth]{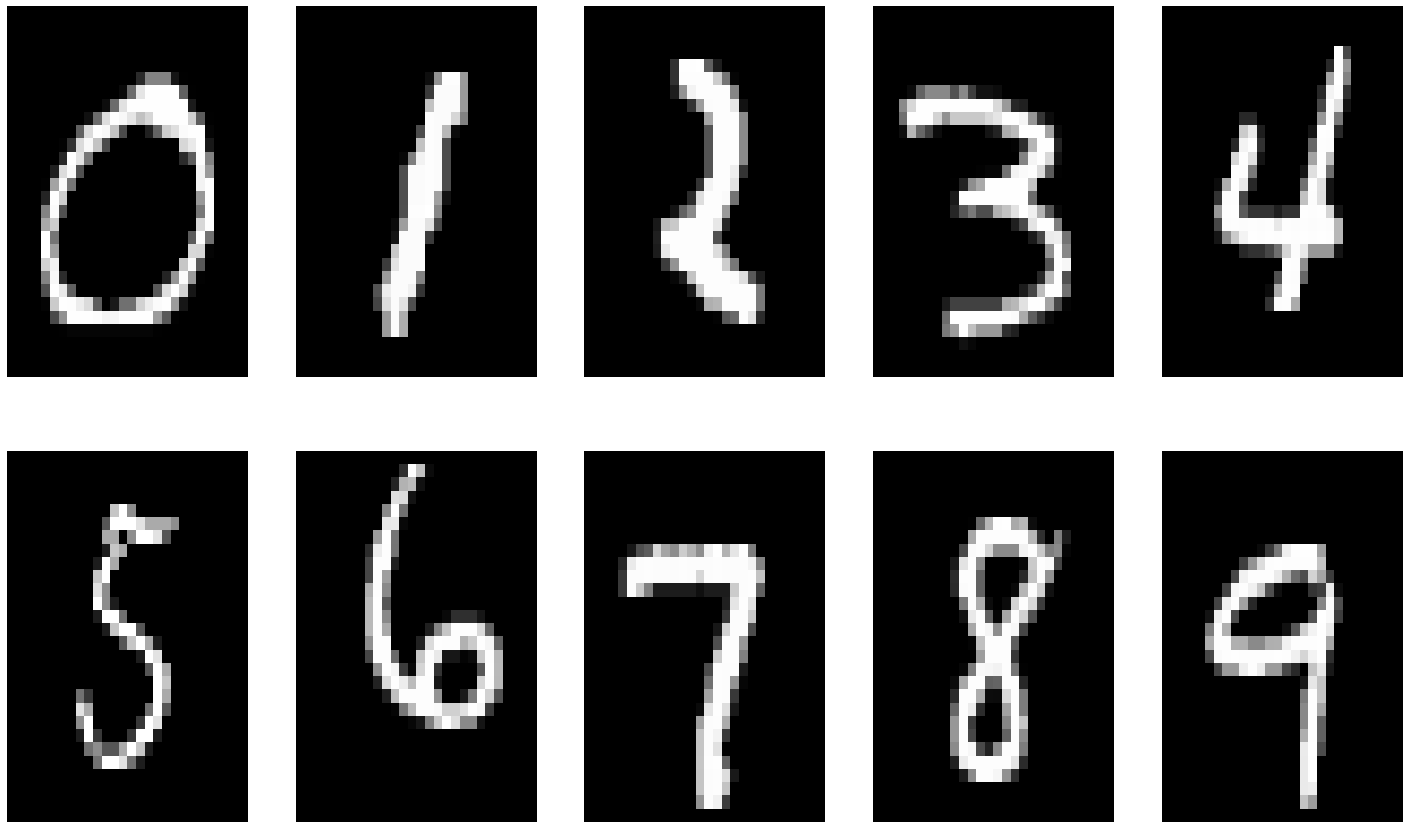}
     \end{subfigure}
     \hspace{1cm}
     \begin{subfigure}[b]{0.4\textwidth}
         \centering
         \includegraphics[width=\textwidth]{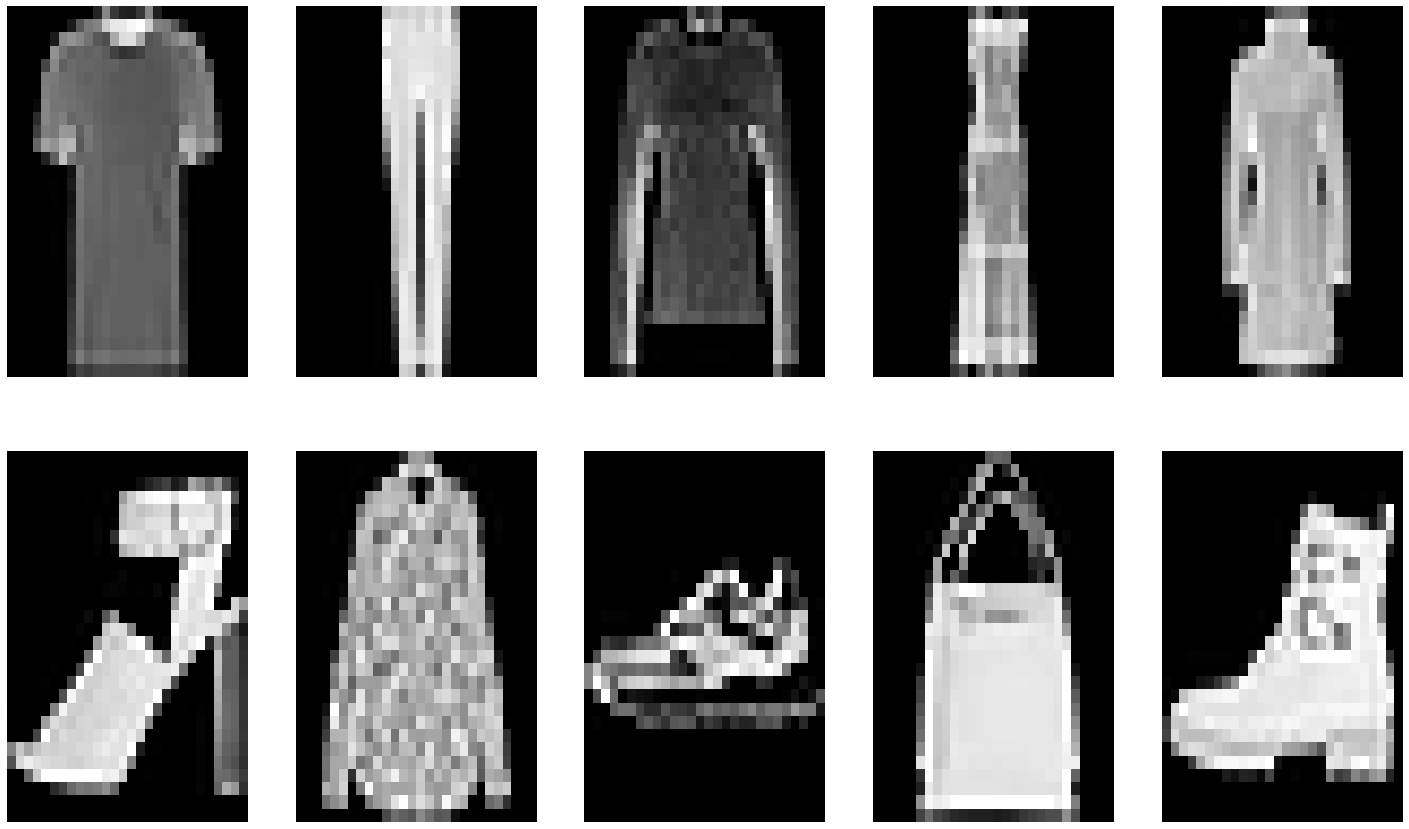}
     \end{subfigure}
        \caption{We report one example chosen randomly from each class contained in MNist (left) and Fashion MNist (right) datasets.}
        \label{fig:datas}
\end{figure}

We recall that MNist is a dataset for image classification built up of handwritten single digits, it has a training set of 60,000 examples, and a testing set of 10,000 examples \cite{deng2012mnist}.
This dataset is made of square 28×28 pixel grayscale images.
Analogously, Fashion-MNIST  comprises 28×28 grayscale images of 70,000 fashion products from 10 categories,  it shares with the MNIST dataset the same image size, data format and the structure of training and testing splits \cite{xiao2017fashion}. By way of example, in Fig. \ref{fig:datas} an item from each class in MNist (left) and Fashion MNist dataset (right) are shown. 

\begin{figure}[t]
    \centering
    \includegraphics[width=12.0cm]{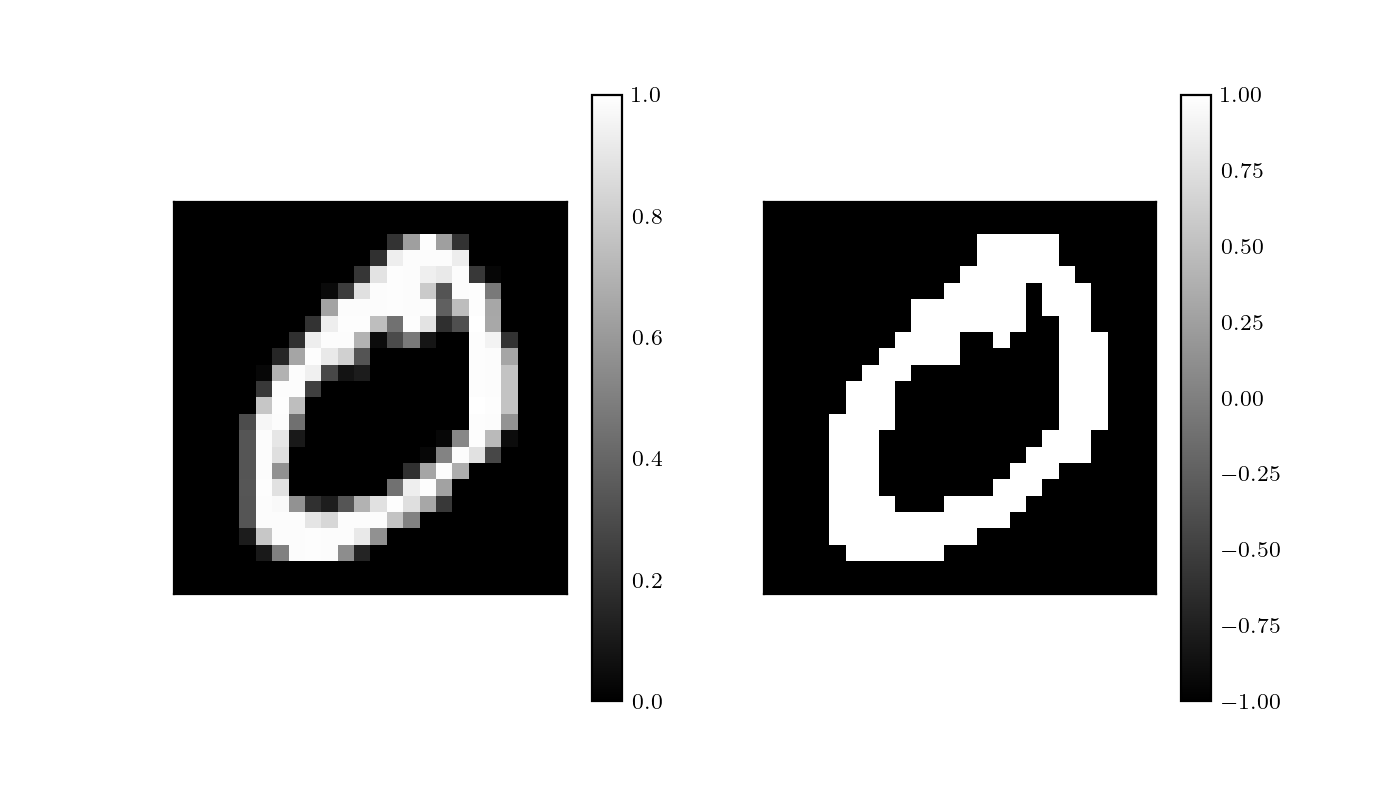}
    \caption{Example from the original real MNist dataset $\lbrace\check{\boldsymbol{\zeta}}^{\mu,a}\rbrace_{a=1,\cdots,M}^{\mu=1,\cdots,K}$ (left) and its binary version  (right); the overall set of binarized patterns constitutes the dataset $\left\lbrace\boldsymbol{{\zeta}}^{\mu,a}\right\rbrace_{a=1,\cdots,M}^{\mu=1,\cdots,K}$.}
    \label{fig:binarizzazione}
\end{figure}

In order to understand the procedure of classification \cite{xiao2015learning, frenay2013classification, veit2017learning}, further notions are in order: from now on we will indicate the structured dataset as $\left\lbrace \check{\boldsymbol{\zeta}}^{\mu,a} \right\rbrace_{a=1,\cdots, M}^{\mu=1,\cdots,K}$, with $\check{\boldsymbol{\zeta}}^{\mu,a}\in \mathbb{R}^N$, where $\mu$ labels the different classes and $a$ labels the items pertaining to the same class.
Before designing the network and inspecting its performances, data must be binarized, namely we map $\left\lbrace \check{\boldsymbol{\zeta}}^{\mu,a} \right\rbrace_{a=1,\cdots, M}^{\mu=1,\cdots,K}$ into its binary version $\left\lbrace {\boldsymbol{\zeta}}^{\mu,a} \right\rbrace_{a=1,\cdots, M}^{\mu=1,\cdots,K}$, ${\boldsymbol{\zeta}}^{\mu,a}\in \{-1,1\}^N$. To this goal, for each pixel $i$ we evaluate the average $\bar{\bar{\zeta}}_i$ over $\mu$ and over $a$, that is
\begin{equation}
    \bar{\bar{\zeta}}_i = \dfrac{1}{KM}\SOMMA{\mu=1}{K}\SOMMA{a=1}{M}\check\zeta_i^{\mu,a}
\end{equation}
and set
\begin{equation}
\zeta_i^{\mu,a}=
    \begin{cases}
    +1\;\;\;\;\mathrm{if}\;\;\;\;\check{\zeta}_i^{\mu,a}>\bar{\bar{\zeta}}_i,
    \\
    -1\;\;\;\;\mathrm{if}\;\;\;\;\check{\zeta}_i^{\mu,a}<\bar{\bar{\zeta}}_i.
    \end{cases}
\end{equation}
In order to show how this procedure works, an example from real and binary MNist dataset are shown in Fig. \ref{fig:binarizzazione}.
The change in the notation for the dataset items with respect to the previous sections ($\bm\eta$ vs $\bm \zeta$) emphasizes that the dataset is now structured and the related archetypes are not defined. 
\begin{figure}[t]
     \centering
         \includegraphics[width=15.5cm]{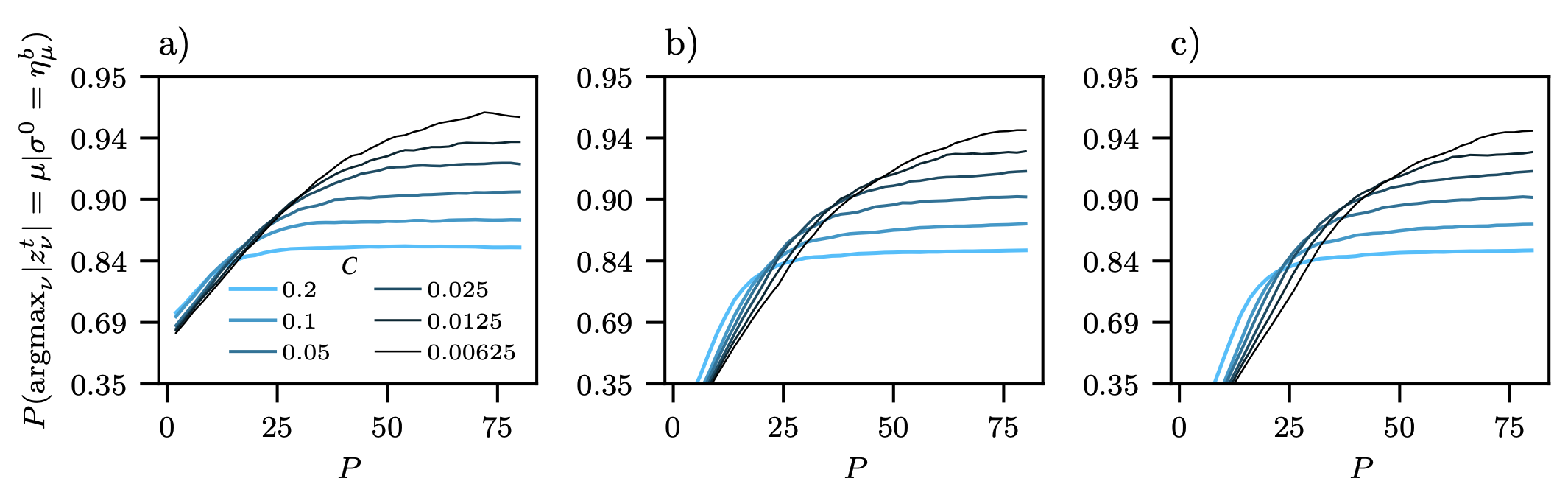}
        \caption{Accuracy of the dense supervised Hebbian networks evaluated in terms of the probability of classifying correctly a test example chosen randomly from the MNIST dataset, as a function of the interaction order $P$, and for different choices of $C$ (as shown by the legend). The different panels $a),b),c)$ corresponds to $1,2,3$ steps of parallel Plefka's dynamics (see Appendix \ref{app:plefka}). Notice that the accuracy grows with $P$ until a saturation accuracy level is reached; we point out that the maximum value of the performance plateau grows as the number of examples increase, namely for $C\to 0$.}
        \label{fig:super_mnist}
\end{figure}

\begin{figure}[t]
     \centering
         \includegraphics[width=15.5cm]{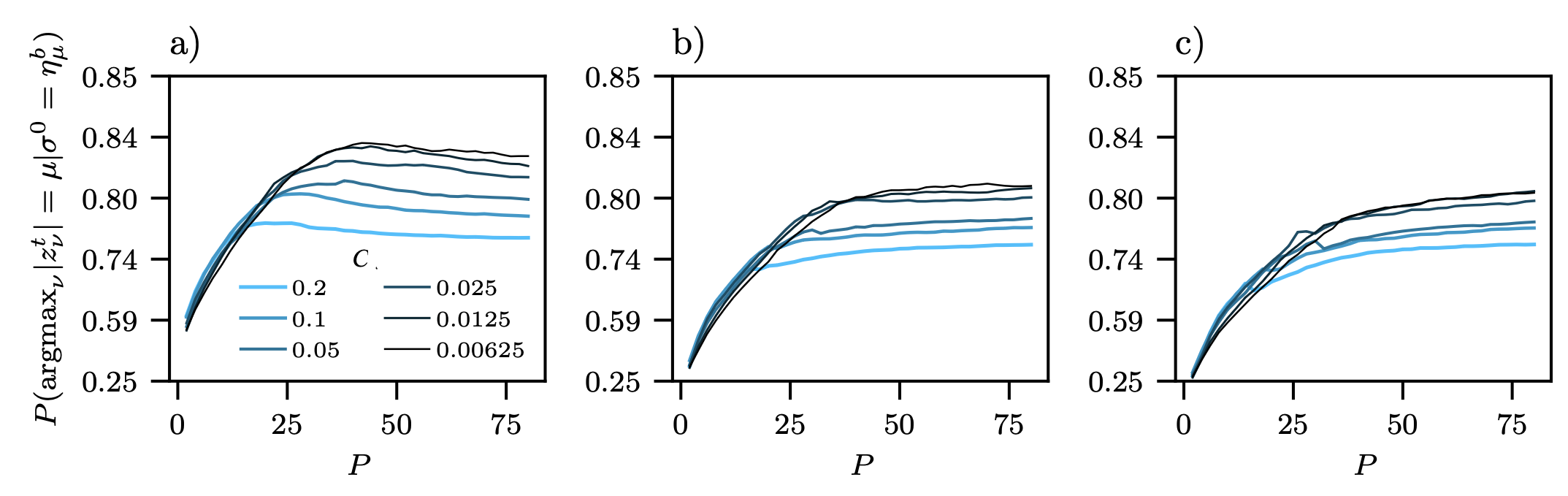}
         \label{fig:sup_normal_fmnist}
        \caption{We repeat the numerical analysis already done in Fig. \ref{fig:super_mnist} for the Fashion-MNIST dataset. Again, the different panels $a),b),c)$ corresponds to $1,2,3$ steps of parallel Plefka's dynamics. Again, the classification performance tends to increase with $P$ until a saturation accuracy level whose value grows as the number of examples increase ($C\to 0$).}
        \label{fig:super_fmnist}
\end{figure}

Now, for each class, we assess the related quality denoted as $r_\mu$ and obtained as follows: given $i$ and $\mu$, we count the number of positive $N_{i,\mu}^+$  and negative $N_{i,\mu}^-$  pixels over the class sample $ \{\zeta_i^{\mu,a}\}_{a=1...M}$, whence
\begin{equation}
    r_{\mu} = \dfrac{1}{N} \SOMMA{i=1}{N} \dfrac{|N_{i,\mu}^+-N_{i,\mu}^-|}{N_{i,\mu}^++N_{i,\mu}^-}.
\end{equation}
Then, recalling that in the random theory a number of examples bigger than $M\circaa r^{-2}$ guarantees the existence of a non-vanishing retrieval region, we can choose an arbitrary, finite value $C$ in $]0;1]$ and introduce for each $\mu=1, \hdots ,K$
\begin{equation}
    M_{\mu}= \dfrac{1}{C}\dfrac{1}{r_{\mu}^2 }, 
\end{equation}
that we expect being a sufficient number of items for correctly processing the $\mu$-th class.

To check the network performance in classifying these structured datasets we rely upon the integral expression of the partition function (see Eq. \ref{eq:partition}), namely we use the dual representation of the network in terms of a generalized RBM, as discussed in Remark \ref{dualEff}: in this representation the network has an input visible layer of size $N = 784$, corresponding to the number of pixels of each item and whose neurons are denoted as $\boldsymbol\sigma = (\sigma_1, \cdots, \sigma_N)$, and it also has a hidden layer of size $K = 10 $, corresponding to the overall number of classes and whose neurons are denoted as $\boldsymbol z = (z_1, \cdots , z_K )$. We fix the weights of the machine as in the grandmother-cell setting\footnote{The  weight associated to the link connecting the $i$-th visible neuron and the $\mu$-th hidden neuron is set directly, without training, equal to $\dfrac{1}{M} \sum_{a=1}^M \zeta_i^{\mu,a}$.}, namely by assigning to each hidden neuron a specific class, one per neuron. 
% NON FACCIAMO VERAMENTE TRAINING, MA FISSIAMO I PESI IN ACCORDO CON GRANDMOTHER CELL, GIUSTO?
% \green{SI}
%-, as this guarantees that -by marginalizing the partition function over the hidden neurons- the dense Hebbian network has the archetypes as candidate minima of its memory landscape.  
Finally,  just as a matter of practical convenience, coupled to this hidden layer, the network is endowed with an output visible layer of the same size $K$ and whose neurons are denoted as $\boldsymbol \pi = (\pi_1, \cdots, \pi_K)$ which is a softmax layer relating  $\pi_\mu$ with the probability that the input supplied to the network belongs to the $\mu$-th class.\\
Let us now specify the activation functions of the hidden and of the output neurons. As detailed in Appendix \ref{app:plefka} (see Eqs.~\eqref{eq:P_a} - \eqref{eq:P_b}), exploiting Plefka's dynamics equations 
``in tandem'' to make the system evolve, we get
\begin{align}
\left\langle z_{\mu} \right\rangle &=\frac{1}{\sqrt{  N^{P-1}}} \left(\SOMMA{i=1}{N} (\hat{\zeta}_M)_i^{\mu} \,\left\l\sigma_i\right\r\right)^{P/2},
     \label{eq:P_b_anticipo}
     \\
    \left\l \sigma_i \right\r&= \tanh{\left\{\bbt \dfrac{P}{2}\dfrac{1}{\sqrt{N^{P-1}}}\SOMMA{\mu=1}{K}\left[\left\langle z_{\mu} \right\rangle(\hat{\zeta}_M)_i^{\mu} \left(\SOMMA{j=1}{N} (\hat{\zeta}_M)_j^{\mu} \,\left\l \sigma_j \right\r\right)^{P/2-1}\right]\right\}},
    \label{eq:P_a_anticipo}
\end{align}
where $(\hat{\zeta}_M)_i^\mu:= \dfrac{1}{rM}\sum_{a=1
}^M \zeta_i^{\mu,a}$. Thus, we proceed as follows: we start from an initial configuration $\boldsymbol{\sigma}^{(0)}$ initializing the visible layer as $\boldsymbol\sigma^{(0)} = \boldsymbol{\zeta}^{\mu,b}$, where $\boldsymbol{\zeta}^{\mu,b}$  is an item sampled from the $\mu$-th class, then, using \eqref{eq:P_b_anticipo}, we evaluate the related $\l z_{\mu}^{(0)}\r$ for any $\mu$ and we use them in Eq. \eqref{eq:P_a_anticipo} to get $\l \sigma_i^{(0)} \r$ for any $i$, the latter is then used in Eq. \eqref{eq:P_b_anticipo} to get $\l z_{\mu}^{(1)}\r$, and we proceed this way, bouncing from \eqref{eq:P_a_anticipo} to \eqref{eq:P_b_anticipo}, for a few iterative steps. Finally, the accuracy of the classification process is checked  with the output layers which is determined as $\pi_{\mu} = \mathrm{softmax}(|\boldsymbol{z}|)$ and returns the classification probability for the input; the absolute value of the hidden neuron is carried out element-wise and it is meant to preserve the gauge invariance characterizing the model.
We  recall that in the softmax layer  a free parameter denoted with $\theta$ is applied, that is
\begin{equation}
    \pi_\mu(|\boldsymbol{z}^{(t)}|) = \dfrac{e^{\theta |z_{\mu}^{(t)}|}}{\sum\limits_{\nu=1}^{K} e^{\theta |z_{\nu}^{(t)}|}}\,,
    \label{eq:softmax}
\end{equation}
where the superscript $(t)$ refers to the number of steps made with Plefka's dynamics, and $\theta$ tunes the broadness of the distribution, playing a role similar to the temperature in the statistical mechanics framework. If we let $\theta\to\infty$ the softmax collapses to a delta function peaked at $\mathrm{argmax}_{\mu=1,\cdots,K}(|\boldsymbol{z}|)$.
\\
As $\theta\to\infty$, we identify the probability of correctly classifying a given example as 
\begin{equation}
    \mathcal P_+%=\pi_\mu(\boldsymbol\sigma ^{(0)}= \boldsymbol{\resub{\zeta}}^{\mu,b})
    = \mathbb P\left(\mathrm{argmax}_{\mu}( |\boldsymbol z^{(t)}|)=\mu | \boldsymbol\sigma^{(0)}=\boldsymbol{\zeta}^{\mu,b}\right)
\end{equation}
that is the probability of classifying an example $\zeta^{\mu,b}$ as correctly belonging to the class $\mu$-th class. This probability is shown by varying the order of interactions $P$ in Figs. \ref{fig:super_mnist} and \ref{fig:super_fmnist} respectively for MNist and Fashion MNist dataset. Instead with
\begin{equation}
    \mathcal P_-%=\pi_\mu(\boldsymbol\sigma ^{(0)}= \boldsymbol{\resub{\zeta}}^{\nu\neq\mu,b})
    = \mathbb P\left(\mathrm{argmax}_{\mu}( |\boldsymbol z^{(t)}|)=\mu | \boldsymbol\sigma^{(0)}\neq\boldsymbol{\zeta}^{\mu,b}\right)
\end{equation}
we identify the class-specific probability of misclassification, in other words it is the probability of classifying an example $\zeta^{\nu\neq\mu,b}$ as belonging to the $\mu$-th class. These probabilities are normalised as $\mathcal P_+ + 9\mathcal P_- = 1$ since in total there are 10 different classes for both datasets.

Let us discuss results regarding the performances of the network. \\
It is trivial to show that in the random case the classification is perfect. As far as the MNist concernes, 
as shown in Fig.~\ref{fig:super_mnist}, the classification is not flawless but it seems to improve as the order of interactions $P$ of the network increases. 
As for the Fashion MNist, as shown in Fig.~\ref{fig:super_fmnist}, the results are a bit worse that those on MNist dataset. This could be related to the presence of an elevated number of details which increase the structure of the images. 
%However, also in this dataset we can state that using one single step of dynamic gives us better performances. 

\subsection{Comparison between unsupervised and supervised regimes}\label{TacciDeThe}

\begin{figure}[t]
    \centering
    \includegraphics[width = 15.5cm]{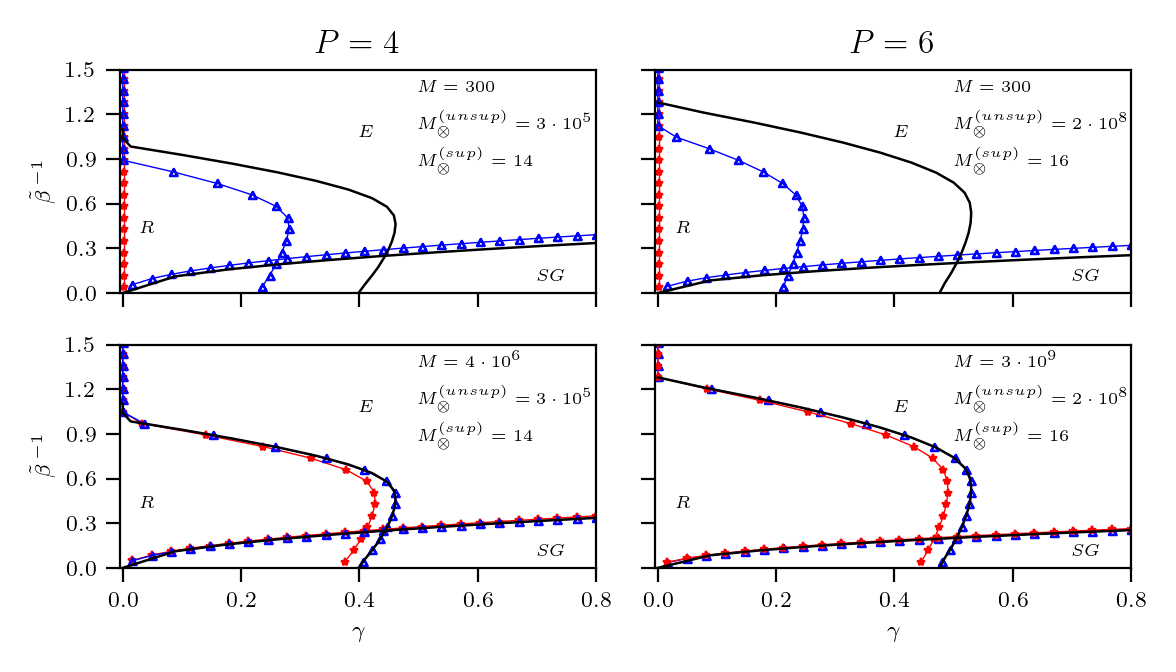}
    \caption{Comparison between supervised (blue triangles) and unsupervised (red stars) transition lines for $r=0.2$ considering $P=4$ (left) and $P=6$ (right). For each $P$, in the first row we choose $M=300$, on the other hand in the second row  we increase $M$ up to  $M=4\cdot 10^6$  and $M=3\cdot 10^9$ for $P=4$ and $P=6$ respectively. \eqref{eq:scaling_Ms}). Note that, fixed $P$, as long as $M>M_\otimes^{(unsup)}$  also $M>M_\otimes^{(sup)}$ thus, fixed the dataset quality, if the unsupervised model has a non-null retrieval region, also the supervised model has a non-null retrieval region (while, as shown in the first row, the contrary is not the case). 
    Even if both unsupervised and supervised models (for $M\to +\infty$)  approach the phase diagram of the standard Hebbian storage by  the P-spin Hopfield model (solid black line), the retrieval region of the supervised regime is always bigger than the unsupervised one, this is due to the fact that, denoting with $\hat{\gamma}$ the critical load for the standard dense Hopfield network, the noisy part of supervised model depends on $M_\otimes^{(sup)}= (\sqrt[P]{\hat\gamma/P\,})r^{-2}$ while in the unsupervised one it depends on $M_\otimes^{(unsup)}=(\hat\gamma/P)r^{-2P}$ and $M_\otimes^{(unsup)}>M_\otimes^{(sup)}$.} 
    \label{fig:sup_confronto}
\end{figure}

\begin{figure}[t]
    \centering
    
         \centering
         \includegraphics[width = 15cm]{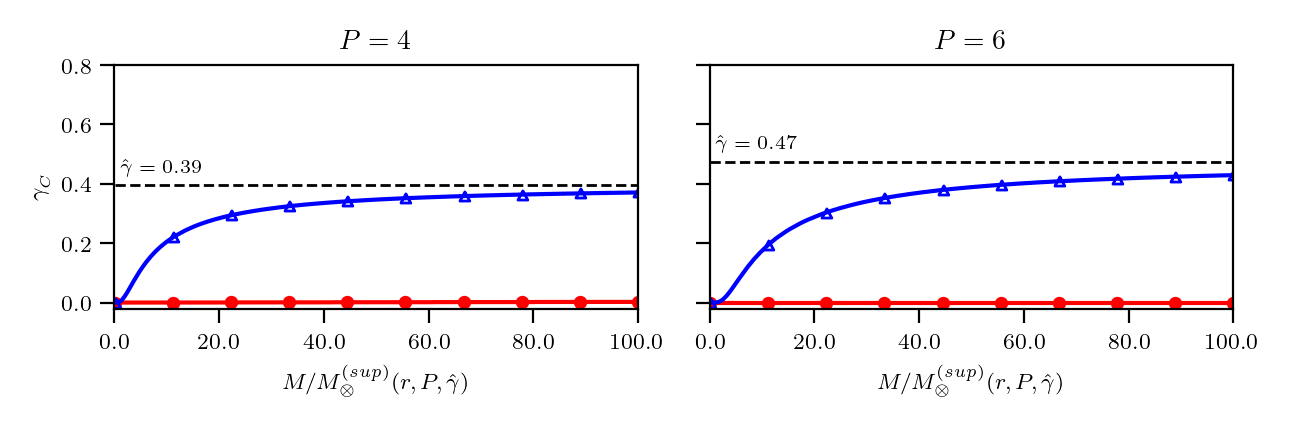}
    \caption{Comparison between supervised (blue triangles) and unsupervised (red stars) critical load $\gamma_C$  w.r.t. the ratio between the number of examples $M$ and the related lower bound $M_{\otimes}^{sup}(r,P,\hat\gamma)$ (see Eq. \eqref{eq:loadM}),  considering $P=4$ (left) and $P=6$ (right) in the null-temperature limit networks with $r=0.2$.  
    %the minimum number of examples to have a retrieval region scales as found in supervised regime from analytical computation in Sec. \eqref{sec:datalimit_sup} (namely $M > M_\otimes(r,P,\gamma)$). 
   In the figure  we also represent
    the critical load of the standard Hebbian storage, whose value is indicated by $\hat\gamma$ (dashed black line).
    Note that, fixed $P$ and the dataset quality, as long as $M/M_\otimes^{(sup)}(r,P,\hat\gamma)>1$ the supervised model has a non-null retrieval region, namely $\gamma_C>0$, while the unsupervised model, even for $M=100 M_\otimes^{(sup)}(r,P,\hat\gamma)$ has a null retrieval region. This behavior is related to the fact that the critical number of examples needed to the unsupervised model is much larger than the supervised one ($M_\otimes^{(sup)}(r,P,\hat\gamma)\ll  M_\otimes^{(unsup)}(r,P,\hat\gamma)$).
    }
    \label{fig:phase_confronto}
\end{figure}

The purpose of this section is to compare dense Hebbian neural networks in supervised and unsupervised \cite{unsup} settings. 
In a nutshell, these networks are both made of $N$ neurons whose interactions involve groups of $P$ units and the relate strengths read as  
\begin{align}
    J^{(unsup)}_{i_1i_2\hdots i_P} &\sim \frac{1}{ M N^{P-1}} \sum_{\mu=1}^K \sum_{a=1}^M \eta_{i_1}^{\mu,a} \cdots\:\eta_{i_P}^{\mu,a}, \label{J_unsup}\\
    J^{(sup)}_{i_1i_2\hdots i_P} &\sim \frac{1}{ M^P N^{P-1}} \sum_{\mu=1}^K \sum_{a_1, \hdots, a_P=1}^{M, \hdots, M} \eta_{i_1}^{\mu,a_1} \cdots\:\eta_{i_P}^{\mu,a_P} \label{J_sup}
\end{align}
where $\left\lbrace\boldsymbol \eta^{\mu}_a \right\rbrace_{a=1,...,M}^{\mu=1,...,K}$ represent the items of a dataset, corresponding to some perturbed versions of an archetype from a set $\left\lbrace \bm\xi^\mu\right\rbrace^{\mu=1, \hdots , K}$  which is never supplied to the network. In Eq.\eqref{J_sup}, the model can build its own archetypes starting from a set of examples already labeled in each class. Instead, 
in Eq.\eqref{J_unsup} there is no teacher that knows the labels and can thus cluster the examples archetype-wise, this is why the two generalizations are associated to, respectively, unsupervised and supervised protocols \cite{EmergencySN,prlmiriam}. 

In both settings there is a distinction between high- and low-load regimes and, in the former, by analytical investigations we reach the following expressions for the quenched statistical pressure:

%\footnotesize
\begin{align}
\mathcal{A}_{\gamma, M, r,\beta}^{(P, unsup)} 
     =& \mathbb{E}\left \{ \ln{2\cosh{\left[\tilde\beta \dfrac{P}{2}\n^{P-1}(1+\rho)^{P-1} \tilde{\eta}+Y \tilde\beta \sqrt{\gamma (1+\rho_{_P})
     \dfrac{P}{2}\q^{^{P-1}} }\right]}}\right\}+\notag
         \\
         &-\dfrac{\tilde\beta}{2}(P-1)(1+\rho)^{P}\n^P+  \gamma \dfrac{\tilde\beta^2  (1+\rho_P)}{4}\left(1-P\q^{P-1}+(P-1)\q^P\right)\,,
\\\notag \\
         \mathcal{A}_{\gamma, M,r, \beta}^{(P, sup)} 
     =&\mathbb{E}\left\lbrace{\ln\Bigg[2\cosh\left(\tilde\beta\dfrac{P}{2}(1+\rho)^{P-1}\n^{P-1}\tilde{\eta}+Y\tilde\beta\sqrt{   \gamma
         \dfrac{P}{2}(1+\rho)^P\q^{P-1}}\right)}\Bigg]\right\rbrace+ \notag
      \\
      &-\dfrac{\tilde\beta}{2}(P-1)(1+\rho)^{P}\n^{P}+\gamma \dfrac{\tilde\beta^2  (1+\rho)^{P}}{4}\left(1-P\q^{P-1}+(P-1)\q^P\right)
\end{align}
\normalsize
where $\mathbb{E}= \mathbb{E}_{\xi}\mathbb{E}_{(\eta|\xi)}\mathbb{E}_Y$, $\rho_P=\frac{1-r^{2P}}{r^{2P}M}$ and $\tilde{\beta}=\b (1+\rho)^{-P/2}$.

Moreover, the minimal number of examples the 
network needs in order to have a non-null retrieval region with a non-null load $\gamma$ in unsupervised \cite{unsup} and supervised (Sec.~\ref{sec:o_times}) settings are, respectively, 
\begin{eqnarray}
   M^{(sup)}_\otimes (r,P,\gamma)&=& \left( \dfrac{\gamma}{P}\right)^{1/P} \dfrac{1}{r^{\2}}  , \label{eq:M_sup} 
   \\
    M^{(unsup)}_\otimes (r,P,\gamma)&=&\dfrac{\gamma}{P} \dfrac{1}{r^{2P}} = \left[ M^{(sup)}_\otimes (r,P,\gamma)\right]^P. \label{eq:M_unsup}
\end{eqnarray}
Thus, since $r \in (0,1)$ and $\gamma > 0$, the number of necessary examples in the former regime is always smaller than in the latter where $P\geq 2$ even; in other words, the teacher, by grouping examples, ``simplifies'' the task. 
%: supervised learning allows to work even when small datasets are available,  still achieving by far better performances if compared with the unsupervised counterpart. 

%After the analytical description, let us focus on numerical results. 

In Fig. \ref{fig:sup_confronto} we compare the phase diagrams of dense networks achieved in unsupervised vs supervised settings. We notice that,  fixing the order of interaction $P$ and the dataset quality $r=0.2$, $M_{\otimes}^{(sup)}$ in Eq. \eqref{eq:M_sup} is not a sufficient amount of examples for the unsupervised setting to be able to retrieve (as shown in the first row of the panel), since it is much smaller than the amount necessary (see Eq. \eqref{eq:M_unsup})\footnote{In \cite{unsup} it is also discussed what happens when we do not respect this lower bound for unsupervised regime.}.

In order to preserve network's abilities, as  long as  $r\ll 1$,  the dataset size for the supervised and unsupervised model has to scale, respectively, according to the relations
\begin{eqnarray}
M^{(sup)}_\otimes\sim r^{-2}\;\;\;& \mathrm{and} &\;\;\;\;M^{(unsup)}_\otimes\sim r^{-2P}\,.
\label{eq:scaling_Ms}
\end{eqnarray}   
%This is equivalent setting $\rho\sim \mathcal{O}(1)$ and  $\rho_{_P}\sim \mathcal{O}(1)$. 
If we fix the dataset quality $r$, and we use a number of examples that leads to a non-null retrieval region for the unsupervised model, which allows us to have a non-null supervised retrieval zone (as illustrated in the second row of Fig. \ref{fig:sup_confronto}). However the opposite is not  the case; to deepen this, we show in Fig. \ref{fig:phase_confronto}, the behaviour of the critical load for different values of $M$ greater than $M^{(sup)}_\otimes$, supporting the fact that the presence of supervision in learning dramatically improves learning outcomes.

\section{Conclusion and outlooks}
\label{sec:conclusion}

In this paper we investigated the dense Hopfield neural network endowed with supervised Hebbian couplings.
The theory is analytically developed for the case of structureless datasets, but it is numerically corroborated also for MNist \cite{deng2012mnist} and Fashion MNist \cite{xiao2017fashion}.
We recall that the network is made of $N$ neurons that interact in groups of $P$ units with a strength encoded by the tensor $\boldsymbol J^{(sup)}$, whose generic entry (retaining only the highest order) reads as
\begin{equation}
    J^{(sup)}_{i_1\hdots i_P} \sim \frac{1}{ M^P N^{P-1}} \sum_{\mu=1}^K \sum_{a_1, \hdots, a_P=1}^{M, \hdots, M} \eta_{i_1}^{\mu,a_1} ...\:\eta_{i_P}^{\mu,a_P}
\end{equation}
where $\{\boldsymbol \eta^{\mu} \}_{a=1,...,M}^{\mu=1,...,K}$ is the dataset, made of $K$ subsets of $M$  examples, labelled by $a$, referred to $K$ unknown archetypes  $\{\boldsymbol \xi^{\mu} \}^{\mu=1,...,K}$. The quality of the dataset is determined by $r$, and
we recover the standard dense Hopfield model by setting $M=1$ and $r=1$:   
\begin{align}
    J^{(Hop)}_{i_1 \hdots i_P} \sim \dfrac{1}{N^{P-1}} \sum_{\mu=1}^K \xi_{i_1}^\mu \hdots \xi_{i_P}^\mu. 
\end{align}
%\cite{Baldi,Bovier,DenseRSB}.

Hereafter we summarize the main outcomes of our work. Some of them are shared with the unsupervised counterpart, whose analysis appeared in the twin paper \cite{unsup}: 
\begin{enumerate}
\item The high-load regime, where the number of stored archetypes grows as $K \sim N^{P-1}$, is still achievable  when the standard Hebbian coupling for storing is replaced by the Hebbian coupling for supervised learning. One can still introduce a load $\alpha_{P-1} = \lim_{N \to \infty} \dfrac{K}{N^{P-1}}$ and determine its critical dataset-dependent value beyond which a black-out scenario emerges.

\item There exists a (quality dependent $r$) threshold value $M^{(sup)}_{\otimes}$, for the  dataset size $M$ allowing archetype's learning (and thus its storage and successive retrieval), that  displays the  power-law scaling $ M^{(sup)}_\otimes\sim 1/r^{\2}$. The fact that this threshold value does not increase with $P$ is a main difference w.r.t. learning in an  unsupervised regime and highlights how learning by dense networks highly benefits from  teacher supervision.

\item In the low-load regime, the excess in the available resources can be exploited to mitigate the effects of noise affecting the synaptic tensor. 
In particular, the network can tolerate 
an additive noise that scales super-linearly with the network size (see also \cite{BarraPRLdetective} to deepen the role of redundancy in lowering the signal-to-noise threshold for pattern detection).
\end{enumerate}
As far as the computational and mathematical technicalities are concerned, in the supervised regime, the post-synaptic potential do not exhibit  any  Gaussian shape, thus universality of the quenched noise in spin glasses \cite{CarmonaWu, Genovese, Longo} does not apply  directly. However, we have exploited CLT in order to replace the potential as a product of Gaussians and, then, thanks to universality of quenched noise results in \cite{CarmonaWu}, we were able to deal with a single gaussian variable. 
%and standard techniques such as the replica trick and interpolation approaches do not work straightforwardly.  We adapted the latter by implementing estimates on \resub{the distributions of the post-synaptic potentials} and evaluated their averages by stability analysis; 
%This route is in full accordance with numerical simulations. 
\newline
Finally, we implemented a numerical MC scheme that employs Plefka's effective dynamics to avoid a numerical slowdown related to the evaluation of the synaptic tensors whose sizes get exponentially bigger as $P$ grows. It is worth noting that these technical extensions are of broad generality and can be applied to other kinds neural network.

Beyond full coherence  between  analytical predictions and numerical simulations in the random setting,  we also tested the whole theory on MNist and Fashion MNist datasets to observe that 
\begin{itemize}
    \item When the network lies in the retrieval region of its phase diagram, a single step of Plefka's dynamics is sufficient to accomplish the classification task. This further corroborates the fact that classification with our networks is computationally inexpensive and highlights a practical usage of the knowledge stemming from the phase diagram
    \item The higher the interaction degree $P$, the higher the performance of the machine, being measured by classification accuracy, until a plateau is reached. This highlights the advantage of using networks characterized by  higher-order interactions, especially equipped with supervised training as in this case the minimal dataset size for secure learning does not scale with $P$.
\end{itemize}

We conclude by pointing out that, despite working at the RS level of description, statistical mechanics techniques have made it possible to successfully employ a neural network at work with a classical computer vision task as image classification. This suggests that such a discipline can contribute to theoretical Artificial Intelligence by providing models and algorithms that can be fully understood both analytically and numerically, still accomplishing computational tasks of practical interest.
%; this can be particularly welcome in eXplainable AI (XAI). 
Further, by the knowledge of the phase diagram, we can set the machine in optimal working regions (in the space of the control parameters) and this can be particularly  useful for  optimization of algorithms and protocols,  en route towards SAI.

\appendix
\section*{Appendices}

\section{Proof of Proposition \ref{prop:highnoise}}
\label{app:proof}

In order to prove Proposition \ref{prop:highnoise} 
% we need to introduce one more observable, linked to the additional set of real variables introduced through the Hubbard-Stratonovich transformation (see Eq. \eqref{eq:partition} and Remark \ref{dualEff}): 
% \begin{equation}
%     p_{\resub{cd}}:=\dfrac{1}{K}\SOMMA{\mu=1}{K}z_{\mu}^{(\resub{c})}z_{\mu}^{(\resub{d})}\,,
% \end{equation}
% where $c,\ d$ are two replicas. 
% Moreover, we state that, under the replica-symmetry assumption, in the thermodynamic limit, also this variable self-averages around their mean values, i.e.
% \begin{equation}
%     	\lim_{N\to \infty} \langle (p_{12} - \bar p)^2 \rangle_t = 0 \Rightarrow \lim_{N \to \infty}  \langle p_{12} \rangle_t = \bar p.
% \end{equation}
% Next, 
we need to prove the following

\begin{lemma} 
The $t$ derivative of interpolating pressure is given by 
\begin{equation}
\label{eq:stream_proof}
\begin{array}{lll}
     \dfrac{d\mathcal A^{(P)}_{N,K,\rho,\beta}}{dt}=&\dfrac{\beta'  (1+\rho)^{P/2}}{2}\left(\l n_{1}^P\r_t - \dfrac{2 \psi }{\beta  (1+\rho)^{P/2}}\l n_{1}\r_t\right)
     +\dfrac{{\b}^2P!K}{8N^{P-1}}\left(1-\l q_{12}^P\r_t\right)-\dfrac{A^2}{2}\left(1-\l q_{12}\r_t\right).
\end{array}
\end{equation}
\end{lemma}

Since the computation is lengthy but not cumbersome we  omit it.

\begin{assumption}
\label{ass:RSassumption}
In the thermodynamic limit, under the RS approximation, the distribution of the generic order parameter $X$ is centred at its expectation value $\bar X$ with vanishing fluctuations. Thus, being $\Delta X = X-\bar X$,  in the thermodynamic limit, the following relation holds
\begin{equation}
    \l (\Delta X)^2\r_t \xrightarrow[N\to +\infty]{} 0\,.
\end{equation}
The RS approximation also implies that, in the thermodynamic limit, $\langle \Delta X \Delta Y \rangle_t \to 0$ for any generic pair of order parameters $X,Y$. Moreover, in the thermodynamic limit, we have $\langle (\Delta X)^k \rangle_t \rightarrow 0$ for $k \geq 2$.
	
	From now on, we omit the subindeces $t$.

\end{assumption}

Using the following relations computed by directly applying Newton's Binomial  

\begin{equation}
 \begin{array}{lll}
     \langle x^P \rangle -P\,\bar{x}^{P-1}\langle x \rangle =- (P-1) \bar{x}^{P} + \SOMMA{k=2}{P} \begin{pmatrix}P\\k\end{pmatrix} \langle (x-\bar{x})^k \rangle \bar{x}^{P-k}  \,,
% \\\\
% \langle p_{12}\,q_{12}^{P/2}\rangle= \q^{P/2}\langle p_{12} \rangle + \dfrac{P}{2}\, \p\q^{P/2-1}\langle q_{12} \rangle-\dfrac{P}{2}\,\p\q^{P/2} 
% \\\\
% \textcolor{white}{\langle p_{12}\,q_{12}^{P/2}\rangle=}+\SOMMA{k=1}{P/2} \begin{pmatrix}\frac{P}{2}\\k\end{pmatrix}\bar{q}^{P/2-k} \langle (p_{12}-\bar{p})(q_{12} - \bar{q})^k \rangle + \SOMMA{k=2}{P/2} \begin{pmatrix}\frac{P}{2}\\k\end{pmatrix} \bar{q}^{P/2-k} \bar{p}\langle (q_{12} - \bar{q})^k \rangle \,
\end{array}  
\label{eq:RS_pq_PotenzialiMiriam}
\end{equation}
% Now, using \eqref{eq:RS_pq_PotenzialiMiriam} relations, 
where $x$ is any order parameter and, if we fix the four constants appearing in \eqref{eq:stream_proof} as 
\begin{equation}
    \begin{array}{lll}
        \psi= \dfrac{P}{2}\beta'(1+\rho)^{P/2}\n^{P-1},&& A^2 =\dfrac{KP!}{2 N^{P-1}}\dfrac{\beta^2  }{2}P \q^{P-1} \,,
        % \\\\
        % B^2=\beta'\green{\dfrac{P!}{2}} N^{1-P/2}\q^{P/2}&&C=\beta'\green{\dfrac{P!}{2}} N^{1-P/2}(1-\q^{P/2}),
    \end{array}
\end{equation}
\normalsize
we can recast \eqref{eq:stream_proof} as 
\begin{equation}
\label{eq:dertfinite}
    \begin{array}{lll}
\dfrac{d\mathcal{A}^{(P)}_{N,K,\rho ,\beta}}{dt} &= -\dfrac{\beta'}{2}(1+\rho)^{P/2}(P-1)\n^{P}+\dfrac{\b\,^2}{4}\dfrac{P!K}{2N^{P-1}}(1-P\q^{P-1}+ (P-1)\q^P) + V_N(t)
    \end{array}
\end{equation}
where 
\begin{equation}
\label{eq:VN}
    \begin{array}{lll}
V_N(t)&= \dfrac{\beta'}{2}(1+\rho)^{P/2}\SOMMA{k=2}{P} \begin{pmatrix}P\\k\end{pmatrix} \langle (n_{1}-\bar{n})^k \rangle_t \bar{n}^{P-k} + \dfrac{\b\,^2}{4}\dfrac{P!K}{2N^{P-1}} \SOMMA{k=2}{P} \begin{pmatrix}P\\k\end{pmatrix} \langle (q_{12}-\bar{q})^k \rangle_t \bar{q}^{P-k}
    \end{array}
\end{equation}
\normalsize

We stress that, due to RS assumption, in the thermodynamic limit, the potential $V_N(t)$ vanishes and the derivative of the interpolating pressure w.r.t. $t$ becomes
\begin{equation}
\begin{array}{lll}
     \dfrac{d\mathcal{A}^{(P)}_{N,K, \rho ,\beta}}{dt}=&-\dfrac{\beta'}{2}(1+\rho)^{P/2}(P-1)\n^{P}+\alpha_{P-1}\dfrac{P!}{2}\dfrac{\b\,^2}{4}\left(1-P\q^{P-1}+ (P-1)\q^P\right) .
\end{array}
\label{eq:dert_TD}
\end{equation}
where we use the thermodynamic limit expression of the network load in \eqref{eq:carico_TDL}, namely $\lim\limits_{N\to \infty} \frac{K}{N^{P-1}}= \alpha_{P-1}$.

\begin{proof}
\textit{(of Proposition \ref{prop:highnoise})}
Let us start from finite size $N$ expression. We apply the Fundamental Theorem of Calculus: 
\begin{equation}
    \mathcal{A}_{N, K,\rho, \beta}^{(P)}=\mathcal{A}_{N, K,\rho, \beta}^{(P)}(t=1)=\mathcal{A}_{N, K,\rho, \beta }^{(P)}(t=0)+\int\limits_0^1\, \partial_s \mathcal{A}_{N, K,\rho, \beta }^{(P)}(s
)\Big|_{s=t}\,dt.
\label{eq:F_T_CalculusMiriam}
\end{equation}
\normalsize

We have already computed the derivative w.r.t. $t$ in Eq. \eqref{eq:dertfinite}. It only remains to compute the one-body term:
\begin{equation}
\begin{array}{lll}
      \mathcal{Z}_{N, K,\rho, \beta }(t=0)=&\SOMMA{\boldsymbol{\sigma}}{}\exp{\Bigg[\SOMMA{i=1}{N}\left(\dfrac{\psi}{2}  \dfrac{r^2}{\mathcal{R}}\tilde{\eta}+A\,J_i\right)\sigma_i\Bigg]} ,
\end{array}
\end{equation}
\normalsize
where  $\tilde{\eta}=\dfrac{1}{rM}\SOMMA{a=1}{M}\eta_i^{1,a}$. Using the definition of quenched statistical pressure \eqref{hop_GuerraAction} we have 
\begin{equation}
\begin{array}{lll}
       \mathcal{A}_{N, K,\rho, \beta}^{(P)}(t=0)&=\mathbb E\left\{\ln{}\left[{2\cosh\left(\dfrac{\psi}{2} \dfrac{r^2}{\mathcal{R}}\tilde{\eta}+A\,J\right)}\right]\right\}
      \\\\
      &=\mathbb{E}\left\lbrace\ln{}{2\cosh\left[\beta'\dfrac{P}{2}(1+\rho)^{P/2-1}\n^{P-1}\tilde{\eta}+J\sqrt{\dfrac{KP!}{4 N^{P-1}}\beta^2  \dfrac{P}{2} \q^{P-1}}\right]}\right\rbrace 
\end{array}
\label{onebody}
\end{equation}
\normalsize
where $\mathbb{E}= \mathbb{E}_\xi\mathbb{E}_{(\eta|\xi)}\mathbb{E}_Y$. Finally, putting  \eqref{onebody} and \eqref{eq:dertfinite} in \eqref{eq:F_T_CalculusMiriam}, we get
\begin{equation}
\begin{array}{lll}
      \mathcal{A}_{N, K,\rho, \beta }^{(P)}=&&\mathbb{E}\left\lbrace\ln{}{2\cosh\left[\beta'\dfrac{P}{2}(1+\rho)^{P/2-1}\n^{P-1}\tilde{\eta}+J\sqrt{\dfrac{KP!}{4 N^{P-1}}\beta^2  \dfrac{P}{2} \q^{P-1}}\right]}\right\rbrace + 
      \\\\
      && -\dfrac{\beta'}{2}(1+\rho)^{P/2}(P-1)\n^{P}+\dfrac{\b\,^2}{4}\dfrac{P!K}{2N^{P-1}} (1-P\q^{P-1}+(P-1)\q^P) +  \displaystyle\int_0^1 V_N(t) dt,
\end{array}
\label{eq:AMiriam_finiteRS}
\end{equation}

so, we reach the thesis considering that in the thermodynamic limit, as for $N\to\infty$, the potential $V_N$ vanishes and we can recast the load of the network as in  Eq. \eqref{eq:carico_TDL}, using the explicit expression of $\alpha_{P-1}= 2\gamma/P! $. 
\end{proof}

\begin{corollary}
For large datasets $M \gg 1$ in the high-storage regime, the replica-symmetric self-consistency equations \eqref{eq:n_selfRS_sup}-\eqref{eq:m_selfRS_sup} can be expressed as
\begin{align}
     \m&=\mathbb{E}\left[\tanh{g(\beta, Z, \mb)}\right], \label{eq:largedatasetm} \\
     \q&=\mathbb{E}\left[\tanh{}^2{g(\beta, Z, \mb)}\right], \label{eq:largedatasetq} \\
     \n&=(1+\rho)^{-1}\m+\tilde\beta\dfrac{P}{2}\rho(1+\rho)^{P-2}(1-\q)\n^{P-1},
     \label{eq:largedatasetn}
\end{align}
where $Z\sim\mathcal{N}(0,1)$ and
\begin{equation}
    g(\beta, Z, \mb)=\Tilde{\beta}\dfrac{P}{2}\m^{P-1}+Z \Tilde{\beta} \sqrt{\rho\left(\dfrac{P}{2}\m^{P-1}\right)^2 +\gamma \dfrac{P}{2}(1+\rho)^{P}\q^{P-1}}
    \label{eq:g_supervised_large_M_app}
\end{equation}
with $\mathbb{E}=\mathbb{E}_Z$ and $\beta'=\Tilde{\beta}(1+\rho)^{P/2}$.
\end{corollary}

\begin{proof}
For large datasets, if we use the CLT on $\etaM$ and knowing that
\begin{equation}
\begin{array}{lll}
      \mathbb{E}_{(\eta|\xi)}[\etaM] = \xi^1\,,&\;\;\;\;\;&  \mathbb{E}_{(\eta|\xi)}[(\etaM)^2]-\Big(\mathbb{E}_{(\eta|\xi)}[\etaM]\Big)^2 = \rho (\xi^1)^2,
\end{array}
\end{equation}
we replace $\etaM$ as 
\begin{equation}
    \etaM \sim \xi^1\left[1+Z\sqrt{\rho}\right]\,
    \label{eq:CLT_sup}
\end{equation}
where $Z$ is a standard Gaussian variable $Z \sim \mathcal{N}(0,1)$. Replacing Eq. \eqref{eq:CLT_sup} in the self-consistency equation for $\n$ \eqref{eq:n_selfRS_sup},  exploiting the parity of the hyperbolic tangent, we can explicitly compute the mean over $\xi$ 
\begin{equation}
\small
    \begin{array}{lll}
        \n&=&\dfrac{1}{(1+\rho)}\mathbb{E}_{\xi}\mathbb{E}_{Z}\mathbb{E}_{Y}\left\{\tanh{\left[\beta'\dfrac{P}{2}(1+\rho)^{P/2-1}\n^{P-1} \xi^1\left[1+Z\sqrt{\rho}\right]+Y\beta'\sqrt{   \gamma
         \dfrac{P}{2}\q^{P-1}}\right]} \xi^1\left[1+Z\sqrt{\rho}\right]\right\}
         \\\\
         &=&
         \dfrac{1}{(1+\rho)}\mathbb{E}_{Z}\mathbb{E}_{Y}\left\{\tanh{\left[\beta'\dfrac{P}{2}(1+\rho)^{P/2-1}\n^{P-1} \left[1+Z\sqrt{\rho}\right]+Y\beta'\sqrt{   \gamma
         \dfrac{P}{2}\q^{P-1}}\right]} \left[1+Z\sqrt{\rho}\right] \right\}\,.
    \end{array}
\end{equation}

Then, applying  Stein's lemma\footnote{This lemma, also known as Wick's theorem, applies to standard Gaussian variables, say $J \sim \mathcal N(0, 1)$, and states that, for a generic
function $f(J)$ for which the two expectations $\mathbb{E}\left( J f(J)\right)$ and $\mathbb{E}\left( \partial_J f(J)\right)$ both exist, then
\begin{align}
    \label{eqn:gaussianrelation2}
    \mathbb{E} \left( J f(J)\right)= \mathbb{E} \left( \frac{\partial f(J)}{\partial J}\right)\,.
\end{align}
} in order to recover the expression for $\m$ and $\q$, we get the large dataset equation for $\n$, i.e. Eq. \eqref{eq:largedatasetn}. 

We now replace this new expression for $\n$ in the argument of the other two self-consistency equations. We will use the relation 
\begin{align}
    \mathbb{E}_{\lambda,Y}[F(a+b \lambda+cY)]=\mathbb{E}_{_Z}[F(a+Z\sqrt{b^2+c^2})] \, ,
\end{align}
where $\lambda, Y, Z$ are i.i.d. Gaussian variables, $F(a+b\lambda+cY)= \tanh(a+b\lambda+c)$ and we pose $a={\beta}\dfrac{P}{2}\m^{P-1}{(1+\rho)^{-P/2}}$, $b=\beta '\dfrac{P}{2}\n^{^{P-1}}{(1+\rho)^{P/2}}$ and $c=\beta' \sqrt{\gamma\dfrac{P}{2}\q^{^{P-1}}\;}$. Doing so, we obtain
\begin{equation}
   g(\alpha_{P-1},\Tilde{\beta},Z)=\Tilde{\beta}\dfrac{P}{2}\m^{P-1}+Z \Tilde{\beta} \sqrt{\rho\left(\dfrac{P}{2}\m^{P-1}\right)^2 +\gamma\dfrac{P}{2}(1+\rho)^{P}\q^{P-1}} \ ,
\end{equation}
where we posed $\beta'=\Tilde{\beta}(1+\rho)^{P/2}$ and we also truncated\footnote{We notice that the second term right-side in \eqref{eq:largedatasetm}, as far as the retrieval zone concerns, $1-\q$ has small value, and $P$ is an even integer greater than $2$.} $\n = \m(1+\rho)^{-1}$.
\end{proof}

%It holds the following,
\begin{corollary}
The self-consistency equations in the large dataset assumption ($M\gg1$) and  vanishing noise ($\tilde\beta\to+\infty$) are 
\begin{equation}
    \begin{array}{lll}
         \m = \mathrm{erf}\left[\dfrac{P}{2}\dfrac{\m^{P-1}}{G}\right] \ ,
          \\\\
         G=\sqrt{2\left[\rho\left(\dfrac{P}{2}\m^{P-1}\right)^2 +\gamma \dfrac{P}{2}(1+\rho)^{P}\right]} \ ,
         \\\\
         \q=1.
    \end{array}
\end{equation}
\end{corollary}
\begin{proof}
We follow the same path of the unsupervised case in \cite{unsup}.
We start by assuming finite the limit
\begin{equation}
    \lim_{\tilde\beta\to\infty}\tilde\beta(1-\q) = D \in \mathbb{R}
\end{equation}
and we notice that, as $\Tilde{\beta} \rightarrow \infty$, we have $\q \rightarrow 1$.
As a consequence, the following reparametrization is found to be useful,
\begin{equation}
    \q=1-\dfrac{\delta q}{\tilde\beta}\;\;\;\;\mathrm{as}\;\;\;\;\tilde\beta\to \infty.
\end{equation}
We introduce the additional term $\beta x$ in the argument of the hyperbolic tangent ($g({\beta},Z,\m)$), thus
\begin{equation}
    \dfrac{\partial\m}{\partial x}=\tilde\beta\left[1-\left(1-\dfrac{\delta\q}{\tilde\beta}\right)\right]=\delta\q.
\end{equation}
Therefore, as $\Tilde{\beta}\to \infty$ and $x\to 0$, it yields
\begin{equation}
    \begin{array}{lll}
         \m = \mathbb{E}_z\left\{\mathrm{sign}\left[\dfrac{P}{2}\m^{P-1}+z\sqrt{\rho\left(\dfrac{P}{2}\m^{P-1}\right)^2 +\gamma\dfrac{P}{2}(1+\rho)^{P}}\right]\right\} \ ,
          \\\\
         \q\to 1 \ ;
    \end{array}
\end{equation}
where we have also used the relation
\begin{equation}
    \mathbb{E}_z\mathrm{sign}[A+Bz]=\mathrm{erf}\left[\dfrac{A}{\sqrt{2}B}\right] \ ,
\end{equation}
where $z$ is a Gaussian variable $\mathcal{N}(0,1)$, $A=\dfrac{P}{2}\m^{P-1}$ and $B=\sqrt{\rho\left(\dfrac{P}{2}\m^{P-1}\right)^2 +\gamma\dfrac{P}{2}(1+\rho)^{P}} \, $.
\end{proof}

\section{Evaluation of momenta of the effective post-synaptic potential}
\label{app:momenta}
In this section we want to compute the first and second momenta $\mu_1$ and $\mu_2$ of the effective post-synaptic potential introduced in Sec. \ref{subsec:S2N_sup}.

Let us start from $\mu_1$:
%\footnotesize
\begin{eqnarray}
\mu_1&\coloneqq& \mathbb{E}_{\xi}\mathbb{E}_{(\eta|\xi)}\left[\xi_{i_1}^{1}h_{i_1}^{(1)}(\boldsymbol{\xi}^{1}\vert\bm\eta)\right]=\dfrac{1}{N^{P-1}M^P \mathcal{R}^{P/2}}\SOMMA{\mu=1}{K}\SOMMA{(i_2,\cdots,i_P)}{N,\cdots,N}\SOMMA{a_1,\cdots,a_P}{M,\cdots,M}\mathbb{E}_{\xi}\mathbb{E}_{(\eta|\xi)}\left[\xi_{i_1}^1\eta^{\mu, a_1}_{i_1}\hdots\xi_{i_P}^1\eta^{\mu,a_P}_{i_P}\right]\notag
\\
&=&\dfrac{1}{N^{P-1}M^P \mathcal{R}^{P/2}}\SOMMA{\mu=1}{K}\SOMMA{(i_2,\cdots,i_P)}{N,\cdots,N}\SOMMA{a_1,\cdots,a_P}{M,\cdots,M}\mathbb{E}_{\xi}\mathbb{E}_{\resub{(\eta|\xi)}}\left[\xi_{i_1}^1\eta^{\mu, a_1}_{i_1}\right]\hdots\mathbb{E}_{\xi}\mathbb{E}_{(\eta|\xi)}\left[\xi_{i_P}^1\eta^{\mu,a_P}_{i_P}\right]\notag
\\
&=&\dfrac{1}{N^{P-1}M^P \mathcal{R}^{P/2}}\SOMMA{\mu=1}{K}\SOMMA{(i_2,\cdots,i_P)}{N,\cdots,N}\SOMMA{a_1,\cdots,a_P}{M,\cdots,M}\mathbb{E}_{\xi}\left[r\xi_{i_1}^1\xi^{\mu}_{i_1}\right]\hdots\mathbb{E}_{\xi}\left[r\xi_{i_P}^1\xi^{\mu}_{i_P}\right].\notag
\end{eqnarray}
\normalsize
Since $\mathbb{E}_{\xi}[\xi_i^{\mu}]=0$ the only non-zero terms are the ones with $\mu=1$:
%\footnotesize
\begin{eqnarray}
\mu_1&\coloneqq& \dfrac{1}{N^{P-1}M^P \mathcal{R}^{P/2}}\SOMMA{\mu=1}{K}\SOMMA{(i_2,\cdots,i_P)}{N,\cdots,N}\SOMMA{a_1,\cdots,a_P}{M,\cdots,M}\mathbb{E}_{\xi}\left[r\left(\xi_{i_1}^1\right)^2\right]\hdots\mathbb{E}_{\xi}\left[r\left(\xi_{i_P}^1\right)^2\right]
\\
&=&\dfrac{1}{ N^{P-1}M^P \mathcal{R}^{P/2}}N^{P-1}M^{P}r^{P}=\dfrac{1}{2 }\left(\dfrac{r^{2}}{\mathcal{R}}\right)^{P/2}=\left(1+\rho\right)^{-P/2}\,.\notag
\end{eqnarray}
\normalsize
Moving on, we start the computation of $\mu_2$:
\begin{align}
\begin{array}{llll}
     &&\mu_2\coloneqq \mathbb{E}_{\xi}\mathbb{E}_{(\eta|\xi)}\left[\lbrace h_{i_1}^{(1)}(\boldsymbol{\xi}^{1}\vert\bm\eta)\rbrace ^2\right]
     \\\\
     &=&\dfrac{1}{N^{2P-2}M^{2P}\mathcal{R}^P}\SOMMA{\mu,\nu=1}{K}\SOMMA{a_1=1}{M}\SOMMA{b_1=1}{M}\SOMMA{ a_2, \hdots a_P}{M, \hdots, M} \SOMMA{b_2, \hdots, b_P}{M,\hdots, M}\SOMMA{(i_2,\cdots,i_P)}{N,\hdots,N}\SOMMA{(j_2,\cdots,j_P)}{N,\hdots,N}\mathbb{E}_{\xi}\mathbb{E}_{(\eta|\xi)}\left[\eta^{\mu, a_1}_{i_1}\eta^{\nu, b_1}_{i_1}\right.
     \\\\
     &&\left. \left(\eta^{\mu, a_2}_{i_2}\eta^{\nu, b_2}_{j_2}...\eta^{\mu,a_P}_{i_P}\eta^{\nu,b_P}_{j_P}\right) \left(\xi^{1}_{i_2}\xi^{1}_{j_2}...\xi^{1}_{i_P}\xi^{1}_{j_P}\right)\right].
\end{array}
\end{align}
Due to $\mathbb{E}_{(\eta|\xi)}[\eta_{i_1}^{\mu}\eta_{i_1}^{\nu}]=\left(r\xi^{\mu}_{i_1}\right)\left(r\xi^{\nu}_{i_1}\right)$ and $\mathbb{E}_{\xi}[\xi_{i_1}^{\mu}\xi_{i_1}^{\nu}]=\delta^{\mu\nu}$, the only non-zero terms are the ones with $\mu=\nu$:
\begin{align}
\begin{array}{llll}
     \mu_2&\coloneqq& \mathbb{E}_{\xi}\mathbb{E}_{(\eta|\xi)}\left[\lbrace h_{i_1}^{(1)}(\boldsymbol{\xi}^{1}\vert\bm\eta)\rbrace ^2\right]\\\\
     &=&\dfrac{1}{N^{2P-2}M^{2P}\mathcal{R}^P}\SOMMA{\mu=1}{K}\SOMMA{a_1=1}{M}\SOMMA{b_1=1}{M}\SOMMA{ a_2, \hdots a_P}{M, \hdots, M} \SOMMA{b_2, \hdots, b_P}{M,\hdots, M}\SOMMA{(i_2,\cdots,i_P)}{N,\hdots,N}\SOMMA{(j_2,\cdots,j_P)}{N,\hdots,N}\mathbb{E}_{\xi}\mathbb{E}_{(\eta|\xi)}\left[\eta^{\mu, a_1}_{i_1}\eta^{\mu, b_1}_{i_1}\right.
     \\\\
     &&\left. \left(\eta^{\mu, a_2}_{i_2}\eta^{\mu, b_2}_{j_2}...\eta^{\mu,a_P}_{i_P}\eta^{\mu,b_P}_{j_P}\right)\left(\xi^{1}_{i_2}\xi^{1}_{j_2}...\xi^{1}_{i_P}\xi^{1}_{j_P}\right)\right]
     \\\\
     &=&A_{\mu=1}+B_{\mu>1},
\end{array}
\end{align}
namely we will analyze separately the case for $\mu=1$ ($A_{\mu=1}$) and $\mu>1$ ($B_{\mu>1}$).
\begin{align}
\begin{array}{llll}
    A_{\mu=1}&=&\dfrac{1}{N^{2P-2}M^{2P}\mathcal{R}^P}\SOMMA{ a_1, \hdots a_P}{M, \hdots, M} \SOMMA{b_1, \hdots, b_P}{M,\hdots, M}\SOMMA{(i_2,\cdots,i_P)}{N,\hdots,N}\SOMMA{(j_2,\cdots,j_P)}{N,\hdots,N}\mathbb{E}_{\xi}\mathbb{E}_{(\eta|\xi)}\left[\eta^{1, a_1}_{i_1}\eta^{1, b_1}_{i_1}\left(\eta^{1, a_2}_{i_2}\eta^{1, b_2}_{j_2}...\eta^{1,a_P}_{i_P}\eta^{1,b_P}_{j_P}\right)\right.
     \\\\
     &&\left. \left(\xi^{1}_{i_2}\xi^{1}_{j_2}...\xi^{1}_{i_P}\xi^{1}_{j_P}\right)\right]
    \\\\
    &=&\dfrac{1}{N^{2P-2}M^{2P}\mathcal{R}^P}\SOMMA{a_1=1}{M}\SOMMA{ a_2, \hdots a_P}{M, \hdots, M} \SOMMA{b_2, \hdots, b_P}{M,\hdots, M}\SOMMA{(i_2,\cdots,i_P)}{N,\hdots,N}\SOMMA{(j_2,\cdots,j_P)}{N,\hdots,N}\mathbb{E}_{\xi}\mathbb{E}_{(\eta|\xi)}\left[\left(\eta^{1, a_2}_{i_2}\eta^{1, b_2}_{j_2}...\eta^{1,a_P}_{i_P}\eta^{1,b_P}_{j_P}\right)\right.
     \\\\
     &&\left. (\eta^{1, a_1}_{i_1})^2\left(\xi^{1}_{i_2}\xi^{1}_{j_2}...\xi^{1}_{i_P}\xi^{1}_{j_P}\right)\right]
     \\\\
    &+&\dfrac{1}{N^{2P-2}M^{2P}\mathcal{R}^P}\SOMMA{a_1=1}{M}\SOMMA{\underset{(b_1\neq a_1=1)}{b_1=1}}{M}\SOMMA{ a_2, \hdots a_P}{M, \hdots, M} \SOMMA{b_2, \hdots, b_P}{M,\hdots, M}\SOMMA{(i_2,\cdots,i_P)}{N,\hdots,N}\SOMMA{(j_2,\cdots,j_P)}{N,\hdots,N}\mathbb{E}_{\xi}\mathbb{E}_{(\eta|\xi)}\left[\eta^{1, a_1}_{i_1}\eta^{1, b_1}_{i_1}\right.
     \\\\
     &&\left. \left(\eta^{1, a_2}_{i_2}\eta^{1, b_2}_{j_2}...\eta^{1,a_P}_{i_P}\eta^{1,b_P}_{j_P}\right)\left(\xi^{1}_{i_2}\xi^{1}_{j_2}...\xi^{1}_{i_P}\xi^{1}_{j_P}\right)\right]
    \\\\
    &=&\dfrac{1}{N^{2P-2}M^{2P}\mathcal{R}^P}\SOMMA{a_1=1}{M}\SOMMA{ a_2, \hdots a_P}{M, \hdots, M} \SOMMA{b_2, \hdots, b_P}{M,\hdots, M}\SOMMA{(i_2,\cdots,i_P)}{N,\hdots,N}\SOMMA{(j_2,\cdots,j_P)}{N,\hdots,N}\mathbb{E}_{\xi}\left[(\xi^{1}_{i_1})^2r^{2P-2}\left(\xi^{1}_{i_2}\xi^{1}_{j_2}...\xi^{1}_{i_P}\xi^{1}_{j_P}\right)^2\right]
    \\\\
    &+&\dfrac{1}{N^{2P-2}M^{2P}\mathcal{R}^P}\SOMMA{a_1=1}{M}\SOMMA{\underset{(b_1\neq a_1=1)}{b_1=1}}{M}\SOMMA{ a_2, \hdots a_P}{M, \hdots, M} \SOMMA{b_2, \hdots, b_P}{M,\hdots, M}\SOMMA{(i_2,\cdots,i_P)}{N,\hdots,N}\SOMMA{(j_2,\cdots,j_P)}{N,\hdots,N}\mathbb{E}_{\xi}\left[r^{2P}\left(\xi^{1}_{i_1}\xi^{1}_{i_2}\xi^{1}_{j_2}...\xi^{1}_{i_P}\xi^{1}_{j_P}\right)^2\right]
    \\\\
    &=&\dfrac{r^{2P}}{\mathcal{R}^P}\left(1+\dfrac{1-r^2}{Mr^2}\right)=\left(1+\rho\right)^{1-P}. 
\end{array}
\end{align}
As for the latter
\begin{align}
\begin{array}{llll}
     B_{\mu>1}&=&\dfrac{1}{N^{2P-2}M^{2P}\mathcal{R}^P}\SOMMA{\mu>1}{K}\SOMMA{ a_2, \hdots a_P}{M, \hdots, M} \SOMMA{b_2, \hdots, b_P}{M,\hdots, M}\SOMMA{(i_2,\cdots,i_P)}{N,\hdots,N}\SOMMA{(j_2,\cdots,j_P)}{N,\hdots,N}\mathbb{E}_{\xi}\mathbb{E}_{(\eta|\xi)}\left\{\left(\SOMMA{a_1=1}{M}\left(\eta^{\mu, a_1}_{i_1}\right)^2\right.\right.
     \\\\
     && \left. +\SOMMA{a_1=1}{M}\SOMMA{\underset{(b_1\neq a_1)}{b_1=1}}{M}\eta^{\mu, a_1}_{i_1}\eta^{\mu, b_1}_{i_1}\right)+ \left(\eta^{\mu, a_2}_{i_2}\eta^{\mu, b_2}_{j_2}...\eta^{\mu,a_P}_{i_P}\eta^{\mu,b_P}_{j_P}\right)\left(\xi^{1}_{i_2}\xi^{1}_{j_2}...\xi^{1}_{i_P}\xi^{1}_{j_P}\right)\Bigg\}
     \end{array}
\end{align}
\begin{align}
\begin{array}{llll}
      B_{\mu>1}&=&\dfrac{1}{N^{2P-2}M^{2P}\mathcal{R}^P}\SOMMA{\mu>1}{K}\SOMMA{ a_2, \hdots a_P}{M, \hdots, M} \SOMMA{b_2, \hdots, b_P}{M,\hdots, M}\SOMMA{(i_2,\cdots,i_P)}{N,\hdots,N}\SOMMA{(j_2,\cdots,j_P)}{N,\hdots,N}\mathbb{E}_{\xi}\left\{\left(\SOMMA{a_1=1}{M}\left(\xi^{\mu}_{i_1}\right)^2+\SOMMA{a_1=1}{M}\SOMMA{\underset{(b_1\neq a_1)}{b_1=1}}{M}\left(r\xi^{\mu}_{i_1}\right)^2\right)\cdot \right.
     \\\\
     && \cdot r^{2P-2}\left(\xi^{\mu}_{i_2}\xi^{\mu}_{j_2}...\xi^{\mu}_{i_P}\xi^{\mu}_{j_P}\right)\left(\xi^{1}_{i_2}\xi^{1}_{j_2}...\xi^{1}_{i_P}\xi^{1}_{j_P}\right)\Bigg\},
\end{array}
\label{eq:amu}
\end{align}
the only non-zero terms are the ones in which the sum over $i$ and $j$ will be equal in pairs:
\begin{align}
\begin{array}{llll}
     B_{\mu>1}&=&\dfrac{(P-1)!}{N^{2P-2}M^{2P}\mathcal{R}^P}\SOMMA{\mu>1}{K}\SOMMA{ a_2, \hdots a_P}{M, \hdots, M} \SOMMA{b_2, \hdots, b_P}{M,\hdots, M}\SOMMA{(i_2,\cdots,i_P)}{N,\hdots,M}\mathbb{E}_{\xi}\left\{\left(\SOMMA{a=1}{M}\left(\xi^{\mu}_{i_1}\right)^2+\SOMMA{a=1}{M}\SOMMA{\underset{(b\neq a)}{b=1}}{M}\left(r\xi^{\mu}_{i_1}\right)^2\right)\right.\notag \\
     &&\left.\cdot r^{2P-2}\left(\xi^{\mu}_{i_2}\xi^{1}_{i_2}...\xi^{\mu}_{i_P}\xi^{1}_{i_P}\right)^2\right\}
     \\\\
     &=&\dfrac{(P-1)!}{N^{2P-2}M^{2P}\mathcal{R}^P}K M^{2P-2}N^{P-1}\left\{\left[M+M(M-1)\right]r^{2P-2}\right\}.
\end{array}
\label{termine2}
\end{align}
To clarify the discussion, we now focus on
\begin{align}
     \dfrac{1}{M^2}\SOMMA{a,b=1}{M}\mathbb{E}_{\xi}\mathbb{E}_{(\eta|\xi)}\left[\eta^{\mu,a}_{i}\eta^{\mu,b}_{i}\right]&=\dfrac{1}{M^2}\SOMMA{a=1}{M}\mathbb{E}_{\xi}\mathbb{E}_{(\eta|\xi)}\left[\left(\eta^{\mu,a}_{i}\right)^2\right]+\dfrac{1}{M^2}\SOMMA{a=1}{M}\SOMMA{\underset{(b\neq a)}{b=1}}{M}\mathbb{E}_{\xi}\mathbb{E}_{(\eta|\xi)}\left[\eta^{\mu,a}_{i}\eta^{\mu,b}_{i}\right] \notag 
     \\
     &=\dfrac{1}{M^2}\SOMMA{a=1}{M}\mathbb{E}_{\xi}\left[\left(\xi^{\mu}_{i}\right)^2\right]+\dfrac{1}{M^2}\SOMMA{a=1}{M}\SOMMA{\underset{(b\neq a=1)}{b=1}}{M}\mathbb{E}_{\xi}\left[\left(r\xi_i^\mu\right)^2\right]  
     \\
     &=\dfrac{1}{M^2}\left(M+M(M-1)r^2\right)=r^2\left(1+\rho\right).\notag
     \label{termine}
\end{align}
Inserting Eq. \eqref{termine} in Eq. \eqref{termine2} we have
\begin{align}
\begin{array}{llll}
     B_{\mu>1}=\dfrac{ (P-1)!}{N^{P-1}\mathcal{R}^P}K r^{2P}(1+\rho)^P=\dfrac{ (P-1)!}{N^{P-1}}K;
\end{array}
\end{align}
if we set $K=\alpha_{P-1}{N^{P-1}}$ we get
\begin{align}
\begin{array}{llll}
     B_{\mu>1}=(P-1)!\alpha_{P-1}.
\end{array}
     \label{eq:bmu}
\end{align}

Putting together Eq. \eqref{eq:amu} and Eq. \eqref{eq:bmu} we reach the expression of $\mu_2$.

\section{Plefka's Expansion on Gibbs potential}
\label{app:plefka}
The purpose of this section is to follow the path from \cite{Plefka1}, namely the computation of the expansion of the Gibbs potential. We use this in the computational part of our work to produce the figures in the main text. We refer to \cite{unsup} for the discussion concerning dense Hebbian neural networks and dense unsupervised Hebbian networks. Here, we focus on the supervised regime.
 
The Hamiltonian of the described system is 
\begin{align}
        \mathcal H_{N,K, M,r}^{(P)}(\bm \sigma \vert \bm \eta; \bm z)= - \frac{1}{\sqrt{M^P \mathcal{R}^{P/2} N^{P-1}}} \sum_{\mu=1}^K  \sum_{a_1, \hdots a_{P/2}=1}^M  \sum_{i_1, \hdots , i_{P/2}=1}^N \eta_{i_1}^{\mu,a_1}\cdots \eta_{i_{P/2}}^{\mu, a_{P/2}}\sigma_{i_1} \hdots \sigma_{i_{P/2}}z_\mu
\end{align}
which is the interaction part of the Hamiltonian in the integral representation of the partition function of dense Hebbian supervised network provided in eq. \eqref{eq:partition}.

In order to proceed to our computations, we need to introduce a parameter $\NotAlpha$ and define an Hamiltonian $\mathcal H_{N,K,M,r}^{(P)}(\bm \sigma \vert \bm \eta; \bm z, \bm h, \bm {\tilde{h}}, \NotAlpha)$ that reads as 
\begin{align}
    \mathcal H_{N,K,M,r}^{(P)}(\bm \sigma \vert \bm \eta; \bm z, \bm h, \bm {\tilde{h}}, \NotAlpha)= \NotAlpha  \mathcal H_{N,K,M,r}^{(P)}(\bm\sigma|\bm \eta; \bm z)- \frac{1}{\sqrt{M^P \mathcal{R}^{P/2} N^{P-1}}}\sum_{i=1}^N h_i \sigma_i - \sum_{\mu=1}^K \tilde h_{\mu} z_\mu;
    \label{Halpha}
\end{align}
we stress that for $\NotAlpha=0$ we have an Hamiltonian with non-interacting units, whereas for $\NotAlpha=1$ we have one with full-interacting units; $\bm h$ and $\bm \tilde h$ are sets of external fields which play a role on $\bm \sigma$ and $\bm z$ respectively. 

Using \eqref{Halpha} we write the partition function as 
\begin{align}
    \mathcal Z^{(P)}_{ N,K, M,r,\beta}(\NotAlpha)=\sum_{\sigma} \int \prod_{\mu=1}^K &dz_\mu \sqrt{\frac{\bbt }{2\pi}} \exp \left( - \frac{\bbt }{2} \sum_{\mu=1}^K z_\mu^2 \right. \notag \\
    &- \dfrac{\NotAlpha \bbt }{\sqrt{M^P r^{P} N^{P-1}}} \sum_{\mu=1}^K \sum_{i_1, \hdots, i_{P/2}}^{N, \hdots, N}\sum_{ a_1, \hdots, a_{P/2}}^{M, \hdots, M} \eta_{i_{1}}^{\mu,a_{1}} \cdots \eta_{i_{P/2}}^{\mu,a_{P/2}} \sigma_{i_1} \hdots \sigma_{i_{P/2}} z_\mu  \notag \\
    &\left.+\bbt   \sum_{\mu=1}^K \tilde h_{\mu} z_\mu +\dfrac{\bbt }{\sqrt{M^P r^{P} N^{P-1} }} \sum_{i=1}^N h_i \sigma_i \right),
\end{align}
where $\tilde{\beta}$ is defined as in Eq. \eqref{eq:tildebeta}.

Let us consider the Gibbs potential of this model, which is the Legendre transformation of the free energy constrained w.r.t. the magnetizations $m_i= \langle \si \rangle$ and $\langle z_\mu \rangle$ averaged w.r.t. Boltzmann distribution $\propto \exp^{-\beta H(\NotAlpha)}$: 
\begin{align}
\label{gibbsPotential_dense}
     %\mathcal G(\NotAlpha, \beta, {m_i}) 
     % \mathcal G(\NotAlpha ) 
    \mathcal G_{N,K,M,r,\beta}^{(P)}( \bm \eta ;\bm z,\bm h,  \tilde{ \bm h}, \NotAlpha )
     = -\frac{1}{\bbt}\ln  \mathcal Z^{(P)}_{ N,K, M,r,\beta}(\NotAlpha) +   \sum_{\mu=1}^K \tilde h_{\mu} \langle z_\mu \rangle + \dfrac{1}{\sqrt{M^P r^{P} N^{P-1}}}\sum_{i=1}^N h_i m_i.
\end{align}

Now we want to expand the Gibbs potential around $\NotAlpha=0$, namely 
\begin{align}
    \mathcal G_{N,K,M,r,\beta}^{(P)}( \bm \eta; \bm z,\bm h,  \tilde{ \bm h}, \NotAlpha)= \mathcal G_{N,K,M,r,\beta}^{(P)}(\bm \eta; \bm z,\bm h,  \tilde{ \bm h}, \NotAlpha=0 ) + \sum_{n=1}^\infty \frac{\NotAlpha^n \mathcal G^{(n)}}{n!} \ \ \ \  \mathcal G^{(n)}= \left. \frac{\partial^n  \mathcal G^{(P)}_{N,K,M r,\beta}(\NotAlpha)}{\partial \NotAlpha^n} \right \vert_{\NotAlpha=0}, 
     \label{eq:gibbexp}
\end{align}
and stop the expansion at the first order. 

All we need is the non-interacting Gibbs potential $\mathcal G_{N,K,M,r,\beta}^{(P)}(\bm \eta;\bm z,\bm h,  \tilde{ \bm h}, \NotAlpha=0 )$ and the first derivative w.r.t. $\NotAlpha$. Let us start from $\mathcal G_{N,K,M,r,\beta}^{(P)}(\bm \eta; \bm z,\bm h,  \tilde{ \bm h}, \NotAlpha=0 )$ (which from now is denoted $\mathcal{G}(0)$).

\begin{align}
     \mathcal G(0)=& -\frac{N}{\bbt} \log 2 -\frac{1}{\bbt} \sum_{i=1}^N \log \cosh \left(\dfrac{\bbt }{\sqrt{M^P r^{P}N^{P-1}}} h_i\right) +\sum_{\mu=1}^K \tilde h_{\mu} \langle z_\mu \rangle  \notag \\
     &+ \dfrac{1}{\sqrt{M^P r^{P}N^{P-1}}}\sum_{i=1}^N h_i \langle \sigma_i \rangle -\frac{1}{2} \sum_{\mu=1}^K \tilde h_{\mu}^2.
     \label{eq:G0}
\end{align}

Now, if we extremize Eq. \eqref{eq:G0} w.r.t. $h_i$ and $\tilde{ h}_\mu$ we find 

\begin{align}
\label{eq:hi}
    h_i&= \frac{\sqrt{M^P r^{P}N^{P-1}}}{\bbt } {\tanh^{-1}}(m_i) \Longrightarrow h_i=\dfrac{\sqrt{M^P r^{P}N^{P-1}}}{2\bbt }\ln\left(\dfrac{1+m_i}{1-m_i}\right),\\
    \label{eq_hmu}
    \tilde h_\mu &= \langle z_\mu \rangle.
\end{align}

The last step is to insert the expressions \eqref{eq:hi} and \eqref{eq_hmu} in the non-interacting Gibbs potential:
\begin{align}
    \mathcal G(0)=& -\frac{N}{\bbt} \log 2 +\frac{1}{2\bbt} \sum_{i=1}^N \left[(1-m_i) \log(1-m_i) + (1+m_i) \log (1+m_i)\right]+\frac{1}{2} \sum_{\mu=1}^K  \langle z_\mu \rangle^2.
\end{align}
%From the expression of $\mathcal G(0)$ we notice that $h_i$ and $\tilde h_{\mu}$ are the same in \eqref{eq:hi} and \eqref{eq_hmu}.
Concerning the first order derivative w.r.t. $\NotAlpha$, we have 
\begin{align}
  \mathcal{G}(1)= -\dfrac{1}{\sqrt{M^P r^{P}N^{P-1}}}\sum_{\mu=1}^K\sum_{i_1, \hdots, i_{P/2}}^{N,\hdots, N} \sum_{a_1, \hdots, a_{P/2}}^{M, \hdots, M} \eta_{i_1 }^{\mu,  a_1}\cdots\eta_{i_{P/2} }^{\mu,  a_{P/2}}  \langle z_\mu \rangle m_{i_1} \hdots m_{i_{P/2}}.
\end{align}
To sum up, we can write the first order Gibbs  potential \eqref{gibbsPotential_dense} as 
\begin{align}
   \mathcal G_{N,K,\beta}^{(P)}(\bm \eta; \bm z,\bm h,  \tilde{ \bm h}, \NotAlpha) =&-\frac{N}{\bbt} \log 2+  \frac{1}{2\bbt} \sum_{i=1}^N \left[(1-m_i) \log(1-m_i) + (1+m_i) \log (1+m_i)\right]\notag \\
    &+\frac{1}{2} \sum_{\mu=1}^K  \langle z_\mu \rangle^2 -  \dfrac{\NotAlpha}{\sqrt{N^{P-1}}}\SOMMA{\mu=1}{K}\langle z_\mu \rangle \left(\dfrac{1}{r M}\SOMMA{i=1}{N}\SOMMA{a=1}{M}\eta_i^{\mu,a}m_i\right)^{P/2},
\end{align}
which allows us to compute the self-consistence equations for $\NotAlpha=1$ as
\begin{align}
    m_i &= \tanh{\left\{\bbt \dfrac{P}{2}\dfrac{1}{\sqrt{N^{P-1}}}\SOMMA{\mu=1}{K}\left[\langle z_{\mu} \rangle\dfrac{1}{r M}\SOMMA{a=1}{M} \eta_i^{\mu,a} \left(\dfrac{1}{r M}\SOMMA{b=1}{M}\SOMMA{j=1}{N} \eta_j^{\mu,b} m_j \right)^{P/2-1}\right]\right\}},
    \label{eq:P_a_1} \\
     \langle z_{\mu} \rangle &=\frac{1}{\sqrt{  N^{P-1}}} \left(\dfrac{1}{r M}\SOMMA{a=1}{M}\SOMMA{i=1}{N} \eta_i^{\mu,a} m_i\right)^{P/2}.
     \label{eq:P_b_1}
\end{align}
In order to lighten the notation, we introduce the average of the examples reads as
\begin{equation}
    (\hat{\eta}_M)^{\mu}_i=\dfrac{1}{rM}\SOMMA{a=1}{M} \eta^{\mu,a}_i\,,
\end{equation}
thus we can rewrite Eq. \eqref{eq:P_a_1} and \eqref{eq:P_b_1} as 
\begin{align}
    m_i=\l \sigma_i \r&= \tanh{\left\{\bbt \dfrac{P}{2}\dfrac{1}{\sqrt{N^{P-1}}}\SOMMA{\mu=1}{K}\left[\langle z_{\mu} \rangle(\hat{\eta}_M)_i^{\mu} \left(\dfrac{1}{r}\SOMMA{j=1}{N} (\hat{\eta}_M)_j^{\mu}\, m_j \right)^{P/2-1}\right]\right\}},
    \label{eq:P_a}
\\
     \langle z_{\mu} \rangle &=\frac{1}{\sqrt{  N^{P-1}}} \left(\SOMMA{i=1}{N} (\hat{\eta}_M)_i^{\mu} \,m_i\right)^{P/2}.
     \label{eq:P_b}
\end{align}

These equations are then used ``in tandem'' to make the system evolve: starting from an initial configuration $\boldsymbol{\sigma}^{(0)}$, we evaluate the related $z_{\mu}^{(0)}$ for any $\mu$ and we use them in Eq. \eqref{eq:P_a} to get $m_i^{(0)}$ for any $i$, the latter is then used in Eq. \eqref{eq:P_b} to get $z_{\mu}^{(1)}$, and we proceed this way, bouncing from \eqref{eq:P_a} to \eqref{eq:P_b}, up to thermalization.

We stress that these equations allow to implement an effective dynamics that avoid the computation of spin configurations. In fact, this coarse grained dynamics only cares about the Boltzmann average of each spin direction, whose behaviour is given by Eq. \eqref{eq:P_a}.
Even if at first glance Eq. \eqref{Halpha} seems to require three set of auxiliary variables, 
$\{z_{\mu} \}$ , $\{h_i\}$ and $\{\tilde{h}_i\}$,  the extremization of the Gibbs potential at first order fixes the external fields $\{h_i\}$ and $\{\tilde{h}\}$. The Gaussian variables $\{z_{\mu} \}$ act  as latent dynamical variables that evolve according to \eqref{eq:P_b}. In such an iterative MC scheme these hidden degrees of freedom are suitably updated in order to effectively retrieve the pattern that constitutes the signal.

% As corroborating our results, we plot in Fig. \ref{fig:s2n_unsup} the results with MC simulations with Plefka's dynamics and signal to noise analysis in terms of magnetization $m$ and susceptibility $\partial_r m$. We show that they are perfectly aligned. Furthermore, in Fig. \ref{fig:sup_attractor} we plot the attractors for a random pattern using MC simulations with Plefka's dynamics. We highlight that, as $\rho$ increases, the attractors shrinks. 

\newpage

\section*{Table of symbols (in alphabetical order)}
\begin{itemize}
    \item $\mathcal{A}$ is the statistical quenched pressure
     \item $\alpha_{b}$ is the storage of the network defined as $\alpha_P=\lim_{N\rightarrow +\infty} K/N^b$
    \item $\mathcal{B}(\bm \sigma; t)$ is the Boltzmann factor, defined as $\mathcal{B}(\bm \sigma; t)=\exp[-\beta {H}(\bm \sigma; t)]$
    \item $\beta$  is a non negative constant called (fast) noise such that for $\beta \to 0$ the behavior of the network is a pure random walk while for $\beta \to \infty$ it is a steepest descent toward the minima
    \item $\tilde \beta $ is defined as $\dfrac{2\beta}{P!}(1+\rho)^{-P/2}$ 
    \item $\mathbb{E}$ denotes the average over all the variables taken in examination
    \item  $\bm \eta \in \{ -1, +1 \}^{K\times M}$ are the noisy examples, namely a noisy version of the archetypes % where $\bm \eta^{\mu,a} = \bm \xi^{\mu} \bm \chi^{\mu,a}$
    \item $\gamma$ is the $P$ independent part of $\alpha_{b}$, namely $\alpha_{b}= \gamma \frac{2}{P!}$
    \item $\mathcal H$ is the Hamiltonian of the model
    \item $K \in \mathbb{N}$ is the amount of archetypes $\bm \xi$ to learn and retrieve {, in the  maximum storage limit $K=\gamma\dfrac{2}{P!}N^{P-1}$}
    \item $N \in \mathbb{N}$ is the amount of spins of the network, also known as network size
    \item $M \in \mathbb{N}$ is the amount of examples per archetype
    \item $m_\mu$ is the Mattis magnetization w.r.t. the archetype $\xi^\mu$ defined as $\frac{1}{N} \sum_{i=1}^N \xi_i^{\mu} \sigma_i$. We always take into consideration $m_1$. In RS assumption the thermodynamic limit is $\bar{m}$
    \item $n_{\mu}$ is the magnetization w.r.t. the example $\eta^{\mu,a}$ defined as $\frac{r}{\mathcal{R}NM} \sum_{i=1}^N\sum_{a=1}^{M} \eta_i^{\mu,a} \sigma_i$. In RS assumption the thermodynamic limit is $\bar{n}_\mu$
    \item $\omega_t(O(\bm \sigma))$ is the generalized Boltzmann measure, namely $\omega_t(O(\bm \sigma))= \frac{1}{\mathcal{Z}_N} \sum_{\bm \sigma} O(\bm \sigma) \mathcal{B}(\bm \sigma; t)$
    \item $P$ is the degree of interaction among spins of the network
%    \item $p$ is the probability of flipping w.r.t. the archetypes
    %\item $p_{ab}$ is the overlap among the real gaussian variables $z_{\mu}$ defined as $\frac{1}{K} \sum_{\mu=1}^K z_{\mu}^{(a)} z_{\mu}^{(b)}$. In RS assumption the thermodynamic limit is $\bar{p}$.
    \item $q_{ab}$ is the overlap among two replicas, namely two identical copies of spin configurations which share quenched noise defined as $\frac{1}{N} \sum_{i=1}^N \sigma_i^{(a)} \sigma_i^{(b)}$. In RS assumption the thermodynamic limit is $\bar{q}$
    \item $\mathcal{R}$ is defined as $\mathcal{R}=r^2 + \frac{1-r^2}{M}$
    \item $r$ assesses the dataset quality; for $r\rightarrow 1$ the example matches perfectly with the archetype whereas for $r=0$ solely noise remains
    \item $\rho$, also known as dataset entropy, quantifies the ignorance of the archetypes dataset, namely $\rho=\frac{1-r^2}{r^2M}$
    \item $t$ is the parameter for Guerra's interpolation; when $t=1$ we recover the original model, whereas for $t=0$ we compute the one-body terms
 %   \item $\V$ is the right variance to be set a posteriori using the signal to noise analysis, namely $\V=\sqrt{\dfrac{P(2P-3)!!}{2}}$
    \item $\mathcal{Z}$ is the partition function
    \item $\langle O(\bm \sigma) \rangle$ is the generalize average defined as $\langle O(\bm \sigma) \rangle = \mathbb{E} \omega_t(O(\bm \sigma))$
\end{itemize}

\section*{Acknowledgments}

E.A. acknowledges ﬁnancial support from Sapienza University of Rome (RM120172B8066CB0) and from PNRR MUR project n. PE0000013-FAIR.
\newline
A.B. acknowledges ﬁnancial support from Ministero degli Affari Esteri e della Cooperazione Internazionale Italy-Israel (F85F21006230001) and PRIN grant {\em Statistical Mechanics of Learning Machines: from algorithmic and information-theoretical limits to new biologically inspired paradigms} n.
20229T9EAT.
\newline
L.A. acknowledges financial support from INdAM –GNFM Project (CUP E53C22001930001) and from PRIN grant Stochastic Methods for Complex Systems n. 2017JFFHS
\newline
D.L. acknowledges INdAM and C.N.R. (National Research Council),  and A.A. acknowledges  UniSalento, both for financial  support via PhD-AI. 
\newline
E.A., L.A., F.A., A.A., A.B acknowledge the stimulating research environment provided by the Alan Turing Institute's Theory and Methods Challenge Fortnights event ``Physics-informed Machine Learning".

\bibliographystyle{abbrv}
\bibliography{Supervised_AB_V5_AllBlack.bib}

\end{document}